\documentclass[12pt]{article}
\usepackage[utf8]{inputenc}
\usepackage{mathtools}
\usepackage[margin=1.0in]{geometry}
\usepackage{multirow,multicol}
\usepackage{amsfonts}
\usepackage{amsmath,amssymb,amsthm}
\usepackage[mathscr]{euscript}
\usepackage{bbm}
\usepackage{graphicx,color}
\usepackage{pifont} 
\usepackage{textcomp}
\usepackage{booktabs}
\usepackage{caption}
\usepackage{tabularx}
\usepackage{xr}
\usepackage{array} 
\usepackage{titlesec}
\usepackage{algorithm}
\usepackage{algorithmic}
\usepackage{float}
\usepackage{subfig}
\usepackage[english]{babel}
\usepackage[toc,page]{appendix}
\usepackage{natbib}
\usepackage{titletoc}

\newtheorem{assumption}{Assumption}

\newtheorem{definition}{Definition}
\newtheorem{lemma}{Lemma}

\newtheorem{theorem}{Theorem}

\newtheorem{remark}{Remark}

\newcommand{\bm}{\boldsymbol}
\newcommand{\cm}[1]{\mbox{\boldmath$\mathscr{#1}$}}

\newcommand{\Fr}{{\mathrm{F}}}
\newcommand{\op}{{\mathrm{op}}}

\definecolor{turquoise}{rgb}{0.03, 0.7, 0.87}

\newcommand{\defeq}{\vcentcolon=}
\newcommand{\norm}[1]{\left\lVert#1\right\rVert}
\DeclareMathOperator*{\vectorize}{vec}

\DeclareMathOperator*{\argmin}{arg\,min}
\DeclareMathAlphabet\mathrsfso{U}{rsfso}{m}{n}

\newcommand{\vertiii}[1]{{\left\vert\kern-0.25ex\left\vert\kern-0.25ex\left\vert #1 
		\right\vert\kern-0.25ex\right\vert\kern-0.25ex\right\vert}}

\makeatletter
\newcommand*{\addFileDependency}[1]{
\typeout{(#1)}
\@addtofilelist{#1}
\IfFileExists{#1}{}{\typeout{No file #1.}}
}
\makeatother

\newcommand*{\myexternaldocument}[1]{%
\externaldocument{#1}%
\addFileDependency{#1.tex}%
\addFileDependency{#1.aux}%
}

\myexternaldocument{supplement}

\numberwithin{equation}{section}
\allowdisplaybreaks
\usepackage[colorlinks,linkcolor=blue,citecolor=blue]{hyperref}
\usepackage{titling}
\usepackage{authblk} 

\title{An Efficient and Interpretable Autoregressive Model for High-Dimensional Tensor-Valued Time Series}

\author{
\centering
Yuxi Cai, Lan Li, Yize Wang, and Guodong Li \\
\textit{Department of Statistics and Actuarial Science, University of Hong Kong}
}

\begin{document}

\setlength{\droptitle}{-4.5em}
\maketitle
\begin{abstract}
	In autoregressive modeling for tensor-valued time series, Tucker decomposition, when applied to the coefficient tensor, provides a clear interpretation of supervised factor modeling but loses its efficiency rapidly with increasing tensor order. Conversely, canonical polyadic (CP) decomposition maintains efficiency but lacks a precise statistical interpretation.
	To attain both interpretability and powerful dimension reduction, this paper proposes a novel approach under the supervised factor modeling paradigm, which first uses CP decomposition to extract response and covariate features separately and then regresses response features on covariate ones. This leads to a new CP-based low-rank structure for the coefficient tensor. 
	Furthermore, to address heterogeneous signals or potential model misspecifications arising from stringent low-rank assumptions, a low-rank plus sparse model is introduced by incorporating an additional sparse coefficient tensor. 
	Nonasymptotic properties are established for the ordinary least squares estimators, and an alternating least squares algorithm is introduced for optimization. Theoretical properties of the proposed methodology are validated by simulation studies, and its enhanced prediction performance and interpretability are demonstrated by the  El Ni$\tilde{\text{n}}$o-Southern Oscillation example.  
	
\end{abstract}

\newpage
\section{Introduction}
With the rapid advancement of information technology, multi-dimensional or tensor-valued data are now ubiquitous, and particularly, many of them are time-dependent. Examples can be found in various fields, such as neuroscience \citep{zhou2013tensor, li2021tensor}, digital marketing \citep{bi2018multilayer, hao2021sparse}, economics and finance \citep{chen2022factor,wang2024high} and climatology \citep{bahadori2014fast, xu2018muscat}. 
It is arguably the most important task to conduct supervised learning on these data, say tensor autoregression \citep{li2021multi, wang2022high} or more general tensor-on-tensor regression \citep{lock2018tensor,raskutti2019convex}. However, the corresponding responses and covariates may both be tensors with high order and high dimension. This ends up with scenarios where the number of parameters to estimate is substantially larger than the available sample size.
Meanwhile, these data often possess hidden low-dimensional structures that can be leveraged to improve model interpretability and estimation efficiency at the same time.

As an illustrative example, a few oceanic variables from specific regions are found to be closely associated with the El Ni$\tilde{\text{n}}$o-Southern Oscillation (ENSO), the most prominent phenomenon of contemporary natural climate variability \citep{sarachik2010the,zhou2023a}. 
The data, which include seven-layer ocean temperature anomalies in the upper 150 meters, are observed over time, forming a three-way tensor-valued time series with dimensions 51 (\textit{longitude}) $\times$ 41 (\textit{latitude}) $\times$ 7 (\textit{variables}). 
The high dimensionality of the data makes its direct modeling and forecasting formidable since even a simple autoregressive model of order one would involve $(51 \times 41 \times 7)^2 = 14637^2$ parameters, an overwhelming number that far exceeds the observed time points. On the other hand, research indicates that climate conditions of the so-called ``El Ni$\tilde{\text{n}}$o basin" region affect, and are affected by, those of other regions in different ways, especially before the onset of an El Ni$\tilde{\text{n}}$o event \citep{josef2013improved}. The El Ni$\tilde{\text{n}}$o basin is approximately around the equatorial Pacific corridor; see Figure \ref{fig:area} for an illustration. Yet, the literature has not reached a consensus on its precise location.
It is crucial for a successful inference tool to correctly identify the low-dimensional structure, and this can be achieved by low-rankness-induced approaches in a coherent way \citep{zhou2013tensor,feng2021brain}.

\begin{figure}[t!]
	\centering
	\includegraphics[width=0.8\textwidth]{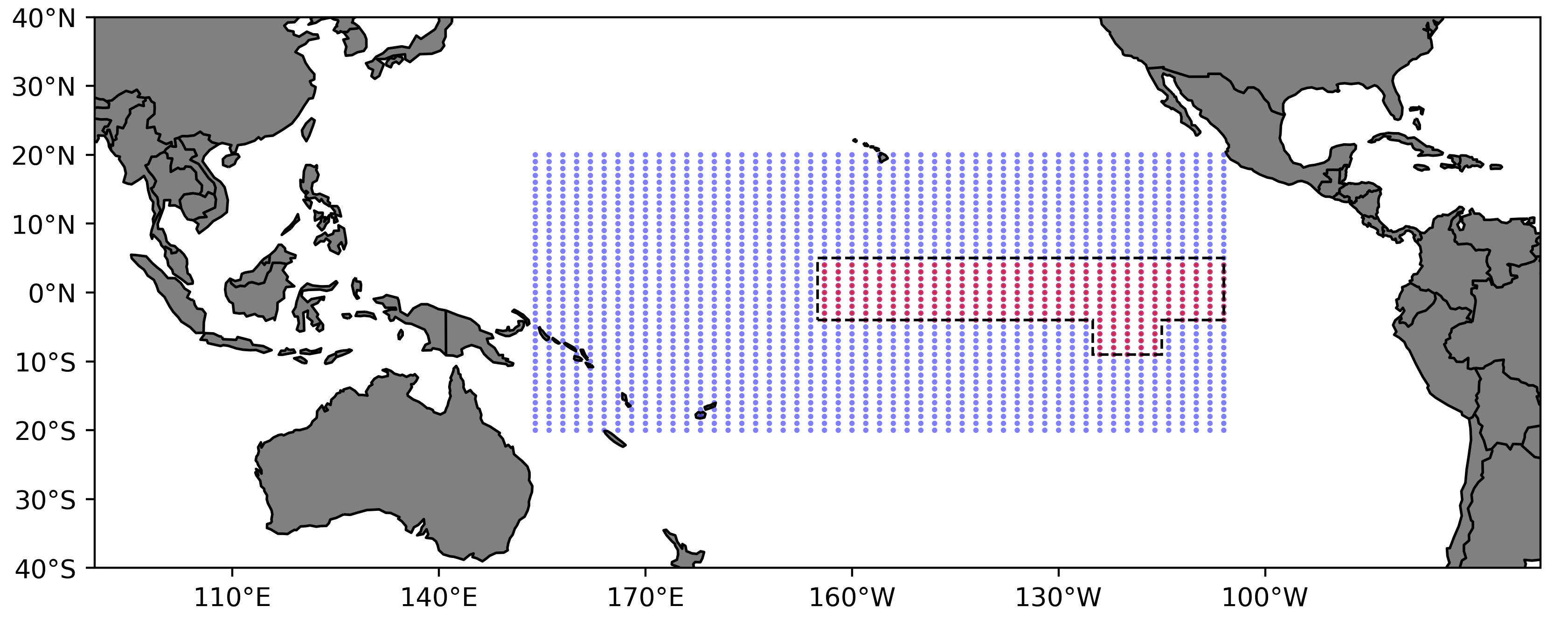}
	\caption{The grid points represent the geographical locations where ocean temperature data are being considered. One possible boundary of the El Ni$\tilde{\text{n}}$o basin is delineated by dashed lines, with grid points within this region highlighted in red.}
	\label{fig:area}
\end{figure}

There are two extensively used methods in the literature for adding low-rank constraints to tensors: Tucker and canonical polyadic (CP) decomposition; see \cite{kolda2009tensor}.
Tucker decomposition \citep{tucker1966some} has been widely applied to various tensor regression settings, including scalar-on-tensor and tensor-on-tensor models, due to its numerical stability \citep{yu2016learn, li2017parsimonious, li2018tucker, gahrooei2021multiple, han2022an}.
\cite{wang2024high} considered Tucker decomposition to tensor autoregressive time series models, leading to a nice interpretation of supervised factor modeling; see also \cite{wang2022high} and \cite{huang2023supervised}.
Specifically, they first assume coefficient tensors to have a high-order singular value decomposition (HOSVD), a special Tucker decomposition. Some loading matrices can then be moved to the left-hand side of autoregression or regression to summarize responses into low-dimensional factors or features, and the others will summarize covariates into another few features.
The response features capture all the predictable components, while the covariate ones encapsulate all the driving forces, thereby facilitating model explanations.
The above physical interpretation can help us study the ENSO example. However, Tucker decomposition will lose its efficiency in compressing tensors very quickly as the order of tensors increases, say greater than four or five \citep{oseledets2011tensor,cai2022provable}.
This issue is particularly prominent when dealing with autoregression for tensor-valued time series, where the coefficient tensor has a minimum order of six. 

In the meanwhile, CP decomposition \citep{harshman1970foundation} has been well recognized for its efficiency in dimension reduction \citep{kolda2009tensor}. 
\cite{guo2012tensor} and \cite{zhou2013tensor} considered CP decomposition for scalar-on-tensor regression problems, followed by many others \citep{guhaniyogi2017bayesian, zhang2019tensor, hao2021sparse, feng2021brain}. A nice interpretation can be achieved since CP decomposition can extract principal components or features from covariates \citep{cai2022provable}.
For instance, in the application to neuroimaging analysis, the model can identify the relevant brain regions to a clinical response \citep{zhou2013tensor}.
Besides, CP decomposition has also been employed to effectively reduce the parameter dimension of tensor-on-tensor regression models \citep{lock2018tensor,llosa2023reduced, billio2024bayesian}, however, its statistical interpretation becomes vague. 
In contrast to Tucker decomposition, the loading matrices of CP decomposition are generally not orthogonal \citep{han2024cp}. This makes it difficult to move some of them to the left-hand side of regression to extract features from tensor-valued responses. 
Consequently, the resulting model lacks the interpretation of feature extraction from both responses and covariates. The limitation of CP decomposition in the settings of tensor-on-tensor regression, and hence tensor autoregression, suggests a pressing need to develop a suitable model that is not only more interpretable but can also dramatically compress parameter dimension.

Beyond Tucker and CP decompositions, several works have explored alternative low-rank structures for coefficients in tensor autoregression \citep{li2021multi} and tensor-on-tensor regression \citep{llosa2023reduced, si2024efficient}. Despite these efforts, none have successfully fulfilled the objective of attaining both interpretability and powerful dimension reduction. The first contribution of this paper is to propose a novel CP-based low-rank structure to the coefficient tensor in tensor-on-tensor regression models and adapt the idea to autoregressive models for tensor-valued time series. The design leads to a supervised factor modeling interpretation,  i.e., first use CP decomposition to extract response and covariate features separately and then regress response features on covariate ones. It also inherits the parsimonious structure of CP decomposition, and hence both the interpretability and estimation efficiency can be achieved.

Moreover, in the literature on high-dimensional time series, another widely used method is to assume that coefficients are sparse and then apply sparsity-inducing regularized estimation; see, e.g., $l_1$-regularization for vector autoregressive models \citep{basu2015regularized,davis2016sparse} and vector autoregressive moving average models \citep{Wilms2023,Zheng2024}.
There is even a much bigger literature on sparsity-inducing regularization methods for regression models \citep{tibshirani1996lasso, fan2001scad, davis2016sparse, raskutti2019convex, bertsimas2020sparse}. 
On the one hand, despite its powerful performance in prediction, this type of methods may fail to represent underlying complex structures in data \citep{basu2019low}. 
For the ENSO example, the oceanic conditions of the entire region are likely to change in a concerted manner, driven by climate conditions of the El Ni$\tilde{\text{n}}$o basin.
It is hard to identify this influential region by solely relying on sparsity regularization, whereas the low-rankness may capture this driving force in a more coherent way \citep{zhou2013tensor,feng2021brain}.
On the other hand, an exact low-rank assumption can be limited, leaving the method susceptible to model misspecification or heterogeneous signals \citep{agarwal2012noisy, basu2019low, cai2023generalized, bai2023multiple}. 
It hence may be more reasonable to assume that the coefficient tensor has a low-rank plus sparse structure. This leads to the second contribution of this paper: we designate a coefficient tensor with CP-based low-dimensional structure to model the associations that are collaboratively driven by latent features, while using a sparse one to account for heterogeneous signals or potential model misspecifications arising from an exact low-rank assumption.


These two main contributions are detailed in Sections \ref{subsec:method-idea} and \ref{subsec:method-ar}. Besides, 
this paper also proposes the corresponding high-dimensional estimation in Section \ref{subsec:HDestimate}, and adopts the alternating least squares (ALS) algorithm \citep{de1976additive,  zhou2013tensor} to search for the estimates, which casts the update of a single block at each iteration into a traditional linear regression problem; see Section \ref{subsec:method-algorithm} for details. 
Theoretically, the non-asymptotic error bounds are given in Section \ref{sec:theory} for high-dimensional autoregression as well as regression settings. Section \ref{sec:simulation} conducts the simulation experiments to assess the finite-sample performance of the proposed models, and their effectiveness is further validated by the ENSO example in Section \ref{sec:real-data}. Finally, Section \ref{sec:conclusion} provides a concise summary and discussion to conclude the paper, and all technical proofs are deferred to the Supplementary Material.

Throughout the paper, we denote vectors by boldface small letters, e.g., $\bm{x}$; matrices by boldface capital letters, e.g., $\bm{X}$; tensors by boldface Euler capital letters, e.g., $\cm{X}$.  
For a $d$-th order tensor $\cm{X} \in \mathbb{R}^{p_1 \times p_2 \times \cdots \times p_d}$, 
denote by $\vectorize(\cm{X})$ its vectorization, and further denote by $\cm{X}_{(s)}$ and $[\cm{X}]_s$ the mode matricization and sequential matricization at mode $s$ with $1 \leq s\leq d$, respectively. The Frobenius norm of $\cm{X}$ is defined as $\norm{\cm{X}}_{\mathrm{F}}=\sqrt{\langle\cm{X},\cm{X}\rangle}$, where $\langle\cdot,\cdot\rangle$ is the inner product;
and the $l_\infty$ and $l_1$-norms are defined as $\norm{\cm{X}}_\infty = \max_{i_1, i_2, \dots, i_d} \lvert \cm{X}_{i_1, i_2, \dots, i_d}\rvert$ and $\norm{\cm{X}}_1 = \sum_{i_1, i_2, \dots, i_d} \mathbbm{1}\{\cm{X}_{i_1, i_2, \dots, i_d}\ \neq 0\}$, respectively.
For a generic matrix $\bm{X}$, we denote by $\bm{X}^\top$, $\norm{\bm{X}}_{\mathrm{F}}$, $\norm{\bm{X}}_2$, and $\vectorize(\bm{X})$ its transpose, Frobenius norm, spectral norm, and a long vector obtained by stacking all its columns, respectively. If $\bm{X}$ is further a square matrix, then denote its minimum and maximum eigenvalue by $\lambda_{\min}(\bm{X})$ and $\lambda_{\max}(\bm{X})$, respectively.
For two real-valued sequences $x_k$ and $y_k$, $x_k\gtrsim y_k$ if there exists a $C>0$ such that $x_k\geq Cy_k$ for all $k$. In addition, we write $x_k\asymp y_k$ if $x_k\gtrsim y_k$ and $y_k\gtrsim x_k$.

\section{Methodology}
\subsection{CP-based feature extraction in tensor-on-tensor regression}
\label{subsec:method-idea}
This subsection first illustrates how to construct a tensor-on-tensor regression model with CP decomposition used to extract features from responses and covariates separately, forming the basis for an interpretable and efficient tensor autoregressive model in the next subsection. 

We start by considering a scalar-on-tensor regression model with the coefficient tensor having a CP decomposition \citep{guo2012tensor, zhou2013tensor, guhaniyogi2017bayesian, feng2021brain},
\begin{equation*}
	\label{model:scalar-on-tensor-cp}
	y_t = \left\langle  \cm{A}, \cm{X}_t \right\rangle + \varepsilon_t \hspace{3mm}\text{with}\hspace{3mm} \cm{A}=\sum_{r=1}^{R_x} g_r \cdot \bm{v}_r^{(1)} \circ \bm{v}_r^{(2)} \circ \cdots \circ \bm{v}_r^{(m)},
\end{equation*}
where $y_t,\varepsilon_t\in\mathbb{R}$ are the response and the error term, respectively, $\cm{X}_t \in \mathbb{R}^{p_1 \times p_2 \times \cdots \times p_m}$ is the covariate, $\cm{A} \in \mathbb{R}^{p_1 \times p_2 \times \cdots \times p_m}$ is the coefficient tensor, $R_x$ is the CP-rank, $\bm{v}_r^{(j)} \in \mathbb{R}^{p_j}$ with $\lVert \bm{v}^{(j)}_r\rVert_2 = 1$ for each $1 \leq r \leq R_x$ and $1 \leq j \leq m$, and $\circ$ denotes the outer product.
The model can be rewritten into
\begin{equation}
	\label{model:cp-covariate-factor}
	y_t = \sum_{r=1}^{R_x} g_r f^{\mathrm{covariate}}_{r, t} + \varepsilon_t
	\quad \text{with} \quad
	f^{\mathrm{covariate}}_{r, t} = \left\langle \bm{v}_r^{(1)} \circ \bm{v}_r^{(2)} \circ \cdots \circ \bm{v}_r^{(m)}, \cm{X}_t  \right\rangle,
\end{equation}
where $\bm{f}^{\mathrm{covariate}}_t = (f^{\mathrm{covariate}}_{1, t}, \dots,  f^{\mathrm{covariate}}_{R_x, t})^\top \in\mathbb{R}^{R_x}$ are features extracted from covariates $\cm{X}_t$, and
the loading tensors $\bm{v}_r^{(1)} \circ \bm{v}_r^{(2)} \circ \cdots \circ \bm{v}_r^{(m)}$ are of the same shape as $\cm{X}_t$ but may not be orthogonal with each other \citep{kolda2009tensor}.
From model \eqref{model:cp-covariate-factor}, the scalar response is modeled as a linear combination of covariate features, which in turn makes the feature extraction a supervised process. 
These features can provide a reasonable approximation to many low-dimensional signals \citep{zhou2013tensor}.
For example, when $y_t$ is a clinical response and $\cm{X}_t$ is a brain image, these CP-structured loadings can help recover regions of interest in the brain related to a disease. In the ENSO example, features may summarize the driving forces of climate change, while loadings can help identify the location of El Ni$\tilde{\text{n}}$o basin.

However, for time series models, both responses and covariates are observed tensors from the same variable, and hence it is necessary to consider a more general tensor-on-tensor regression model,
\begin{equation}
	\label{model:tensor-regression}
	\cm{Y}_t = \left\langle \cm{A}, \cm{X}_t \right\rangle + \cm{E}_t,
\end{equation}
where $\cm{Y}_t, \cm{E}_t \in \mathbb{R}^{q_1 \times q_2 \times \cdots \times q_n}$ are the response and the error term, respectively, and the covariate $\cm{X}_t$ has the same shape as in model \eqref{model:cp-covariate-factor}. The problem caused by high dimensionality is more pronounced in this model as now the coefficient tensor $\cm{A}\in\mathbb{R}^{q_1 \times \cdots \times q_n \times p_1 \times \cdots \times p_m}$. 
While, the efficiency and interpretability brought by covariate feature extraction in \eqref{model:cp-covariate-factor} motivates us to consider feature extraction from responses as well.
Specifically, we may extract $R_y$ features, denoted by $\bm{f}_t^{\mathrm{response}} = (f_{1,t}^{\mathrm{response}}, \dots, f_{R_y,t}^{\mathrm{response}})^\top \in \mathbb{R}^{R_y}$, from $\cm{Y}_t$, with each feature of the form:
\begin{equation}
	\label{eq:response-factor}
	f_{r, t}^{\mathrm{response}} = \left\langle 
	\bm{u}_r^{(1)} \circ \bm{u}_r^{(2)} \circ \cdots \circ \bm{u}_r^{(n)}, \cm{Y}_t 
	\right\rangle
	\quad \text{where } 
	\bm{u}^{(i)}_r \in \mathbb{R}^{q_i}  \text{ and }\lVert \bm{u}^{(i)}_r\rVert_2 = 1 
\end{equation}
for all $1 \leq r \leq R_y$ and $1 \leq i \leq n$.
The loading tensors $\cm{C}_r = \bm{u}_r^{(1)} \circ \bm{u}_r^{(2)} \circ \cdots \circ \bm{u}_r^{(n)}$ are not necessarily orthogonal to each other, but we require them to be non-degenerate, i.e., $\sum_{r=1}^{R_y} c_r \cm{C}_r = 0$ must imply that $c_1 = \cdots = c_{R_y} = 0$. 
Let $\bm{\Lambda}_y = \bm{U}_n \odot \cdots \odot \bm{U}_1 \in \mathbb{R}^{\prod_{i=1}^{n}q_i \times R_y}$
and
$\bm{\Lambda}_x = \bm{V}_m \odot \cdots \odot \bm{V}_1 \in \mathbb{R}^{\prod_{j=1}^{m}p_j \times R_x}$,
where $\odot$ denotes the Khatri-Rao product, $\bm{U}_i = \left(\bm{u}^{(i)}_1, \dots, \bm{u}^{(i)}_{R_y}\right) \in \mathbb{R}^{q_i \times R_y}$ for $1\leq i\leq n$, and $\bm{V}_j = \left(\bm{v}^{(j)}_1, \dots, \bm{v}^{(j)}_{R_x}\right) \in \mathbb{R}^{p_j \times R_x}$ for $1\leq j\leq m$.
It can be verified that 
$\bm{f}_t^{\mathrm{response}} = \bm{\Lambda}_y^\top \vectorize(\cm{Y}_t)$ and $\bm{f}^{\mathrm{covariate}}_t=\bm{\Lambda}_x^\top \vectorize(\cm{X}_t) $.

We next attempt to design a low-dimensional structure to model \eqref{model:tensor-regression} such that both response and covariate feature extraction are supervised by the predictive relationship between their resulting features. 
We propose to assume that the sequential matricization of $\cm{A}$ has the following decomposition:
\begin{equation}
	\label{eq:A-form}
	[\cm{A}]_n = \bm{\Lambda}_y \left(\bm{\Lambda}_y^\top \bm{\Lambda}_y \right)^{-1} \bm{G} \bm{\Lambda}_x^\top,
\end{equation}
and then model \eqref{model:tensor-regression} can be reformulated into
\begin{equation}
	\label{model:our-cp}
	\vectorize(\cm{Y}_t) = \bm{\Lambda}_y \left(\bm{\Lambda}_y^\top \bm{\Lambda}_y \right)^{-1} \bm{G} \bm{\Lambda}_x^\top \vectorize(\cm{X}_t) + \vectorize(\cm{E}_t) \hspace{3mm}\text{or}\hspace{3mm}
	\bm{f}_t^{\mathrm{response}}=\bm{G}\bm{f}^{\mathrm{covariate}}_t+\bm{e}_t^{\mathrm{CP}},
\end{equation}
where $\bm{G}\in\mathbb{R}^{R_y\times R_x}$ is the coefficient matrix, and $\bm{e}_t^{\mathrm{CP}}=\bm{\Lambda}_y^\top\vectorize(\cm{E}_t)$. 
In other words, it admits a regression form of response features
$\bm{f}_t^{\mathrm{response}}$ on covariate ones $\bm{f}^{\mathrm{covariate}}_t$, where $\bm{f}_t^{\mathrm{response}}$ summarizes all the predictable components and $\bm{f}^{\mathrm{covariate}}_t$ captures all the driving forces; see a similar supervised factor modeling argument in \cite{huang2023supervised} and \cite{wang2024high}. Besides, model \eqref{model:our-cp} also inherits the parsimonious structure of CP decomposition: it has a complexity of $d_\mathrm{LR} = R_yR_x + R_y\sum_{i=1}^nq_i + R_x\sum_{j=1}^mp_j$, which only grows linearly with the tensor order and thereby maintaining estimation efficiency even when $n$ and $m$ are large.

\begin{remark}
	\label{remark:tucker}
	In many works of tensor-on-tensor regression, the coefficient tensor $\cm{A}$ is assumed to have multilinear low ranks, i.e., $\mathrm{rank}\left(\cm{A}_{(i)}\right)\leq r_i$ for each $1\leq i \leq n+m$ \citep{yu2016learn, gahrooei2021multiple, han2022an, wang2022high}. We can then have a Tucker decomposition of $\cm{A} = \cm{G} \times_{i=1}^{n+m} \bm{U}_{i}$, where $\bm{U}_i \in \mathbb{R}^{q_i \times r_i}$ and $\bm{U}_{n+j} \in \mathbb{R}^{p_{j} \times r_{n+j}}$ are standardized to be orthonormal loading matrices for $1\leq i \leq n$, $1\leq j \leq m$, and $\cm{G} \in \mathbb{R}^{r_1 \times \cdots \times r_{n+m}}$ is the core tensor. Putting this decomposition back into model \eqref{model:tensor-regression}, we have 
	\begin{equation*}
		\cm{Y}_t = \left\langle \cm{G} \times_{i=1}^{n+m} \bm{U}_i, \cm{X}_t \right\rangle + \cm{E}_t \hspace{5mm}\text{or}\hspace{5mm}
		\cm{Y}_t \times_{i=1}^{n}  \bm{U}_i^\top =\left\langle \cm{G}, \cm{X}_t  \times_{j=1}^{m} \bm{U}_{n+j}^\top \right\rangle + \bm{e}_t^{\mathrm{Tucker}},
	\end{equation*}
	where $\bm{e}_t^{\mathrm{Tucker}}=\cm{E}_t \times_{i=1}^{n}  \bm{U}_i^\top$.
	If regarding $\cm{Y}_t \times_{i=1}^{n}  \bm{U}_i^\top$ and $\cm{X}_t  \times_{j=1}^{m} \bm{U}_{n+j}^\top $ as the response and covariate features, respectively, the above model then offers an interpretation of supervised factor modeling, which directly assuming $\cm{A}$ to have a CP decomposition does not provide.
	However, the Tucker-based method will lose its efficiency very quickly as the order of coefficient tensors increases, say greater than four or five \citep{oseledets2011tensor}. It has a complexity of $d_\mathrm{Tucker} = \prod_{i=1}^{n+m} r_i + \sum_{i=1}^{n}q_ir_i + \sum_{j=1}^{m}p_{j}r_{n+j}$ that is exponential in tensor orders.
	For model \eqref{model:tensor-regression} with the simplest tensor-valued observations, i.e., $n=m=3$, the coefficient tensor $\cm{A}$ already has an order of six, where Tucker decomposition is doomed to be inefficient.
\end{remark}

\subsection{Efficient and interpretable autoregressive time series model}
\label{subsec:method-ar}
This subsection adopts ideas in Section \ref{subsec:method-idea} to the autoregressive (AR) time series model,
\begin{align}
	\label{model:ARp}
	\cm{Y}_t = \left\langle \cm{A}_1, \cm{Y}_{t-1} \right\rangle + \left\langle \cm{A}_2, \cm{Y}_{t-2} \right\rangle + \cdots +
	\left\langle \cm{A}_P, \cm{Y}_{t-P} \right\rangle + \cm{E}_t, 
\end{align}
where $\cm{Y}_t \in \mathbb{R}^{q_1 \times q_2 \times \cdots \times q_n}$ is the observed tensor, $P$ is the order, $\cm{A}_k \in \mathbb{R}^{q_1 \times \cdots \times q_n \times q_1 \times \cdots \times q_n}$ with $1\leq k\leq P$ are coefficient tensors, and $\cm{E}_t \in \mathbb{R}^{q_1 \times q_2 \times \cdots \times q_n}$ is the error term.

We first consider the case with $P=1$, and the coefficient tensor $\cm{A}_1$ can be directly assumed to have the decomposition in \eqref{eq:A-form} by letting $m =n$ and $p_j = q_j$ for $1\leq j\leq n$. It then leads to an AR model, $\vectorize(\cm{Y}_t) = \bm{\Lambda}_y \left(\bm{\Lambda}_y^\top \bm{\Lambda}_y \right)^{-1} \bm{G}_1 \bm{\Lambda}_x^\top
\vectorize(\cm{Y}_{t-1}) + \vectorize(\cm{E}_t)$, and its interpretable form:
\begin{equation}
	\label{model:ar1}
	\underbrace{\bm{\Lambda}_y^\top \vectorize(\cm{Y}_t)}_{\bm{f}_t^{\mathrm{response}}}  
	= \bm{G}_1 
	\underbrace{\bm{\Lambda}_x^\top \vectorize(\cm{Y}_{t-1})}_{\bm{f}^{\mathrm{covariate}}_{t-1}} 
	+ \bm{\Lambda}_y^\top \vectorize(\cm{E}_t). 
\end{equation}
Here the high-order tensor $\cm{Y}_t$ is summarized into $R_y$ features, $\bm{f}_t^{\mathrm{response}}$, each of which is extracted as in \eqref{eq:response-factor}. And the information from $\cm{Y}_{t-1}$ is compressed into $R_x$ features, denoted by $\bm{f}^{\mathrm{covariate}}_{t-1}$, and each of them is obtained as in \eqref{model:cp-covariate-factor}. The coefficient matrix $\bm{G}_1 \in \mathbb{R}^{R_y \times R_x}$ characterizes the relationship between response features $\bm{f}^{\mathrm{response}}_t$ and covariate features $\bm{f}^{\mathrm{covariate}}_{t-1}$. 


For a general order $P$, we extend model \eqref{model:ar1} and consider the following form:
\begin{equation}
	\label{model:arp-cp}
	\vectorize(\cm{Y}_t) = \sum_{k=1}^{P} [\cm{A}_k]_n \vectorize(\cm{Y}_{t-k}) + \vectorize(\cm{E}_t)
	\quad \text{with} \quad
	[\cm{A}_k]_n = \bm{\Lambda}_y \left(\bm{\Lambda}_y^\top \bm{\Lambda}_y \right)^{-1} \bm{G}_k \bm{\Lambda}_x^\top,
\end{equation}
whose interpretation can be revealed by multiplying $\bm{\Lambda}_y^\top$ to both sides:
\begin{equation*}
	\bm{f}_t^{\mathrm{response}} = \sum_{k=1}^P \bm{G}_k \bm{f}^{\mathrm{covariate}}_{t-k} + \bm{\Lambda}_y^\top \vectorize(\cm{E}_t)
	\quad \text{with} \quad
	\bm{f}^{\mathrm{covariate}}_{t-k} = \bm{\Lambda}_x^\top \vectorize(\cm{Y}_{t-k}). 
\end{equation*}
For the sake of parsimony, the matrices $\bm{\Lambda}_y$ and $\bm{\Lambda}_x$ are shared across all $\cm{A}_k$'s, i.e., we apply a single operation to extract response features and use the same loading, $\bm{\Lambda}_x$, to extract features from the preceding variables $\cm{Y}_{t-k}$'s. The different matrices $\bm{G}_k \in \mathbb{R}^{R_y \times R_x}$'s allow flexible relationships between response and covariate features at different lags. Note that model \eqref{model:arp-cp} has a complexity of $d_{\mathrm{AR}} = PR_yR_x + (R_y + R_x)\sum_{i=1}^{n}q_i$, achieving both the interpretability and efficiency.

\begin{remark}\label{remark-ar2}
	We may alternatively treat the temporal dimension as an additional tensor mode, and stack covariates  $\cm{Y}_{t-1}, \dots, \cm{Y}_{t-P}$ and their coefficients $\cm{A}_{1}, \dots, \cm{A}_{P}$ into new tensors $\cm{X}_t \in \mathbb{R}^{P \times q_1 \times \cdots \times q_n}$ and $\cm{A} \in \mathbb{R}^{q_1 \times \cdots \times q_n \times P \times q_1 \times \cdots \times q_n}$, respectively. The resulting model is exactly the same as in \eqref{model:tensor-regression}, and therefore, our low-dimensional structure can just be applied as in \eqref{model:our-cp} with $m =n+1$, $p_1 = P$ and $p_j = q_{j-1}$ for $2 \leq j \leq m$. This alternative model consists of $d_\mathrm{ALR}=R_yR_x + PR_x + (R_y + R_x)\sum_{i=1}^{n}q_i$ parameters, while it leads to a more complex interpretation for the temporal dynamics of time series. Recovering the AR form from model \eqref{model:our-cp}, each coefficient $\cm{A}_k$ can be expressed as
	\begin{equation*}
		[\cm{A}_k]_n = \bm{\Lambda}_y \left(\bm{\Lambda}_y^\top \bm{\Lambda}_y \right)^{-1} \bar{\bm{G}}_k \bm{\Lambda}_x^\top
		\quad \text{with} \quad
		\bar{\bm{G}}_k = \bm{G}\bm{D}_k,
	\end{equation*}
	where $\bar{\bm{G}}_k \in \mathbb{R}^{R_y \times R_x}$ is the difference compared to \eqref{model:arp-cp}, and $\bm{D}_k \in \mathbb{R}^{R_x \times R_x}$ is a diagonal matrix. The form suggests a further constraint that all $\bar{\bm{G}}_k$'s shall share the same column space. This is less intuitive, and we prefer the dimension reduction at \eqref{model:arp-cp}.
\end{remark}

\begin{remark}
	In addition to supervised models, the literature often investigates the low-rank structures of high-dimensional time series using unsupervised factor models \citep{stock2011dynamic, bai2016econometric,chen2022factor,han2024cp}, while an extra dynamic structure, such as vector AR models, has to be imposed on latent factors to enable forecasting, leading to the so-called dynamic factor model \citep{stock2011dynamic},
	\begin{equation*}
		\vectorize(\cm{Y}_t) = \bm{\Lambda}\bm{f}_t + \bm{e}_t 
		\hspace{3mm}\text{and}\hspace{3mm}
		\bm{f}_t = \sum_{k=1}^{P}\bm{B}_k\bm{f}_{t-k} + \bm{\xi}_t,
	\end{equation*}
	where $\bm{\Lambda} \in \mathbb{R}^{\prod_{i=1}^{n}q_i \times r}$ is the factor loading matrix, $\bm{f}_t \in \mathbb{R}^r$ contains the $r$ factors from $\cm{Y}_t$, $\bm{B}_k \in \mathbb{R}^{r\times r}$'s are the coefficient matrices for the latent process, $\bm{e}_t \in \mathbb{R}^{\prod_{i=1}^{n}q_i}$ and $\bm{\xi}_t \in \mathbb{R}^r$ are error terms. Compared to our model \eqref{model:arp-cp}, both the response and covariate factors here are extracted with identical loading matrices, which limits the flexibility of the forecasting model.
\end{remark}

Finally, the exact low-rank assumption can be restrictive, leaving the method susceptible to model misspecification or heterogeneous signals, while the sparsity is another prevalent way to restrict the parameter space \citep{agarwal2012noisy, basu2019low, cai2023generalized, bai2023multiple}.
It hence may be more reasonable to assume that the coefficient tensor has a low-rank plus sparse structure, which in the autoregression setting can be
\begin{equation}
	\label{model:arp-sparse}
	\cm{Y}_t =  \sum_{k=1}^{P} \left\langle \cm{A}_k, \cm{Y}_{t-k} \right\rangle + \cm{E}_t
	\quad\text{with}\quad
	\cm{A}_k = \cm{A}^\mathrm{L}_k + \cm{A}^\mathrm{S}_k.
\end{equation}
The tensor $\cm{A}^\mathrm{L}_k$ has a low-rank structure, dedicated to modeling the associations that are collaboratively driven by latent features, whereas $\cm{A}^\mathrm{S}_k$ is sparse to account for heterogeneous signals or the potential model misspecifications brought by the exact low-rank assumption. Here we may assume the sequential matricization of $\cm{A}^\mathrm{L}_k$ to have a decomposition, $[\cm{A}^\mathrm{L}_k]_n = \bm{\Lambda}_y \left(\bm{\Lambda}_y^\top \bm{\Lambda}_y \right)^{-1} \bm{G}_k \bm{\Lambda}_x^\top$, as in \eqref{model:arp-cp}, and consider the elementwise sparsity for $\cm{A}^\mathrm{S}_k$. 

\subsection{High-dimensional estimation}
\label{subsec:HDestimate}

We first consider the AR$(P)$ model at \eqref{model:arp-cp} with $[\cm{A}_k]_n = \bm{\Lambda}_y \left(\bm{\Lambda}_y^\top \bm{\Lambda}_y \right)^{-1} \bm{G}_k \bm{\Lambda}_x^\top$.
Since the matrix $\left(\bm{\Lambda}_y^\top \bm{\Lambda}_y \right)^{-1}\in\mathbb{R}^{R_y\times R_y}$ has a full rank and keeps unchanged at different lag $k$, without loss of generality, we can reparameterize the model into $\bm{G}_k:=\left(\bm{\Lambda}_y^\top \bm{\Lambda}_y \right)^{-1} \bm{G}_k$ and $[\cm{A}_k]_n := \bm{\Lambda}_y  \bm{G}_k \bm{\Lambda}_x^\top$.
Moreover, as in Remark \ref{remark-ar2}, we stack covariates  $\cm{Y}_{t-1}, \dots, \cm{Y}_{t-P}$ and coefficients $\cm{A}_{1}, \dots, \cm{A}_{P}$ into tensors $\cm{X}_t $ and $\cm{A} $, respectively, for the sake of convenience.

For an observed time series $\{\cm{Y}_1, \dots, \cm{Y}_T\}$, denote $\bm{Y} = \left(\vectorize({\cm{Y}_T}), \vectorize({\cm{Y}_{T-1}}), \dots, \vectorize({\cm{Y}_{P+1}}) \right)$ $ \in \mathbb{R}^{Q\times (T-P)}$, and
$\bm{X} = \left(\vectorize({\cm{X}_T}), \vectorize({\cm{X}_{T-1}}), \dots, \vectorize({\cm{X}_{P+1}}) \right) \in \mathbb{R}^{PQ \times (T-P)}$, where $Q=\prod_{i=1}^nq_i$.
The ordinary least squares (OLS) estimation for AR$(P)$ models has the loss function of
\begin{equation*}
	\label{eq:loss-lowrank}
	\mathcal{L}_T(\cm{A})= 
	\frac{1}{T-P} \sum_{t=P+1}^{T} \norm{\vectorize(\cm{Y}_t) - [\cm{A}]_n \vectorize(\cm{X}_t)}_2^2 
	= \frac{1}{T-P} \norm{\bm{Y} - [\cm{A}]_n \bm{X}}_\Fr^2.
\end{equation*}

We first define the parameter space of model \eqref{model:arp-cp}.
For the CP-structured loadings $\bm{\Lambda}_y$ and $\bm{\Lambda}_x$, their parameter space can be summarized into
\begin{equation*}
	\begin{split}
		\mathcal{S}\left(\underline{d}, R\right) = \{
		\bm{\Lambda} \in \mathbb{R}^{\prod_{i=1}^{n}d_i \times R}: 
		\bm{\Lambda} =\bm{U}_n \odot \cdots \odot \bm{U}_1, \;
		&\bm{U}_i\in \mathbb{R}^{d_i \times R}, \;\\
		&\lVert\bm{u}^{(i)}_r\rVert_2 = 1, \; 
		\forall 1\leq i \leq n, 1\leq r \leq R
		\},
	\end{split}
\end{equation*}
where $\underline{d} = \{d_1, \dots, d_n\}$.
For each $\bm{\Lambda} \in 	\mathcal{S}\left(\underline{d}, R\right)$, the above definition specifies its structure as the Khatri-Rao product of a series of matrices whose columns are unit-norm, ensuring that each column of $\bm{\Lambda}$ has the form of a vectorized CP rank-1 tensor, i.e., $\vectorize(\bm{u}^{(1)}_r \circ \cdots \circ \bm{u}^{(n)}_r )$. We are now ready to define the parameter space that $\cm{A}$ lies in: 
\begin{equation*}
	\label{eq:parameter-space}
	\begin{split}
		\bm{\Theta}(\underline{q}, P, R_y, R_x) = \{
		\cm{A} \in \mathbb{R}^{q_1 \times \cdots \times q_n \times P \times q_1 \times \cdots \times q_n}:
		& [\cm{A}]_n = \bm{\Lambda}_y {\bm{G}} (\bm{I}_P \otimes \bm{\Lambda}_x^\top),\\
		\bm{\Lambda}_y & \in \mathcal{S}(\underline{q}, R_y), \;
		{\bm{G}} \in \mathbb{R}^{R_y \times PR_x},\;
		\bm{\Lambda}_x \in \mathcal{S}(\underline{q}, R_x)
		\},
	\end{split}
\end{equation*}
where $\underline{q}=\{q_1, \dots, q_n\}$, $\bm{I}_P \in \mathbb{R}^{P \times P}$ is an identity matrix, ${\bm{G}} = ({\bm{G}}_1, \dots, {\bm{G}}_P)$, and $\otimes$ denotes the Kronecker product.
As a result, the OLS estimator of model \eqref{model:arp-cp} can be defined below,
\begin{equation}
	\label{eq:least-square-est}
	\widehat{\cm{A}}= \argmin \mathcal{L}_T(\cm{A}) \hspace{3mm}\text{subject to}\hspace{3mm} 
	\cm{A} \in \bm{\Theta}(\underline{q}, P, R_y, R_x).
\end{equation}

We next consider the low-rank plus sparse model at \eqref{model:arp-sparse}. Note that the coefficients $\cm{A}^\mathrm{L}_k$'s can also be rearranged into a new tensor $\cm{A}_\mathrm{L}  \in \mathbb{R}^{q_1 \times \cdots \times q_n \times P \times q_1 \times \cdots \times q_n}$ such that $[\cm{A}_\mathrm{L}]_n = ([\cm{A}^\mathrm{L}_1]_n, \dots, [\cm{A}^\mathrm{L}_P]_n)$; and similarly for $\cm{A}^\mathrm{S}_k$'s, we can get $\cm{A}_\mathrm{S}$. By our assumptions, $\cm{A}_\mathrm{L}  \in \bm{\Theta}(\underline{q}, P, R_y, R_x)$ and $\cm{A}_\mathrm{S}$ has an element-wise sparsity. 
However, there is an inherent identifiability issue in the estimation of $\cm{A}_\mathrm{L}$ and $\cm{A}_\mathrm{S}$ since when $\cm{A}_\mathrm{L}$ is also sparse or $\cm{A}_\mathrm{S}$ is also low-rank, the two components cannot be distinguished from each other. To make them identifiable, we follow \cite{agarwal2012noisy}, \cite{basu2019low} and \cite{bai2023multiple} to similarly impose a mild assumption that $\norm{\cm{A}_\mathrm{L}}_\infty \leq \alpha_\mathrm{L} / (PQ^2)$. The term $\alpha_\mathrm{L}$ is the so-called radius of nonidentifiability, bounding the ``spikeness" of $\cm{A}_\mathrm{L}$. 
We define estimators for model \eqref{model:arp-sparse} as:
\begin{equation}
	\label{eq:least-square-est-sparse}
	\begin{split}
		\left(\widetilde{\cm{A}}_\mathrm{L} , \widetilde{\cm{A}}_\mathrm{S} \right)
		&=
		\argmin
		{\mathcal{L}}_T(\cm{A}_\mathrm{L}+ \cm{A}_\mathrm{S}) + \lambda \norm{\cm{A}_\mathrm{S}}_1\\ &\hspace{3mm}\text{subject to}\hspace{3mm}
		\cm{A}_\mathrm{L} \in \bm{\Theta}(\underline{q}, P, R_y, R_x) \hspace{3mm}\text{and}\hspace{3mm}
		\norm{\cm{A}_\mathrm{L}}_\infty \leq   \alpha_\mathrm{L} / (PQ^2).
	\end{split}
\end{equation}

\begin{remark}
	To gain an intuition of $\alpha_\mathrm{L}$, we consider $\cm{A}_\mathrm{L}$ with $\norm{\cm{A}_\mathrm{L}}_\Fr \approx 1$. Then setting $\alpha_\mathrm{L} \approx P^{1/2}Q$ induces that $\left|[\cm{A}_{\mathrm{L}}]_{i_1, \dots, i_n, k, j_1,  \dots j_n}\right| \approx P^{-1/2}Q^{-1}$ for all $1 \leq i_l, j_l \leq q_l$, $1 \leq l \leq n$, and $1 \leq k \leq P$, which implies that the information is evenly distributed across all the entries. In contrast, when $\alpha_\mathrm{L} \approx P Q^2$, it is possible that there is a spiky entry with all the mass, leading to an extremely sparse case. 
\end{remark}

\begin{remark}
	For the AR$(P)$ model in Remark \ref{remark-ar2} or the more general tensor-on-tensor regression model \eqref{model:our-cp}, given observations $\{(\cm{Y}_t, \cm{X}_t)\}_{t=1}^T$, the loss function can be defined as $\mathcal{L}^\mathrm{LR}_T(\cm{A})= \sum_{t=1}^{T} \lVert \vectorize(\cm{Y}_t) - [\cm{A}]_n \vectorize(\cm{X}_t)\rVert_2^2 /T$. We further define the parameter space of $\cm{A}$:
	\begin{equation*}
		\begin{split}
			\bm{\Theta}_\mathrm{LR}(\underline{q}, \underline{p}, R_y, R_x) = \{
			\cm{A} \in \mathbb{R}^{q_1 \times \cdots \times q_n \times p_1 \times \cdots \times p_m}:
			& [\cm{A}]_n  = \bm{\Lambda}_y{\bm{G}} \bm{\Lambda}_x^\top, \\
			& \bm{\Lambda}_y\in \mathcal{S}(\underline{q}, R_y), 
			{\bm{G}} \in \mathbb{R}^{R_y \times R_x},
			\bm{\Lambda}_x \in \mathcal{S}(\underline{p}, R_x)
			\},
		\end{split}
	\end{equation*}
	where ${\bm{G}}$ is a reparameterization of $\left(\bm{\Lambda}_y^\top \bm{\Lambda}_y \right)^{-1} \bm{G}$ in \eqref{model:our-cp} and $\underline{p} = \{p_1, \dots, p_m\}$. As a result, the corresponding OLS estimator can be given below,
	\begin{equation}
		\label{eq:est-lr}
		\widehat{\cm{A}}_\mathrm{LR} = \argmin	\mathcal{L}^\mathrm{LR}_T(\cm{A}) \hspace{3mm}\text{subject to}\hspace{3mm} \cm{A} \in \bm{\Theta}_\mathrm{LR}(\underline{q}, \underline{p}, R_y, R_x).
	\end{equation}
\end{remark}

\subsection{Algorithms}
\label{subsec:method-algorithm}


This subsection first introduces an algorithm to search for the low-rank estimate in \eqref{eq:least-square-est}, while the primary challenge lies in handling the CP-based low-rank constraint. 
Denote by $\mathcal{M}_n(\cdot)$ the reversed transformation of sequential matricization in mode $n$, i.e., $\cm{A} = \mathcal{M}_n([\cm{A}]_n)$.
Note that for any $\cm{A} \in \bm{\Theta}(\underline{q}, P, R_y, R_x)$, there exists a decomposition $\cm{A} = \mathcal{M}_n \{(\bm{U}_n \odot \cdots \odot \bm{U}_1)\bm{G}[\bm{I}_p \otimes (\bm{V}_n \odot \cdots \odot \bm{V}_1)]^\top\}$, where $\bm{U}_i \in \mathbb{R}^{q_i \times R_y}$ and $\bm{V}_j \in \mathbb{R}^{q_j \times R_x}$ are matrices with unit-norm columns for $1\leq i,j \leq n$, and $\bm{G} \in \mathbb{R}^{R_y \times R_x}$. 
As a result, the loss function in \eqref{eq:least-square-est} can be rewritten as $\mathcal{L}_T(\bm{U}_1, \dots, \bm{U}_n, \bm{G}, \bm{V}_1, \dots, \bm{V}_n)  \defeq \mathcal{L}_T( \mathcal{M}_n\{(\bm{U}_n \odot \cdots \odot \bm{U}_1)\bm{G}[\bm{I}_p \otimes (\bm{V}_n \odot \cdots \odot \bm{V}_1)]^\top\})$, and we then can optimize the loss function with respect to $\bm{U}_i$'s, $\bm{G}$ and $\bm{V}_j$'s. 

This paper considers an alternating least squares algorithm \citep{de1976additive, zhou2013tensor} for these individual blocks, $\bm{U}_1, \dots, \bm{U}_n, \bm{G}, \bm{V}_1, \dots, \bm{V}_n$, leading to an alternating update on each of them while keeping the other blocks fixed.
Importantly, the update of each block is cast into a traditional linear regression problem. Specifically, when updating $\bm{U}_i$, it is equivalent to find the OLS estimate for the following regression problem:
\begin{equation*}
	\begin{split}
		\vectorize([\cm{Y}_t]_{(i)}) & = \bm{X}_{t,i} \vectorize(\bm{U}_i)
		+ \vectorize([\cm{E}_t]_{(i)}) \\
		\text{with} \hspace{3mm}
		\bm{X}_{t,i} = [(\bm{U}_n\odot\cdots\odot\bm{U}_{i+1}&\odot\bm{U}_{i-1}\odot \cdots\odot\bm{U}_1)\mathrm{Diag}(\tilde{\bm{X}}_t)]\otimes \bm{I}_{q_i} \in \mathbb{R}^{Q \times q_iR_y}
	\end{split}
\end{equation*} 
being the design matrix and $\bm{U}_i$ being unknown parameters. Here $\tilde{\bm{X}}_t = \bm{G}(\bm{I}_p \otimes \bm{\Lambda}_x^\top)\vectorize(\cm{X}_t)\in\mathbb{R}^{R_y}$, and $\mathrm{Diag}(\cdot)$ takes its argument to form a diagonal matrix. 
Similarly, when updating $\bm{V}_j$, we can rewrite model \eqref{model:arp-cp} as
\begin{equation*}
	\vectorize(\cm{Y}_t) = \bm{X}_{t,j}	\vectorize(\bm{V}_j) + \vectorize(\cm{E}_t)
	\hspace{3mm}\text{with}\hspace{3mm}
	\bm{X}_{t,j} = \sum_{r=1}^{R_x}\bm{b}_r^\top \otimes (\bm{\Lambda}_y \tilde{\bm{G}}_r\tilde{\bm{X}}_t^{(j,r)}{}^\top) \in \mathbb{R}^{Q \times q_jR_x},
\end{equation*}
where $\bm{b}_r \in \mathbb{R}^{R_x}$ is the basis vector with the $r$-th entry being one and the remaining being zeros, $\tilde{\bm{G}}_r = (\bm{G}_{1,r}, \dots, \bm{G}_{P,r}) \in \mathbb{R}^{R_y \times P}$, $\bm{G}_{k,r}\in\mathbb{R}^{R_y}$ denotes the $r$-th column of $\bm{G}_k$, $\tilde{\bm{X}}_t^{(j,r)} = (\bm{y}_{t-1}^{(j,r)}, \dots,\bm{y}_{t-P}^{(j,r)}) \in \mathbb{R}^{q_j\times P}$, and $\bm{y}_{t-k}^{(j,r)} = [\cm{Y}_{t-k}]_{(j)}(\bm{v}^{(n)}_{r} \otimes \cdots \otimes \bm{v}^{(j+1)}_{r}\otimes \bm{v}^{(j-1)}_{r}\otimes \cdots \otimes \bm{v}^{(1)}_{r}) \in \mathbb{R}^{q_j}$. When updating $\bm{G}$, model \eqref{model:arp-cp} can be rewritten as
\begin{equation*}
	\vectorize(\cm{Y}_t) = \bm{X}_{t} \vectorize(\bm{G}) + \vectorize(\cm{E}_t)
	\hspace{3mm}\text{with}\hspace{3mm}
	\bm{X}_{t} = [\vectorize(\cm{X}_t)^\top (\bm{I}_P \otimes \bm{\Lambda}_x)]\otimes \bm{\Lambda}_y \in \mathbb{R}^{Q \times PR_yR_x}.
\end{equation*}
The OLS estimate for each of these subproblems has a closed form.
Besides, define the column-wise normalization operation as $\mathrm{ColNorm}(\bm{U}) = \bm{U}\bm{D}$ for any matrix $\bm{U}$ where $\bm{D}$ is diagonal with $\bm{D}_{ii} = \norm{\bm{u}_i}_2^{-1}$ and $\bm{u}_i$ denoting the $i$-th column of $\bm{U}$.
For simplicity of notations, denote $\{\bm{U}_i\}_{i=1}^n$, $\{\bm{U}_l\}_{l=1}^{i-1}$, $\{\bm{U}_l\}_{l=i+1}^{n}$ as $\bm{U}_{1\leq i \leq n}$, $\bm{U}_{l < i}$, $\bm{U}_{l>i}$, respectively. Similar notations apply to $\bm{V}$.
The full estimation procedure is encapsulated in Algorithm \ref{alg:low-rank}.
\begin{algorithm}[H]
	\caption{ALS for CP-based low-rank estimation}\label{alg:low-rank}
	\begin{algorithmic}
		\STATE \textbf{Input:} data $\{\cm{Y}_t\}_{t=1}^T$, order $P$, initial value $\cm{A}^{(0)} = \mathcal{M}_n\{ ( \bm{U}_n^{(0)} \odot \cdots \odot \bm{U}_1^{(0)}) {\bm{G}^{(0)} } [\bm{I}_P \otimes  (\bm{V}_n^{(0)} \odot \cdots \odot \bm{V}_1^{(0)})^\top] \}$
		\REPEAT
		\STATE for $i = 1, 2, \dots, n$, obtain
		$\bm{U}_i^{(k+0.5)} =\argmin_{\bm{U}_i} \mathcal{L}_T(\bm{U}_{l<i}^{(k+0.5)}, \bm{U}_i,  \bm{U}_{l>i}^{(k)}, \bm{G}^{(k)}, \bm{V}_{1\leq j \leq n}^{(k)})$ 
		\STATE for $j=1,2, \dots, n$, obtain
		$\bm{V}_j^{(k+0.5)} = \argmin_{\bm{V}_j} \mathcal{L}_T(\bm{U}_{1\leq i \leq n}^{(k+0.5)}, \bm{G}^{(k)},  \bm{V}_{l<j}^{(k+0.5)}, \bm{V}_{j},\bm{V}_{l>j}^{(k)})$
		\STATE for $i = 1, 2, \dots, n$, obtain $\bm{U}_i^{(k+1)} = \text{ColNorm}(\bm{U}_i^{(k+0.5)})$ and $\bm{V}_i^{(k+1)} = \text{ColNorm}(\bm{V}_i^{(k+0.5)})$
		
		\STATE $\bm{G}^{(k+1)} = \argmin_{\bm{G}} \mathcal{L}_T(\bm{U}_{1 \leq i \leq n}^{(k+1)},\bm{G}, \bm{V}_{1 \leq j \leq n}^{(k+1)})$	
		\UNTIL{convergence}
		\STATE \textbf{Return } $\cm{A}^{(K)} = \mathcal{M}_n\{ ( \bm{U}_n^{(K)} \odot \cdots \odot \bm{U}_1^{(K)}) {\bm{G}^{(K)} } [\bm{I}_P \otimes  (\bm{V}_n^{(K)} \odot \cdots \odot \bm{V}_1^{(K)})^\top] \}$ where $K$ is the terminated iteration
	\end{algorithmic}
\end{algorithm}


\begin{remark}
	The algorithm to search for the low-rank estimate defined in \eqref{eq:est-lr} can be easily obtained by modifying Algorithm \ref{alg:low-rank}. The loss function shall be changed to $\mathcal{L}^\mathrm{LR}_T(\bm{U}_{1\leq i \leq n}, \bm{G}, $ $\bm{V}_{1\leq j \leq n})\defeq \mathcal{L}^\mathrm{LR}_T( \mathcal{M}_n[(\bm{U}_n \odot \cdots \odot \bm{U}_1)\bm{G}(\bm{V}_m \odot \cdots \odot \bm{V}_1)^\top])$, where $\bm{V}_j \in \mathbb{R}^{p_j \times R_x}$ for $1 \leq j \leq m$, $\bm{U}_i$'s and $\bm{G}$ are defined in the same way as before. And the final output becomes $\cm{A}^{(K) }=  \mathcal{M}_n[( \bm{U}_n^{(K)} \odot \cdots \odot \bm{U}_1^{(K)}) {\bm{G}^{(K)} }(\bm{V}_m^{(K)} \odot \cdots \odot \bm{V}_1^{(K)})^\top]$.
\end{remark}

We next construct an algorithm to search for the estimates in \eqref{eq:least-square-est-sparse}. Following \cite{cai2023generalized}, we propose Algorithm \ref{alg:low-rank-sparse} that involves three major steps: trimming, updating the estimate of the sparse tensor $\cm{A}_\mathrm{S}$ and that of the low-rank tensor $\cm{A}_\mathrm{L}$. To start with, we define a trimming operator $\mathrm{Trim}(\cm{A}, \zeta)$ which trims each entry of $\cm{A}$, denoted as $\cm{A}_{i_1, \dots, i_n, k, j_1, \dots, j_n}$, to $\zeta \cdot \text{sign}(\cm{A}_{i_1, \dots, i_n, k, j_1, \dots, j_n})$ if its magnitude is larger than $\zeta$, and does nothing otherwise.
At each iteration, the trimming operation is first applied to $\cm{A}_\mathrm{L}$ to ensure the low-rank and sparse components are identifiable. Then keeping $\cm{A}_\mathrm{L}$ fixed, the search of $\cm{A}_\mathrm{S}$ becomes solving a linear regression problem with a Lasso penalty, for which many optimization algorithms exist, such as the coordinate descent \citep{hastie2015statistical}. Subsequently, fixing $\cm{A}_\mathrm{S}$, we can update $\cm{A}_\mathrm{L}$ using the same routine as in Algorithm \ref{alg:low-rank}, with the loss function adjusted to
\begin{equation*}
	\label{eq:loss-sparse}
	\mathcal{L}_T(\bm{U}_{1\leq i \leq n}, \bm{G}, \bm{V}_{1\leq j \leq n})  \defeq \mathcal{L}_T( \mathcal{M}_n\{(\bm{U}_n \odot \cdots \odot \bm{U}_1)\bm{G}[\bm{I}_p \otimes (\bm{V}_n \odot \cdots \odot \bm{V}_1)]^\top\} + \cm{A}_\mathrm{S}).
\end{equation*}

\begin{algorithm}[H]
	\caption{$l_1$-regularized optimization and ALS for low-rank plus sparse estimation}\label{alg:low-rank-sparse}
	\begin{algorithmic}
		\STATE \textbf{Input:} data $\{\cm{Y}_t\}_{t=1}^T$, order $P$, initial value $\cm{A}_\mathrm{L}^{(0)} =  \mathcal{M}_n\{ ( \bm{U}_n^{(0)} \odot \cdots \odot \bm{U}_1^{(0)}) {\bm{G}^{(0)} } [\bm{I}_P \otimes  (\bm{V}_n^{(0)} \odot \cdots \odot \bm{V}_1^{(0)})^\top] \}$, penalty strength $\lambda$, radius of non-identifiability $\alpha_\mathrm{L}$
		\REPEAT
		\STATE $\bar{\cm{A}}_\mathrm{L}^{(k)} = \text{Trim}(\cm{A}_\mathrm{L}^{(k)}, \alpha_\mathrm{L} / (PQ^2))$
		\STATE $\cm{A}_\mathrm{S}^{(k+1)} = \argmin_{\cm{A}_\mathrm{S}} \mathcal{L}_T(\bar{\cm{A}}_\mathrm{L}^{(k)} + \cm{A}_\mathrm{S}) + \lambda \norm{\cm{A}_\mathrm{S}}_1$
		\STATE $\cm{A}_\mathrm{L}^{(k+1)} = \argmin_{\cm{A}_\mathrm{L}} \mathcal{L}_T(\cm{A}_\mathrm{L} + \cm{A}_\mathrm{S}^{(k+1)})$ using Algorithm \ref{alg:low-rank}
		\UNTIL{convergence}
		\STATE \textbf{Return } $\cm{A}^{(K)}_\mathrm{L}$ and $\cm{A}^{(K)}_\mathrm{S}$ where $K$ is the terminated iteration
	\end{algorithmic}
\end{algorithm}

Lastly, we consider the selection of hyperparameters and initializations of the algorithms. In Algorithm \ref{alg:low-rank}, the autoregressive order $P$ and ranks $(R_y, R_x)$ need to be selected. For computational efficiency, the hold-out method is employed, where the data is split into training and validation sets, each with $T_\mathrm{train}$ and $T_\mathrm{val}$ samples. We first use Algorithm \ref{alg:low-rank} to fit our low-rank model on the training set; and then conduct a rolling forecast on the validation set, i.e., fit the model to historical data with the ending point iterating from $\cm{Y}_{T_\mathrm{train}}$ to $\cm{Y}_{T_\mathrm{train} + T_\mathrm{val}-1}$ and then make one-step-ahead prediction for each iteration. The parameters $P, R_y, R_x$ are selected by minimizing the mean squared error of predictions on the validation set over the ranges $1 \leq P \leq P^{\mathrm{max}}$, $1 \leq R_y \leq R_y^{\mathrm{max}}$ and $1 \leq R_x \leq R_x^{\mathrm{max}}$ for some predetermined upper bounds. 
For Algorithm \ref{alg:low-rank-sparse},  we set the radius of non-identifiability $\alpha_\mathrm{L} = P^{1/2}Q$ for simplicity and use the same hold-out method to select $P, R_y, R_x$ and the penalty strength $\lambda$, where a predetermined candidate set is considered for $\lambda$ as well.
Given these hyperparameters, we follow \cite{zhou2013tensor} to adopt a random initialization, i.e., we randomly initialize $\bm{U}_{1\leq i \leq n}$ and $\bm{V}_{1\leq j \leq n}$ and normalize their columns to get $\bm{U}^{(0)}_{1\leq i \leq n}$ and $\bm{V}^{(0)}_{1\leq j \leq n}$, while $\bm{G}^{(0)}$ is obtained through $\bm{G}^{(0)} = \arg \min_{\bm{G}}\mathcal{L}_T(\bm{U}_{1\leq i \leq n}^{(0)}, \bm{G}, \bm{V}_{1\leq j \leq n}^{(0)})$. Note that the algorithm should be run for multiple initializations to locate a better stationary point.

\section{Theoretical properties}
\label{sec:theory}
\subsection{Tensor-on-tensor regression}
As a preliminary discussion for tensor autoregression, this subsection first establishes the nonasymptotic property of the low-rank estimator $\widehat{\cm{A}}_\mathrm{LR}$ in \eqref{eq:est-lr} for the tensor-on-tensor regression setting, i.e., $\{(\cm{Y}_t, \cm{X}_t)\}_{t=1}^T$ are independent and identically distributed (i.i.d.) observations. 
Prior to its presentation, we introduce the following assumptions, where we denote $\vectorize(\cm{X}_t)$ and $\vectorize(\cm{E}_t)$ as $\bm{x}_t$ and $\bm{e}_t$, respectively.

\begin{assumption}[Regression inputs]\label{assump:reg-input}
	Covariates $\{\bm{x}_t\}$ are i.i.d. with $\mathbb{E}(\bm{x}_t)=\bm{0}$, $\mathrm{var}(\bm{x}_t)=\bm{\Sigma}_x$ and $\sigma^2$-sub-Gaussian distribution. There exist constants $0<c_{x}<C_{x}<\infty$ such that 
	$c_x \leq \lambda_{\min}(\bm{\Sigma}_x) \leq  \lambda_{\max}(\bm{\Sigma}_x)\leq C_x$,
	where the two quantities $c_x$ and $C_x$ depend on the dimensions of $p_j$'s, and they may shrink to zero or diverge to infinity as the dimensions increase.
\end{assumption}

\begin{assumption}[Regression errors]\label{assump:reg-error}
	Error terms $\{\bm{e}_t\}$ are i.i.d. and, conditional on $\bm{x}_t$, $\bm{e}_t$ has mean zero and follows  $\kappa^2$-sub-Gaussian distribution.
\end{assumption}

\begin{assumption}\label{assump:core}
	The matrix $\bm{G}$ has its operator norm upper bounded by a constant $g$ with $g > 0$.
\end{assumption}
The first two sub-Gaussian assumptions on the covariates and error terms are common in literature \citep{wainwright2019high, si2024efficient}. Assumption \ref{assump:core} imposes an upper bound for the matrix $\bm{G}$, which is not a stringent assumption since the dimensions of $\bm{G}$, $R_y$ and $R_x$, tend to be small.
Denote by $d_\mathrm{LR} = R_yR_x + R_y\sum_{i=1}^{n}q_i + R_x\sum_{j=1}^{m}p_j$ the model complexity. Theorem \ref{thm:reg} gives the nonasymptotic error upper bound.

\begin{theorem}\label{thm:reg}
	Suppose that Assumptions \ref{assump:reg-input} - \ref{assump:core}  hold, and the true coefficient tensor $\cm{A}_\mathrm{LR}^* \in \bm{\Theta}_\mathrm{LR}(\underline{q}, \underline{p}, R_y, R_x)$.
	If the sample size $T \gtrsim \max(1, \sigma^4)(d_cd_{\mathrm{LR}} + \sum_{i=1}^{n}\log q_i)$, then
	\begin{align*}
		\|\cm{\widehat{A}}_{\mathrm{LR}}-\cm{A}^*_\mathrm{LR}\|_{\mathrm{F}}&\lesssim \frac{\kappa \sigma}{c_x}\sqrt{\frac{d_{c}d_{\mathrm{LR}}}{T}}
	\end{align*}
	with probability at least $1 - \exp\left\{-c_1d_cd_{\mathrm{LR}}\right\}
	-\exp\left\{-c_2d_c^\prime d_{\mathrm{LR}}^\prime - c_3\sum_{i=1}^n \log q_i\right\}$, where 
	$d_c = \log ((m \vee n )(R_y \wedge R_x)R_y^{1/2}R_x^{1/2})$, $d_\mathrm{LR}^\prime = R_yR_x + R_y + R_x\sum_{j=1}^{m}p_j$, $d_c^\prime = \log ({m(R_y \wedge R_x)R_y^{1/2}R_x^{1/2}})$,
	and 
	$c_i$'s are some positive constants.
\end{theorem}

Note that the parameter space $\bm{\Theta}_\mathrm{LR}(\underline{q}, \underline{p}, R_y, R_x)$ differs from the standard space defined by Tucker decomposition due to the non-orthogonality of loading matrices, which leads to an additional term $d_c$ in its covering number. Besides, the non-standard structure also poses challenges in establishing restricted strong convexity (RSC). We solve this by deriving the RSC property for each response first, and this leads to the additional terms $d_\mathrm{LR}^\prime$ and $d_c^\prime$; please refer to Lemma \ref{lemma:RSCregression} in Supplementary Material for more details. 
When the orders $n, m$ and the ranks $R_y, R_x$ are fixed, the parameters $\sigma^2$, $\kappa^2$, $C_x$, and $c_x$ are bounded away from zero and infinity. We thus have $\|\cm{\widehat{A}}_{\mathrm{LR}}-\cm{A}^*_\mathrm{LR}\|_{\mathrm{F}} = O_p(\sqrt{d_\mathrm{LR}/ T})$. Note that if assuming the coefficient tensor $\cm{A}$ to have a Tucker decomposition with Tucker ranks $(r_1, r_2, \dots, r_{n+m})$, the corresponding OLS estimator will have an error upper bound of $O_p(\sqrt{d_\mathrm{Tucker}/ T})$, where $d_\mathrm{Tucker}$ is given in Remark \ref{remark:tucker} and could be quite large even if $n$ and $m$ are small values.

\subsection{Tensor autoregression}
This subsection considers the tensor autoregression setting that is specified in model \eqref{model:ARp}, and also begins with introducing the assumptions.

\begin{assumption}[Autoregression errors]\label{assump:errorauto}
	For the error term, let $\bm{e}_t =\bm{\Sigma}_e^{1/2}\bm{\xi}_t$. Random vectors $\{\bm{\xi}_t\}$ are i.i.d. with $\mathbb{E}(\bm{\xi}_t)=\bm{0}$ and $\mathrm{var}(\bm{\xi}_t)=\bm{I}_Q$, the entries $(\bm{\xi}_{tj})_{1\leq j\leq Q}$ of $\bm{\xi}_t$ are mutually independent and $\kappa^2$-sub-Gaussian distributed, and there exist constants $0<c_{e}<C_{e}<\infty$ such that $c_{e} \leq \lambda_{\min}(\bm{\Sigma}_e) \leq \lambda_{\max}(\bm{\Sigma}_e)\leq C_e$.
\end{assumption}

\begin{assumption}[Autoregression stationarity]\label{assump:stationarity}
	The determinant of the matrix polynomial $\cm{A}(z) = \bm{I}_Q-\bm{A}_1 z-\cdots-\bm{A}_P z^P$ with $\bm{A}_k = [\cm{A}_k]_n$ is not equal to zero for all $1\leq k \leq P$, $z \in \mathbb{C}$ and $|z| < 1$.
\end{assumption}

The sub-Gaussian condition in Assumption \ref{assump:errorauto} is commonly imposed in high-dimensional time series literature \citep{zheng2020finite, wang2024high}. Assumption \ref{assump:stationarity} is a necessary and sufficient condition for the existence of a unique strictly stationary solution to model \eqref{model:ARp}. It also guarantees the eigenvalues of the Hermitian matrix $\bar{\cm{A}}(z)\cm{A}(z)$ over the unit circle $\{z \in \mathbb{C}: |z|=1\}$ are positive, where $\bar{\cm{A}}(z)$ denotes the conjugate transpose of $\cm{A}(z)$. Following \cite{basu2015regularized}, we define the following quantities,
\begin{align*}
	\mu_{\mathrm{min}}(\cm{A}) = \min_{|z|=1} \lambda_{\mathrm{min}}(\bar{\cm{A}}(z)\cm{A}(z))
	\quad \text{and} \quad
	\mu_{\mathrm{max}}(\cm{A}) = \max_{|z|=1} \lambda_{\mathrm{max}}(\bar{\cm{A}}(z)\cm{A}(z)).
\end{align*}
Furthermore, the AR($P$) model in \eqref{model:ARp} can be alternatively represented as a $PQ$-dimensional AR(1) model $\bm{x}_t = \bm{B} \bm{x}_{t-1}+\bm{\zeta}_t$ with 
\begin{equation}
	\bm{x}_t=
	\begin{bmatrix}
		\vectorize(\cm{Y}_t) \\ \vectorize(\cm{Y}_{t-1}) \\ \vdots \\ \vectorize(\cm{Y}_{t-P+1})
	\end{bmatrix},
	\quad 
	\bm{B}=
	\begin{bmatrix}
		[\cm{A}_1]_n & [\cm{A}_2]_n & \cdots & [\cm{A}_{P-1}]_n &[\cm{A}_P]_n\\
		\bm{I}_Q &\bm{0} & \cdots & \bm{0} & \bm{0}\\
		\bm{0} & \bm{I}_Q & \cdots & \bm{0} & \bm{0} \\
		\vdots & \vdots & \ddots & \vdots & \vdots \\
		\bm{0} & \bm{0} & \cdots & \bm{I}_Q & \bm{0} 
	\end{bmatrix}, 
	\quad
	\bm{\zeta}_t=
	\begin{bmatrix}
		\bm{e}_t \\ \bm{0} \\ \vdots \\ \bm{0} 
	\end{bmatrix}.
	\label{eq:ar1-B}
\end{equation}
And we similarly define the matrix polynomial $\bm{B}(z) = \bm{I}_{PQ} - \bm{B}z$, its conjucate transpose $\bar{\bm{B}}(z)$, and the quantities $\mu_{\mathrm{min}}(\bm{B}) = \min_{|z|=1} \lambda_{\mathrm{min}}(\bar{\bm{B}}(z)\bm{B}(z))$ as well as 
$\mu_{\mathrm{max}}(\bm{B}) = \max_{|z|=1} \lambda_{\mathrm{max}}(\bar{\bm{B}}(z)\bm{B}(z))$. Note that $\mu_{\mathrm{min}}(\bm{B})$, $\mu_{\mathrm{max}}(\bm{B})$ are not necessarily the same as $\mu_{\mathrm{min}}(\cm{A})$, $\mu_{\mathrm{max}}(\cm{A})$ \citep{basu2015regularized}.
Besides, to establish the nonasymptotic properties in autoregression settings, Assumption \ref{assump:core} is also required, which is mild because a large operator norm of $\bm{G}$ could cause the nonstationarity of AR models.

We next give the error upper bound for the estimator $\widehat{\cm{A}}$ defined in \eqref{eq:least-square-est} and that for the estimator $\widehat{\cm{A}}_\mathrm{ALR}$ which is obtained through \eqref{eq:est-lr}.
\begin{theorem}\label{thm:autoreg}
	Suppose that Assumptions \ref{assump:core} - \ref{assump:stationarity} hold, and the true coefficient tensor $\cm{A}^\ast \in \bm{\Theta}(\underline{q}, P, R_y, R_x)$. If the sample size $T - P \gtrsim \max\left(1,\kappa^4\kappa_{U,B}^2/\kappa_{L,A}^2\right)$ $d_{\mathrm{AR}} d_{\mathrm{c}}$, then
	\begin{align*}
		\|{\widehat{\cm{A}}-\cm{A}^\ast}\|_{\mathrm{F}}&\lesssim 
		\frac{\kappa^2\kappa_{U,B}'}{\kappa_{L,A}}
		\sqrt{\frac{d_{\mathrm{AR}}d_c}{T-P}} 
	\end{align*}
	with probability at least $1-c_1\exp\left\{-c_2 d_{\mathrm{AR}}d_{c}\right\}$, where $d_{\mathrm{AR}} = PR_yR_x + (R_y+R_x)\sum_{i=1}^{n}q_i$, {$d_c = \log (nP^{1/2}R_yR_x^{1/2})$}, $\kappa^\prime_{U, B}=C_{e}/\mu^{1/2}_{\mathrm{min}}(\bm{B}^{\ast})$, $\kappa_{L, A}=c_{e}/\mu_{\mathrm{max}}(\cm{A}^\ast)$, $\kappa_{U, B}=C_{e}/\mu_{\mathrm{min}}(\bm{B}^{\ast})$, 
	$\bm{B}^{\ast}$ is the true value of $\bm{B}$ constructed from $\cm{A}^\ast$ as in \eqref{eq:ar1-B},
	and $c_i$'s are some positive constants.
\end{theorem}

\begin{theorem}\label{thm:autoreg-lr}
	Suppose that Assumptions \ref{assump:core} - \ref{assump:stationarity} hold, and the true coefficient tensor $\cm{A}_\mathrm{ALR}^\ast \in \bm{\Theta}_\mathrm{LR}(\underline{q}, \{P, \underline{q}\}, R_y, R_x)$. If the sample size $T - P \gtrsim \max\left(1,\kappa^4\bar{\kappa}_{U,B}^2/\bar{\kappa}_{L,A}^2\right)$ $d_{\mathrm{ALR}} d_{\mathrm{c}}$, then
	\begin{align*}
		\|{\widehat{\cm{A}}_\mathrm{ALR}-\cm{A}_{\mathrm{ALR}}^\ast}\|_{\mathrm{F}}&\lesssim 
		\frac{\kappa^2\bar{\kappa}_{U,B}'}{\bar{\kappa}_{L,A}}
		\sqrt{\frac{d_{\mathrm{ALR}}d_c}{T-P}} 
	\end{align*}
	with probability at least $1-c_1\exp\left\{-c_2 d_{\mathrm{ALR}}d_{c}\right\}$, where $d_{\mathrm{ALR}} = R_yR_x + PR_x + (R_y+R_x)\sum_{i=1}^{n}q_i$, {$d_c = \log ((n+1)(R_y \wedge R_x)R_y^{1/2}R_x^{1/2})$}, $\bar{\kappa}^\prime_{U, B}=C_{e}/\mu^{1/2}_{\mathrm{min}}(\bm{B}_{\mathrm{ALR}}^{\ast})$, $\bar{\kappa}_{L, A}=c_{e}/\mu_{\mathrm{max}}(\cm{A}^\ast_\mathrm{ALR})$, $\bar{\kappa}_{U, B}=C_{e}/\mu_{\mathrm{min}}(\bm{B}_\mathrm{ALR}^{\ast})$, 
	$\bm{B}^{\ast}_\mathrm{ALR}$ is the true value of $\bm{B}_\mathrm{ALR}$ constructed from $\cm{A}_\mathrm{ALR}^\ast$ as in \eqref{eq:ar1-B},
	and $c_i$'s are some positive constants.
\end{theorem}

When the order $n$ and the ranks $R_y, R_x$ are fixed, the parameters $c_e$, $C_e$, $\kappa^2$, $\mu_{\mathrm{min}}(\bm{B}^{\ast})$, $\mu_{\mathrm{max}}(\cm{A}^\ast)$, $\mu_{\mathrm{min}}(\bm{B}_\mathrm{ALR}^{\ast})$, and $\mu_{\mathrm{max}}(\cm{A}_\mathrm{ALR}^\ast)$ are bounded away from zero and infinity. Therefore, we have $\|{\widehat{\cm{A}}-\cm{A}^\ast}\|_{\mathrm{F}} = O_p(\sqrt{d_{\mathrm{AR}} / (T-P)})$ and $\|{\widehat{\cm{A}}_\mathrm{ALR}-\cm{A}_\mathrm{ALR}^\ast}\|_{\mathrm{F}} = O_p(\sqrt{d_{\mathrm{ALR}} / (T-P)})$, where $d_{\mathrm{AR}}$ and $d_{\mathrm{ALR}}$ are the corresponding model complexities.

Lastly, we consider the low-rank plus sparse model in \eqref{model:arp-sparse} and provide the error upper bound for the low-rank and sparse estimators that are defined in \eqref{eq:least-square-est-sparse}. Two additional assumptions are imposed: Assumption \ref{assump:identifiability} is to ensure the identifiability of the low-rank and sparse components as discussed in Section \ref{subsec:method-ar}, while Assumption \ref{assump:sparse} specifies the elementwise sparsity for the sparse tensor $\cm{A}_\mathrm{S}$.

\begin{assumption}[Identifiability]\label{assump:identifiability}
	The low-rank coefficient tensor $\cm{A}_\mathrm{L}$ satisfies $\norm{\cm{A}_\mathrm{L}}_\infty \leq \alpha_\mathrm{L} / (PQ^2)$, where $\alpha_\mathrm{L}$ is the radius of non-identifiability.
\end{assumption}

\begin{assumption}[Sparsity]\label{assump:sparse}
	The sparse coefficient tensor $\cm{A}_\mathrm{S}$ has at most $s$ nonzero entries.
\end{assumption}

\begin{theorem}\label{thm: ar-sparse}
	Suppose that Assumptions \ref{assump:core} - \ref{assump:stationarity} hold, and the true coefficient tensor $\cm{A}^\ast_\mathrm{L} \in \bm{\Theta}(\underline{q}, P, R_y, R_x)$ and satisfies Assumption \ref{assump:identifiability} while $\cm{A}^\ast_\mathrm{S}$ satisfies Assumption \ref{assump:sparse}. 
	If the sample size $T-P \gtrsim \max\left(1,\kappa^4\kappa_{U,B}^2/\kappa_{L,A}^2\right)\left(d_{\mathrm{AR}} d_{\mathrm{c}} + s\log(PQ^2) \right)$, then
	\[
	\norm{ \widetilde{\cm{A}}_{\mathrm{L}} - \cm{A}^\ast_{\mathrm{L}}}_\Fr^2 + \norm{ \widetilde{\cm{A}}_{ \mathrm{S}} - \cm{A}^\ast_{\mathrm{S}}}_\Fr^2 
	\lesssim
	\underbrace{
		\frac{\kappa^4 C_e\kappa_{U,B}}{\kappa_{L, A}^2} \cdot \frac{d_cd_{\mathrm{AR}} + {s\log(PQ^2)}}{T-P}}_{\text{estimation error}}+ 
\underbrace{
	\frac{\kappa_{U, A}^2 }{\kappa_{L, A}^2} \cdot \frac{s \alpha_\mathrm{L}^2}{PQ^2}
}_{\text{unidentifiability error}}
\]
with probability at least $	1-c_1\exp\left\{-c_2d_{c}d_{\mathrm{AR}}\right\}
-c_3 \exp \left\{-c_4 s\log(PQ^2)\right\}$,
where 
$\kappa_{U, A}=C_{e}/$ $\mu_{\mathrm{min}}(\cm{A}^{\ast})$, $\kappa_{L, A}=c_{e}/\mu_{\mathrm{max}}(\cm{A}^\ast)$, $\kappa_{U, B}=C_{e}/{\mu_{\mathrm{min}}(\bm{B}^{\ast})}$,  $d_{\mathrm{AR}} = PR_yR_x + (R_y+R_x)\sum_{i=1}^{n}q_i$, {$d_c = \log (nP^{1/2}R_yR_x^{1/2})$}, $\cm{A}^\ast =  \cm{A}^\ast_\mathrm{L} +\cm{A}^\ast_\mathrm{S}$, 
$\bm{B}^{\ast}$ is the true value of $\bm{B}$ constructed from $\cm{A}^\ast$ as in \eqref{eq:ar1-B},
and $c_i$'s are some positive constants.
\end{theorem}
The error bound in the above theorem comprises two parts. The first part corresponds to the estimation error that arises from the randomness in data and limited sample size. Specifically, the term $\kappa^4 C_e\kappa_{U,B} / \kappa_{L, A}^2$ captures the effect of dependence in time series data. When the order $n$ and the ranks $R_y, R_x$ are fixed, the error sum becomes $O_p((d_\mathrm{AR} + s\log(PQ^2))/(T-P))$, where $d_\mathrm{AR} + s\log(PQ^2)$ is the model complexity. On the other hand, the second part of the error bound emanates from the non-identifiability inherent in the structure of the true low-rank and sparse components \citep{basu2019low}. It does not rely on the sample size, and thus will not vanish no matter how large the sample size is.

\section{Simulation studies}
\label{sec:simulation}
This section conducts two experiments to evaluate finite-sample performances of the proposed low-rank estimator in \eqref{eq:least-square-est} and low-rank plus sparse estimator in \eqref{eq:least-square-est-sparse}, respectively.

In the first experiment, the data generating process is the autoregression model at \eqref{model:ARp} with $[\cm{A}]_n = ([\cm{A}_1]_n, \dots, [\cm{A}_P]_n) = \bm{\Lambda}_y \bm{G}(\bm{I}_P \otimes \bm{\Lambda}_x^\top)$, $P=2$ and $n=3$, i.e., $\cm{Y}_t \in \mathbb{R}^{q_1 \times q_2 \times q_3}$. We consider three different scenarios for generating the entries of $\cm{E}_t$: (i) independently sampling from a uniform distribution on $(-0.5, 0.5)$, (ii) independently sampling from a standard normal distribution, or (iii) correlated with normality, i.e., $\vectorize(\cm{E}_t) \sim N(\bm{0}, \bm{\Sigma}_e)$ with $\bm{\Sigma}_e = (0.5^{|i-j|})_{1\leq i, j \leq Q}$ and $Q = q_1q_2q_3$. To generate the coefficient $\cm{A}$, we first construct matrices $\{\bm{U}_i \in \mathbb{R}^{q_i \times R_y}\}_{i=1}^3$ and $\{\bm{V}_j \in \mathbb{R}^{q_j \times R_x}\}_{j=1}^3$ by sampling their entries independently from a standard normal distribution and normalizing each column to have unit-norm, and subsequently let $\bm{\Lambda}_y= \bm{U}_3 \odot \bm{U}_2 \odot \bm{U}_1$ and $\bm{\Lambda}_x = \bm{V}_3 \odot \bm{V}_2 \odot \bm{V}_1$. Besides, the matrix $\bm{G}$ is also generated with i.i.d. standard normal entries and then rescaled to make $\lVert \cm{A} \rVert_\Fr = 0.9$.

\begin{figure}[t!]
\includegraphics[width=1.\textwidth]{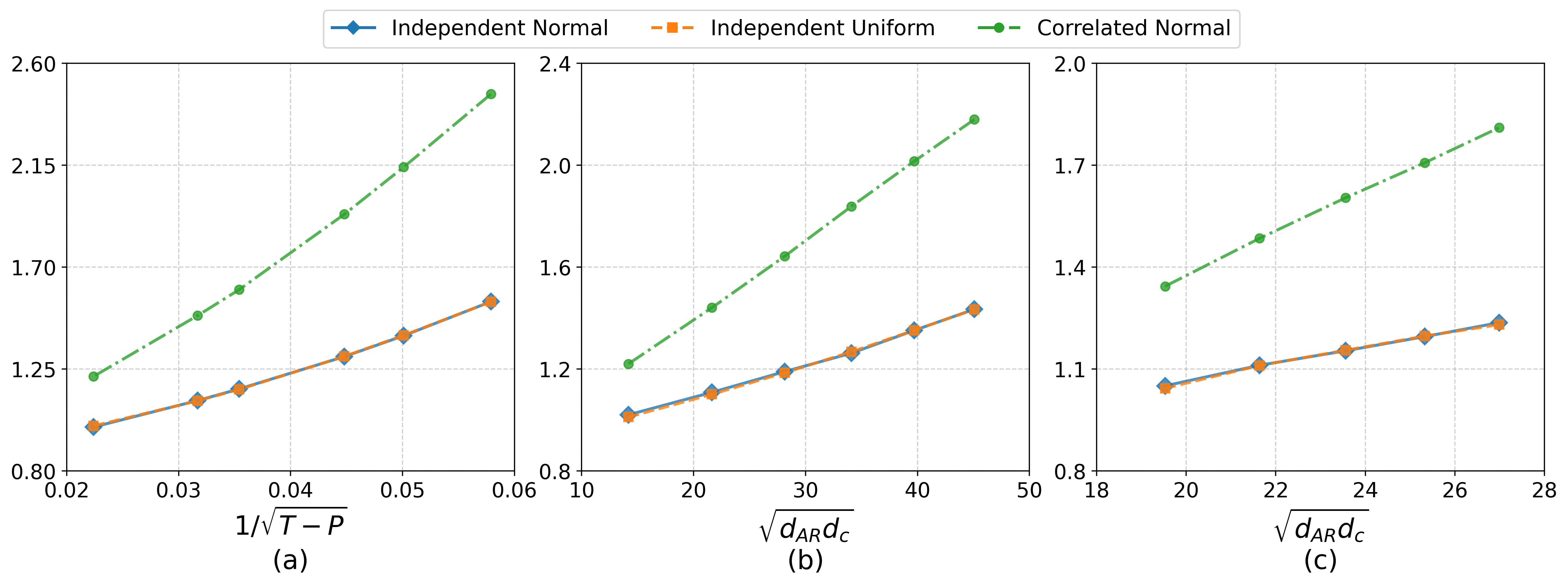}
\caption{Average estimation errors $\lVert \widehat{\cm{A}} - \cm{A}\rVert_\Fr$ for the tensor autoregression model with low-rank structure under settings: (a) changing the sample size $T$, (b) changing the ranks $R_y$ and $R_x$, and (c) changing the dimensions $q_i$'s. Three different ways of generating the error $\cm{E}_t$ are considered, as specified in the legend. }
\label{fig:simulation-low-rank}
\end{figure}

From Theorem \ref{thm:autoreg}, the estimation error converges at a rate of $\sqrt{d_\mathrm{AR} d_c / (T-P)}$, where $d_{\mathrm{AR}} = PR_yR_x + (R_y+R_x)\sum_{i=1}^{n}q_i$ and {$d_c = \log (nP^{1/2}R_yR_x^{1/2})$}. To verify this convergence rate numerically, three different settings are considered: (a) $(q_i, R_y, R_x)$ are fixed at $(10, 3, 2)$ for all $i \in \{1, 2, 3\}$, while the sample size $T$ varies among the set $\{300, 400, 500, 800, 1000, 1200\}$; (b) $(q_i, T)$ are fixed at $(10, 1000)$, while the rank $R_x$ varies from one to six and $R_y = R_x+1$; and (c) $(R_y, R_x, T)$ are fixed at $(3, 2, 1000)$, while a value that varies among $\{8, 10, 12, 14,16\}$ is applied to all $q_i$'s. To search for the low-rank estimates $\widehat{\cm{A}}$, Algorithm \ref{alg:low-rank} is employed. The estimation errors $\lVert \widehat{\cm{A}} - \cm{A} \rVert_\Fr$ for these settings, averaged over 500 replications, are presented in Figure \ref{fig:simulation-low-rank}. The estimation errors exhibit linear growth in Figure \ref{fig:simulation-low-rank} (a)-(c), which implies that $\lVert \widehat{\cm{A}} - \cm{A} \rVert_\Fr$ is proportional to $1 / \sqrt{T-P}$ and $\sqrt{d_\mathrm{AR} d_c}$ and thus validates the theoretical results in Theorem \ref{thm:autoreg}. Additionally, in scenarios where the entries of $\cm{E}_t$ follow a correlated normal distribution as in (iii), the estimation errors are slightly worse than independent cases. Utilizing a generalized least squares method could potentially enhance their performances.

\begin{figure}[t!]
\includegraphics[width=1.\textwidth]{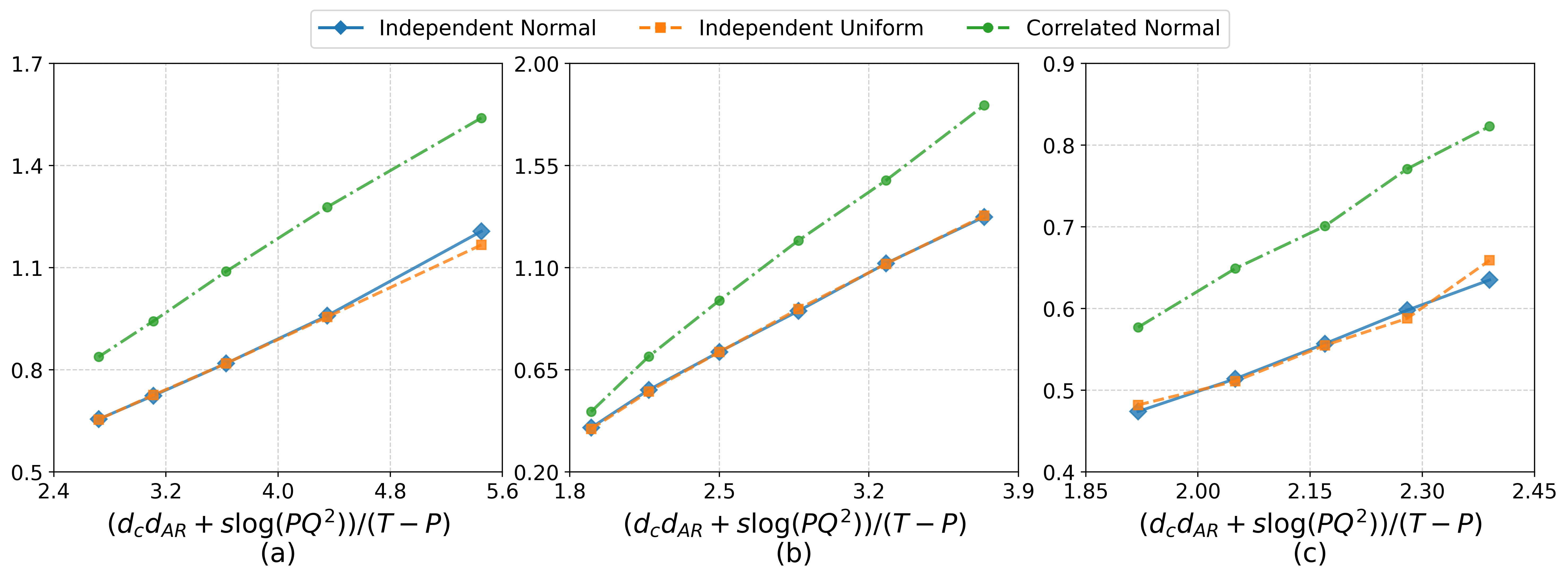}
\caption{Average errors $\lVert \widetilde{\cm{A}}_\mathrm{L} - \cm{A}_\mathrm{L} \rVert_\Fr^2 + \lVert \widetilde{\cm{A}}_\mathrm{S} - \cm{A}_\mathrm{S} \rVert_\Fr^2$ for tensor autoregression model with low-rank plus sparse structure under settings: (a) changing the sample size $T$, (b) changing the ranks $R_y$ and $R_x$, and (c) changing the dimensions $q_i$'s. Three different ways of generating the error $\cm{E}_t$ are considered, as specified in the legend. }
\label{fig:simulation-low-rank-sparse}
\end{figure}

In the second experiment, the data generating process is the low-rank plus sparse model at \eqref{model:arp-sparse}, where $[\cm{A}_\mathrm{L}]_n = ([\cm{A}^\mathrm{L}_1]_n, \dots, [\cm{A}^\mathrm{L}_P]_n) = \bm{\Lambda}_y \bm{G}(\bm{I}_P \otimes \bm{\Lambda}_x^\top)$ and we keep $P=2$ and $n=3$. The generation of the error $\cm{E}_t$, and the components $\bm{\Lambda}_y$, $\bm{\Lambda}_x$, $\bm{G}$ are the same as in the first experiment, except that $\bm{G}$ is now rescaled to ensure that $\lVert \cm{A}_\mathrm{L} \rVert_\infty \leq \alpha_\mathrm{L} / (PQ^2)$. Moreover, we consider a sparse tensor $\cm{A}_\mathrm{S}$ with $s$ non-zero entries, each drawn from a standard normal distribution. This tensor is then rescaled to make $\lVert  \cm{A}_\mathrm{L} + \cm{A}_\mathrm{S} \rVert_\Fr = 0.6$.

Based on Theorem \ref{thm: ar-sparse}, the sum of squared errors, $\lVert \widetilde{\cm{A}}_\mathrm{L} - \cm{A}_\mathrm{L} \rVert_\Fr^2 + \lVert \widetilde{\cm{A}}_\mathrm{S} - \cm{A}_\mathrm{S} \rVert_\Fr^2$, consists of two components: estimation error and unidentifiability error. Firstly, to verify the estimation error, we fix $s = 60$ and $\alpha_\mathrm{L} = P^{1/2}Q$ such that $\lVert \widetilde{\cm{A}}_\mathrm{L} - \cm{A}_\mathrm{L} \rVert_\Fr^2 + \lVert \widetilde{\cm{A}}_\mathrm{S} - \cm{A}_\mathrm{S} \rVert_\Fr^2$ has a convergence rate of $(d_cd_\mathrm{AR} + s\log(PQ^2)) / (T-P)$. Three different settings are further considered: (a) $(q_i, R_y, R_x )$ are fixed at $(10, 3, 2)$ for all $i \in \{1, 2, 3\}$, whereas the sample size $T$ varies among the set $\{400, 500, 600, 700, 800\}$; (b) $(q_i, T)$ are fixed at $(10, 1000)$, whereas the rank $R_x$ ranges from one to six and $R_y = R_x + 1$; and (c) $(R_y, R_x, T)$ are fixed at $(3, 2, 1000)$, whereas a value from $\{8, 9, 10, 11, 12\}$ is applied to all $q_i$'s. 
Algorithm \ref{alg:low-rank-sparse} is used to search for the low-rank estimate $\widetilde{\cm{A}}_\mathrm{L}$ and the sparse estimate $\widetilde{\cm{A}}_\mathrm{S}$.
The sum of squared errors, $\lVert \widetilde{\cm{A}}_\mathrm{L} - \cm{A}_\mathrm{L} \rVert_\Fr^2 + \lVert \widetilde{\cm{A}}_\mathrm{S} - \cm{A}_\mathrm{S} \rVert_\Fr^2$, averaged over 500 replications, are plotted in Figure \ref{fig:simulation-low-rank-sparse}. The linearity observed across all settings confirms the aforementioned convergence rate.

\begin{table}[t!]
\centering
\caption{True positive rate (TPR) and false positive rate (FPR) of the low-rank plus sparse model for identifying the sparse components of $\cm{A}_\mathrm{S}$ across various $\alpha_\mathrm{L}$ values. }
\label{tab:simulation-alpha}
\resizebox{1.\textwidth}{!}{
	\begin{tabular}{cccccccccc}
		\hline
		$\alpha_\mathrm{L}$ & $P^{1/2}Q / 8$ &  $P^{1/2}Q / 4$ &  $P^{1/2}Q / 2$  &  $P^{1/2}Q $ & $2P^{1/2}Q$ & $4P^{1/2}Q$  & $8P^{1/2}Q$ & $16P^{1/2}Q$ \\
		\hline
		TPR (in \%)  & 62.93 & 62.37 & 61.39 & 60.43 & 57.38 & 51.41 & 40.58 & 31.47 \\
		FPR (in \%)  & 25.04 & 25.03 & 25.04 & 25.04 &  25.04 & 25.04 & 25.03 & 25.04 \\
		\hline
\end{tabular}}
\end{table}

Secondly, inspired by \cite{basu2019low}, we explore the impact of non-identifiability on discerning the sparsity pattern of $\cm{A}_\mathrm{S}$ by altering the setting of $\alpha_\mathrm{L}$. Specifically, we fix $(q_i, R_y, R_x, T, s)$ at $(8, 3, 2, 1000, 150)$ for all $i\in \{1, 2, 3\}$ and set $\alpha_\mathrm{L} = 2^{k} \cdot P^{1/2}Q$ with $k$ in the range of $\{-3, -2, -1, 0, 1, 2, 3, 4\}$. The true positive rate (TPR) and false positive rate (FPR) serve as evaluation metrics, which are defined as $\mathrm{TPR} = \sum_{i_1i_2\cdots i_7} \mathbbm{1}\{(\widetilde{\cm{A}}_\mathrm{S})_{i_1i_2\cdots i_7} \neq 0  \text{ and } ({\cm{A}}_\mathrm{S})_{i_1i_2\cdots i_7} \neq 0  \} / \sum_{i_1i_2\cdots i_7} \mathbbm{1}\{({\cm{A}}_\mathrm{S})_{i_1i_2\cdots i_7} \neq 0  \}$ and
$\mathrm{FPR} =  \sum_{i_1i_2\cdots i_7} \mathbbm{1}\{(\widetilde{\cm{A}}_\mathrm{S})_{i_1i_2\cdots i_7} \neq 0  \text{ and } ({\cm{A}}_\mathrm{S})_{i_1i_2\cdots i_7} = 0  \} / \sum_{i_1i_2\cdots i_7} \mathbbm{1}\{({\cm{A}}_\mathrm{S})_{i_1i_2\cdots i_7} = 0  \}$, respectively. Here $({\cm{A}}_\mathrm{S})_{i_1i_2\cdots i_7}$ stands for the $(i_1, i_2, \dots, i_7)$-th entry of $\cm{A}_\mathrm{S}$ with $1\leq i_{k}, i_{k+4} \leq q_k$ for $k\in \{1, 2, 3\}$ and $1\leq i_4 \leq P$.
Using the estimates obtained from Algorithm \ref{alg:low-rank-sparse}, we report the TPR and FPR for various $\alpha_\mathrm{L}$ values, averaged over 500 replications, in Table \ref{tab:simulation-alpha}. Consistent with observations in \cite{basu2019low}, we note that a smaller $\alpha_\mathrm{L}$ substantially enhances the identification of true nonzero entries in $\cm{A}_\mathrm{S}$, leading to a better separation between low-rank and sparse components.

\section{An empirical example}
\label{sec:real-data}

El Ni$\tilde{\text{n}}$o-Southern Oscillation (ENSO) is a climate phenomenon characterized by the anomalous episodic warming of the eastern equatorial Pacific Ocean. The event exerts considerable influences on global climatic and environmental conditions, and thus its accurate prediction can yield essential scientific and economic benefits \citep{josef2013improved}. 
This section analyzes the upper ocean temperature data which are proven to be significant predictors for ENSO events \citep{clarke2003improving, mc2003tropical}. The dataset, sourced from \cite{zhou2023a}, contains monthly observations of seven-layer ocean temperature anomalies in the upper 150 meters (at 5, 20, 40, 60, 90, 120, and 150 meters) spanning from 1980 to 2021. Geographically, the data cover a grid of 51 longitudes and 41 latitudes, evenly spaced from 154$^{\circ}$ East to 106$^{\circ}$ West and from $20^{\circ}$ South to $20^{\circ}$ North. 
For preprocessing, the data are first rendered stationary by removing monthly seasonality and applying differencing, and then are standardized to have zero means.
As a result, we have $T = 503$ observations of tensor-valued time series $\{\cm{Y}_t \in \mathbb{R}^{51 \times 41 \times 7}\}_{t=1}^{T}$.

The tensor autoregression model at \eqref{model:ARp} is applied with each coefficient tensor $\cm{A}_k \in \mathbb{R}^{51 \times 41 \times 7 \times 51 \times 41 \times 7}$ containing $14637^2$ parameters. To enable a feasible estimation, we consider our proposed (i) CP-based low-rank estimator at \eqref{eq:least-square-est} and (ii) low-rank plus sparse estimator at \eqref{eq:least-square-est-sparse}. We also compare them with three existing dimension reduction methods in literature, including restricting the stacked coefficient $\cm{A}$ to follow (iii) CP decomposition and (iv) Tucker decomposition, as well as (v) the Lasso method which imposes an $l_1$-regularization on $\cm{A}$. To select the hyperparameters, we split the whole dataset into training, validation, and testing sets, comprising 400, 50, and 53 samples, respectively. The testing set is reserved for evaluating prediction performance. For our low-rank estimator, the hold-out method in Section \ref{subsec:method-algorithm} is used and the values of autoregressive order $P$ and ranks $R_y, R_x$ are chosen as 1, 3, and 2, respectively, with $P^{\mathrm{max}} = R_y^{\mathrm{max}} = R_x^\mathrm{max}=12$. To save computations, we also utilize these choices of $P$, $R_y$, and $R_x$ for our low-rank plus sparse estimator, and the hold-out method further chooses the penalty strength $\lambda = 0.01$ from the candidate set $\{0.001, 0.003, 0.005, 0.008, 0.01, 0.03, 0.05, 0.08, 0.1, 0.3, 0.5, 0.8, 1\}$.
For consistency, we set $P=1$ for all other methods and apply the hold-out method to determine the CP rank as $3$ and the penalty strength for Lasso as 0.01, respectively. Additionally, based on the 95\% cumulative percentage of the total variation in \cite{han2022an}, the Tucker ranks are selected as $(3,6,6,3,6,6)$.

\begin{table}[t!]
\centering
\caption{Mean squared forecast errors (MSFE) and mean absolute forecast errors (MAFE) of our low-rank and low-rank plus sparse models and three competing methods on the testing set of the upper ocean temperature data. The smallest numbers in each row are in bold.}
\label{tab:real-data}
\begin{tabular}{cccccc}
	\hline
	& Low-rank (ours) & Low-rank plus sparse (ours) & CP  & Tucker & Lasso \\
	\hline
	MSFE  & 0.361  & \textbf{0.271} & 0.373 & 0.440   & 0.272  \\
	MSAE  & 0.375  &  \textbf{0.216}  & 0.384  & 0.409  &  0.320 \\
	\hline
\end{tabular}
\end{table}

We next compare the prediction performance of these five methods. To search for the estimates, Algorithms \ref{alg:low-rank} and \ref{alg:low-rank-sparse} are used for our proposed low-rank and low-rank plus sparse models, respectively; the alternating least squares method in \cite{lock2018tensor} and the projected gradient descent algorithm in  \cite{han2022an} are adopted for tensor autoregression models with CP and Tucker decompositions, respectively; and lastly, the coordinate descent in \cite{hastie2015statistical} is employed for Lasso. The performance of the rolling forecast on the testing set is used for comparison. Table \ref{tab:real-data} gives the mean squared forecast errors (MSFE) and mean absolute forecast errors (MAFE) for all five methods, where the smallest forecast errors are obtained by our low-rank plus sparse model. The result confirms the existence of both low-rank and sparse signals in data. Moreover, the superior performance of our low-rank method over CP and Tucker decompositions underscores the benefits brought by the enhanced flexibility and improved efficiency in dimension reduction, respectively. The Lasso method demonstrates a strong prediction performance as anticipated. However, it is less interpretable than our low-rank method due to its challenges in delineating crucial regions coherently, especially when both response and covariate tensors are high-order.

\begin{figure}[t]
\includegraphics[width=1.\textwidth]{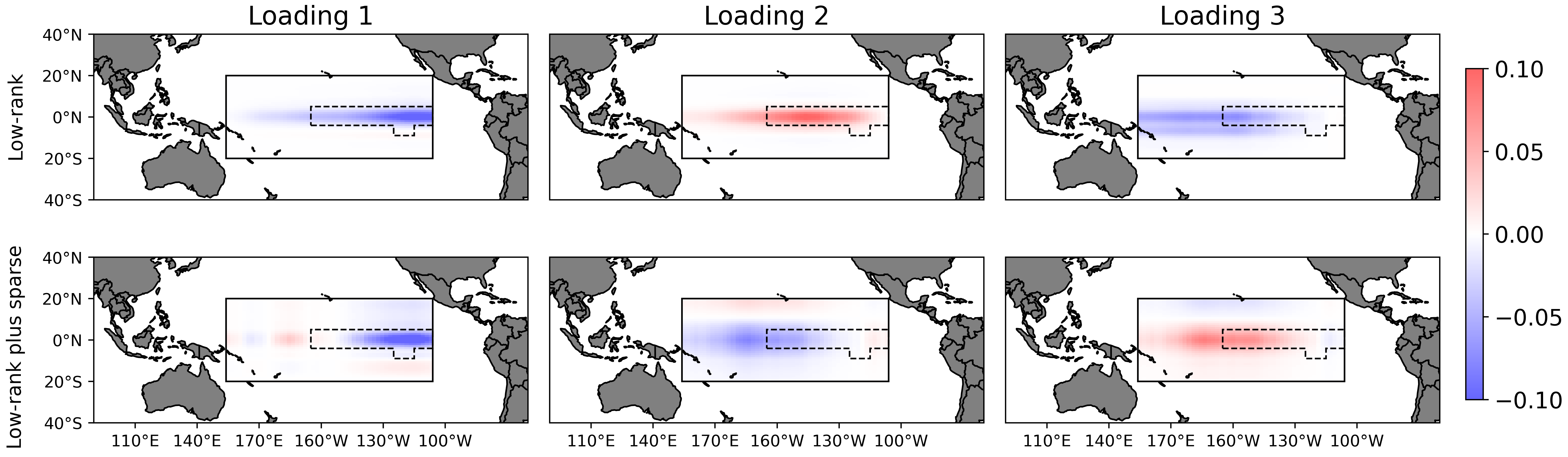}
\caption{Heatmaps of spatial loadings for response features: the top and bottom panels are from our low-rank and low-rank plus sparse models, respectively. One possible boundary of the El Ni$\tilde{\text{n}}$o basin is delineated by dashed lines in figures.}
\label{fig:response-loading-spatial}
\end{figure}

Finally, we fit our low-rank and low-rank plus sparse models using the entire dataset and the chosen ranks $R_y  = 3, R_x = 2$, and obtain the loading matrices $\{\bm{U}_i\}_{i=1}^3$ and $\{\bm{V}_j\}_{j=1}^3$ for both low-rank and low-rank plus sparse estimates. Here, $i= 1, 2, 3$ correspond to the longitude, latitude, and variable dimensions, respectively; and therefore, the matrices $\bm{u}_r^{(1)} \circ \bm{u}_r^{(2)}$ for $1 \leq r \leq R_y$, where $\bm{u}_r^{(i)}$ denotes the $r$-th column of $\bm{U}_i$, represent spatial loadings for different response features. This holds similarly for covariate features. Figure \ref{fig:response-loading-spatial} showcases heatmaps of these spatial loadings for response features, where the upper and bottom panels correspond to our low-rank model and the rest are from our low-rank plus sparse model. We also draw one possible region boundary of the El Ni$\tilde{\text{n}}$o basin with dashed lines on these figures, according to its definition in \cite{josef2013improved}. It can be seen that both models summarize spatial information from response tensors similarly: at least one feature emphasizes the El Ni$\tilde{\text{n}}$o basin with substantial weights, whereas the remaining focus more on the areas outside the basin. This observation aligns with the consensus in the literature that the El Ni$\tilde{\text{n}}$o basin has a unique role in climate dynamics and its climate conditions are influenced in a different way compared to its surrounding areas \citep{gozo2011emergence,josef2013improved}. 
In addition, Figure \ref{fig:response-loading-var} displays the heatmaps of variable loadings for response features, where the left and right panels are from our low-rank and low-rank plus sparse models, respectively. It is evident that both models have two features that rely more on temperatures at deeper ocean levels, while the other summarizes the information mainly from the shallower ocean levels. For covariate features, there is no clear pattern possibly because we do not distinguish periods before and during the El Ni$\tilde{\text{n}}$o episodes as suggested by \cite{josef2013improved}, given the limited sample size.

\begin{figure}[t]
\includegraphics[width=1.\textwidth]{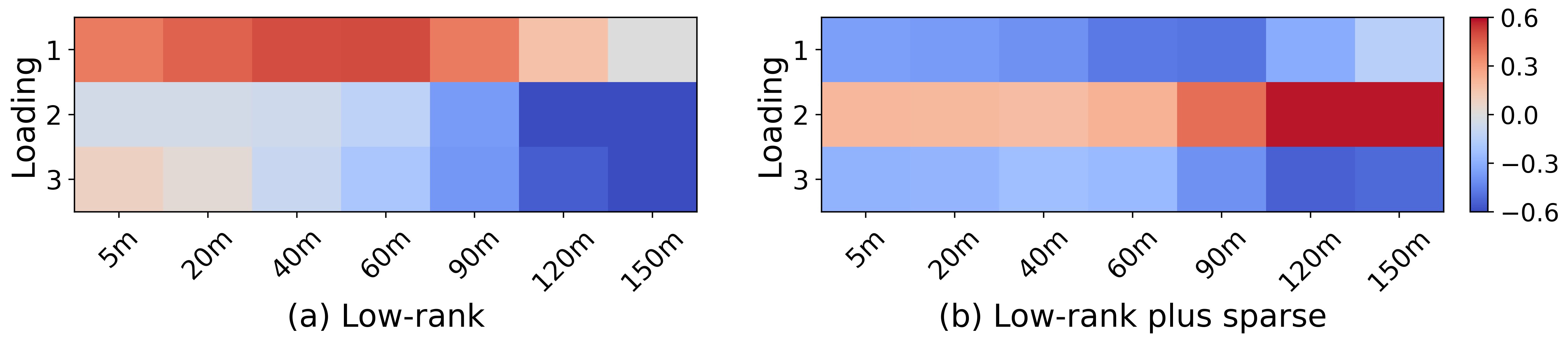}
\caption{Heatmaps of variable loadings for response features: the left and right panels are from our low-rank and low-rank plus sparse models, respectively.}
\label{fig:response-loading-var}
\end{figure}

\section{Conclusion and discussion}
\label{sec:conclusion}
Tucker and CP decompositions are commonly used to impose low-rank constraints on tensors in the literature, and they can be applied to the coefficient tensor in autoregressive models for tensor-valued time series. While Tucker decomposition offers a nice interpretation of supervised factor modeling, its efficiency in dimension reduction diminishes very quickly as tensor order increases. In contrast, CP decomposition remains efficient but only provides a vague statistical interpretation. To attain both interpretability and efficiency, this paper proposes a new approach under the supervised factor modeling paradigm by using CP decomposition to summarize high-order response and covariate tensors into features and then performing a regression on these latent features. This leads to a novel CP-based low-rank structure for the coefficient tensor.
Moreover, to account for heterogeneous signals or potential model misspecifications due to an exact low-rank assumption, this paper also considers a low-rank plus sparse model, incorporating an additional sparse coefficient tensor.
High-dimensional estimation is provided for both models, and their non-asymptotic properties are established. An alternating least squares algorithm is suggested to search for the estimates. 
Finally, as demonstrated through the ENSO example, our models surpass existing methods in prediction accuracy while effectively extracting interpretable features from tensor-valued time series data.

The proposed methodology in this paper opens avenues for three potential extensions. Firstly, for our tensor low-rank autoregressive model, a sample size of $T \gtrsim \sum_{i=1}^{n}q_i$ is required for achieving the estimation consistency. 
However, real applications, such as time-course gene expression data analysis \citep{lozano2009grouped}, may not always meet this requirement.
To address this issue, we may impose sparsity to the response loading matrices $\{\bm{U}_i\}$ and covariate loadings $\{\bm{V}_j\}$, which can further reduce the dimension and relax the sample size requirement.
Secondly, it is of interest to make a statistical inference on estimated tensor coefficients, such as constructing confidence intervals for specific entries or loading subspaces \citep{cai2020uncertainty,xia2022inference}, which we defer to future investigations.
Finally, we may extend the supervised factor modeling framework to utilize the tensor ring decomposition \citep{zhao2016tensorring} for feature extraction. Compared to CP-based structures, tensor ring decomposition may lead to enhanced representation ability and numerical stability while maintaining a linear growth in model complexity with respect to the tensor order.

\newpage
\appendix

\title{Supplementary Material for ``An Efficient and Interpretable Autoregressive Model for High-Dimensional Tensor-Valued Time Series"}
\author{}
\date{\today}
\makeatletter
\renewcommand{\@thanks}{}
\makeatother
\maketitle

\begin{abstract}
	\vspace{-2mm}
	This supplementary material includes two sections. Section \ref{appendix:sec-theory} presents the proofs of theorems in Section \ref{sec:theory} of the manuscript, whereas the auxiliary lemmas are provided in Section \ref{appendix:sec-lemma}. 
\end{abstract}

\section*{Contents}
\startcontents
\printcontents{ }{1}{}

\section{Proof of Theorems} \label{appendix:sec-theory}
This section provides the proofs of theorems in Section \ref{sec:theory} of the manuscript.
\subsection{Proof of Theorem \ref{thm:reg}}
\begin{proof}
	Denote sets $\bm{\Xi}_\mathrm{LR}(\underline{q}, \underline{p}, R_y, R_x) = \{\bm{\Delta} \in \mathbb{R}^{\prod_{i=1}^{n} q_i \times \prod_{j=1}^{m} p_j}: \bm{\Delta} = \bm{A}_1 - \bm{A}_2,    \mathcal{M}_n(\bm{A}_1),$ $  \mathcal{M}_n(\bm{A}_2) \in \bm{\Theta}_\mathrm{LR}(\underline{q}, \underline{p}, R_y, R_x)\}$ and $\bm{\Xi}'_{\mathrm{LR}}(\underline{q}, \underline{p}, R_y, R_x) = \{\bm{\Delta} \in \bm{\Xi}_\mathrm{LR}(\underline{q}, \underline{p}, R_y, R_x): \lVert \bm{\Delta} \rVert_\Fr \leq 1 \}$. 
	For simplicity of notations, we abbreviate them to $\bm{\Xi}_\mathrm{LR}$ and $\bm{\Xi}^\prime_\mathrm{LR}$, respectively, and let $\bm{A} \defeq [\cm{A}]_n$, $\bm{y}_t \defeq \vectorize(\cm{Y}_t)$, $\bm{x}_t \defeq \vectorize(\cm{X}_t)$, and $\bm{e}_t \defeq \vectorize(\cm{E}_t)$.
	
	Recall the quadratic loss function
	\[
	\mathcal{L}_T^{\mathrm{LR}}(\cm{A}) 
	= \frac{1}{T}\sum_{t=1}^T \norm{\vectorize(\cm{Y}_t) - [\cm{A}]_n \vectorize(\cm{X}_t)}_2^2
	= \frac{1}{T}\sum_{t=1}^T \norm{\bm{y}_t - \bm{A}\bm{x}_t}_2^2.
	\]
	By the optimality of $\cm{\widehat{A}}_{\mathrm{LR}}$, we have
	\begin{align}
		&\frac{1}{T}\sum_{t=1}^T\norm{\bm{y}_t-\bm{\widehat{A}}_{\mathrm{LR}}\bm{x}_t}_2^2\leq
		\frac{1}{T}\sum_{t=1}^T\norm{\bm{y}_t-\bm{A}_\mathrm{LR}^\ast\bm{x}_t}_2^2\nonumber\\
		\Rightarrow&\frac{1}{T}\sum_{t=1}^T\norm{\bm{\Delta}\bm{x}_t}_2^2
		\leq
		\frac{2}{T}\sum_{t=1}^T\langle\bm{e}_t,\bm{\Delta}\bm{x}_t\rangle\nonumber\\
		\Rightarrow&\frac{1}{T}\sum_{t=1}^T\norm{\bm{\Delta}\bm{x}_t}_2^2
		\leq
		2\norm{\bm{\Delta}}_{\mathrm{F}}\sup_{\bm{\Delta}\in\bm{\Xi}^\prime_\mathrm{LR}}\frac{1}{T}\sum_{t=1}^T\langle\bm{e}_t,\bm{\Delta}\bm{x}_t\rangle,
		\label{thm1:corereg}
	\end{align}
	where $\bm{\Delta}=\bm{\widehat{A}}_{\mathrm{LR}}-\bm{A}_\mathrm{LR}^\ast \in\bm{\Xi}_\mathrm{LR}$.
	Define the empirical norm $\norm{\bm{\Delta}}_T$ with $\norm{\bm{\Delta}}_T^2 = \sum_{t=1}^T\norm{\bm{\Delta}\bm{x}_t}_2^2$ $/T$.
	Combining Lemmas \ref{lemma:RSCregression} and \ref{lemma:reg-deviation}, we can show when $T \gtrsim \max(1, \sigma^4)(d_cd_{\mathrm{LR}}+ \sum_{i=1}^{n}\log q_i)$,
	\begin{align}
		\mathbb{P}\left\{
		\frac{1}{16}c_x \norm{\bm{\Delta}}_\Fr^2 \leq \norm{\bm{\Delta}}_T^2
		\leq 2C\kappa \sigma \sqrt{d_cd_{\mathrm{LR}} /T}	\norm{\bm{\Delta}}_\Fr
		\right\}
		\geq 1 - \exp\left\{-c_1d_cd_{\mathrm{LR}}\right\}  \nonumber \\
		-\exp\left\{-c_2d_c^\prime d_\mathrm{LR}^\prime -c_3 \sum_{i=1}^n \log(q_i)\right\},
		\label{eq:thm1-prob}
	\end{align}
	where $d_\mathrm{LR} = R_yR_x + R_y\sum_{i=1}^{n}q_i + R_x\sum_{j=1}^{m}p_j$, $d_c = \log ((m \vee n )(R_y \wedge R_x)R_y^{1/2}R_x^{1/2})$, $d_\mathrm{LR}^\prime = R_yR_x + R_y + R_x\sum_{j=1}^{m}p_j$ and $d_c^\prime = \log \left({m(R_y \wedge R_x)R_y^{1/2}R_x^{1/2}}\right)$.
	Thus, \[\norm{\bm{\Delta}}_{\mathrm{F}}\lesssim \frac{\kappa \sigma}{c_x}\sqrt{\frac{d_{c}d_{\mathrm{LR}}}{T}} \] holds with probability no smaller than that in \eqref{eq:thm1-prob}. This finishes the proof.
\end{proof}

\subsection{Proof of Theorem \ref{thm:autoreg}}
\begin{proof}
	Denote sets $\bm{\Xi}(\underline{q}, P, R_y, R_x) = \{\bm{\Delta} \in \mathbb{R}^{\prod_{i=1}^{n} q_i \times P  \prod_{i=1}^{n} q_i}: \bm{\Delta} = \bm{A}_1 - \bm{A}_2,    \mathcal{M}_n(\bm{A}_1),  $ $\mathcal{M}_n(\bm{A}_2) \in \bm{\Theta}(\underline{q}, P, R_y, R_x)\}$ and $\bm{\Xi}'(\underline{q}, P, R_y, R_x) = \{\bm{\Delta} \in \bm{\Xi}(\underline{q}, P, R_y, R_x): \lVert \bm{\Delta} \rVert_\Fr \leq 1 \}$.
	For simplicity of notations, we abbreviate them to $\bm{\Xi}$ and $\bm{\Xi}^\prime$, respectively, and let $\bm{y}_t \defeq \vectorize(\cm{Y}_t)$, $\bm{x}_t \defeq \vectorize(\cm{X}_t)$, $\bm{e}_t \defeq \vectorize(\cm{E}_t)$, $T_0 \defeq T-P$.
	
	Note that $\bm{\Delta}=[\cm{\widehat{A}}]_n - [\cm{A}^\ast]_n \in\bm{\Xi}$.
	By the optimality of estimator, similar to Theorem \ref{thm:reg}, we could obtain
	\begin{equation}\label{thm-AR}
		\frac{1}{T_0}\sum_{t=P+1}^T\|\bm{\Delta}\bm{x}_{t}\|_2^2\leq
		2\norm{\bm{\Delta}}_{\mathrm{F}}\sup_{\bm{\Delta}\in\bm{\Xi}^\prime}\frac{1}{T_0}\sum_{t=P+1}^T\langle\bm{e}_t,\bm{\Delta}\bm{x}_{t}\rangle.
	\end{equation}
	Define the empirical norm $\norm{\bm{\Delta}}_T$ with $\norm{\bm{\Delta}}_T^2 = \sum_{t=P+1}^T\norm{\bm{\Delta}\bm{x}_{t}}_2^2/{T_0}$.
	Applying Lemmas \ref{lemma:AR-RSC} and \ref{lemma:AR-DB} to \eqref{thm-AR}, we can show when $T_0\gtrsim \max\left(1,\kappa^4\kappa_{U,B}^2/\kappa_{L,A}^2\right)d_{\mathrm{AR}} d_{\mathrm{c}}$,
	\begin{align}
		\mathbb{P}\left\{
		\frac{1}{16}\kappa_{L,A}\norm{\bm{\Delta}}_\Fr^2 \leq \norm{\bm{\Delta}}_T^2
		\leq 2C\kappa^2\kappa_{U,B}'\sqrt{d_{\mathrm{AR}}d_c/T_0} 	\norm{\bm{\Delta}}_\Fr
		\right\}
		\geq 1-c_1\exp\left\{-c_2d_{\mathrm{AR}}d_{c}\right\},
		\label{eq:thm2-prob}
	\end{align}
	where 
	$d_{\mathrm{AR}} = PR_yR_x + (R_y+R_x)\sum_{i=1}^{n}q_i$, {$d_c = \log (nP^{1/2}R_yR_x^{1/2})$},  
	$\kappa^\prime_{U, B}=C_{e}/\mu^{1/2}_{\mathrm{min}}(\bm{B}^{\ast})$, $\kappa_{L, A}=c_{e}/\mu_{\mathrm{max}}(\cm{A}^\ast)$, and $\kappa_{U, B}=C_{e}/\mu_{\mathrm{min}}(\bm{B}^{\ast})$.
	Thus, 
	\[
	\norm{\bm{\Delta}}_{\mathrm{F}}\lesssim
	\frac{\kappa^2\kappa_{U,B}'}{\kappa_{L,A}}
	\sqrt{\frac{d_{\mathrm{AR}}d_c}{T_0}} 
	\]
	holds with probability no smaller than that in \eqref{eq:thm2-prob}. This finishes the proof.
\end{proof}

\subsection{Proof of Theorem \ref{thm:autoreg-lr}}
\begin{proof}
	The proof of this theorem can be easily obtained by replacing the parameter space $\bm{\Xi}(\underline{q}, $ $P, R_y, R_x)$ with $\bm{\Xi}_\mathrm{LR}(\underline{q}, \{P, \underline{q}\}, R_y, R_x)$ in Theorem \ref{thm:autoreg} and Lemmas \ref{lemma:AR-RSC} and \ref{lemma:AR-DB}. Note that an upper bound for the covering number of $\bm{\Xi}_\mathrm{LR}(\underline{q}, \{P, \underline{q}\}, R_y, R_x)$ is given by Lemma \ref{lemma:covering}(c), where we let $\underline{p} = \{P, \underline{q}\}$, leading to  
	\begin{align*}
		\mathcal{N}(\bm{\Xi}'_\mathrm{LR}, \norm{\cdot}_\Fr, \epsilon) \leq   
		\left(1 + \frac{16\sqrt{2}g(n+1)(R_y \wedge R_x)\sqrt{R_yR_x}}{\epsilon} \right)^{4R_yR_x + 2PR_x+ 2(R_y + R_x)\sum_{i=1}^{n}q_i}.
	\end{align*} 
	The rest of the proof shall be the same as that of Theorem \ref{thm:autoreg} and is thus omitted.
\end{proof}

\subsection{Proof of Theorem \ref{thm: ar-sparse}}
\begin{proof}
	Denote sets $\mathcal{H}_{2s}(Q, PQ) = \left\{\bm{A} \in \mathbb{R}^{Q \times PQ}: \norm{\bm{A}}_0 \leq 2s\right\}$, and $\mathcal{H}_{2s}^\prime(Q, PQ) = \{\bm{A} \in \mathcal{H}_{2s}$ $(Q, PQ): \norm{\bm{A}}_\Fr \leq 1\}$, where $Q = \prod_{i=1}^{n}q_i$.
	For simplicity of notations, we abbreviate them to $\mathcal{H}_{2s}$ and $\mathcal{H}_{2s}^\prime$, respectively. Besides, the notations $\bm{\Xi}$, $\bm{\Xi}^\prime$, $\bm{x}_t$, $\bm{e}_t$, and $T_0$ follow the same definition as in Theorem \ref{thm:autoreg}.
	
	By the optimality of $\widetilde{\cm{A}}_\mathrm{L} + \widetilde{\cm{A}}_\mathrm{S}$ and similar to \eqref{thm-AR}, we have
	\begin{equation}
		\begin{split}
			\frac{1}{T_0}\sum_{t=P+1}^T\norm{(\bm{\Delta}_\mathrm{L} + \bm{\Delta}_\mathrm{S})\bm{x}_t}_2^2
			\leq
			& 2\norm{\bm{\Delta}_\mathrm{L}}_{\mathrm{F}}  \sup_{\bm{\Delta}_\mathrm{L}\in\bm{\Xi}^\prime}\frac{1}{T_0}\sum_{t=P+1}^T\langle\bm{e}_t,\bm{\Delta}_\mathrm{L}\bm{x}_t\rangle
			\\
			+ 
			&2\norm{\bm{\Delta}_\mathrm{S}}_{\mathrm{F}}\sup_{\bm{\Delta}_\mathrm{S}\in\mathcal{H}_{2s}^\prime}\frac{1}{T_0}\sum_{t=P+1}^T\langle\bm{e}_t,\bm{\Delta}_\mathrm{S}\bm{x}_t\rangle 
			+ 
			\lambda \left( \lVert{\cm{A}^\ast_\mathrm{S}}\rVert_1 - \lVert{\widetilde{\cm{A}}_\mathrm{S}}\rVert_1\right),
		\end{split}
		\label{eq:ar-sparse-core}
	\end{equation}
	where $\bm{\Delta}_\mathrm{L} = [\widetilde{\cm{A}}_\mathrm{L}]_n-[ \cm{A}^\ast_\mathrm{L}]_n$ and $\bm{\Delta}_\mathrm{S} =[ \widetilde{\cm{A}}_\mathrm{S}]_n - [\cm{A}^\ast_\mathrm{S}]_n$. We first give the lower bound for the left-hand side of the inequality. Note that
	\begin{align*}
		\frac{1}{T_0}\sum_{t=P+1}^T\norm{(\bm{\Delta}_\mathrm{L} + \bm{\Delta}_\mathrm{S})\bm{x}_t}_2^2
		=
		\frac{1}{T_0}\sum_{t=P+1}^T\norm{\bm{\Delta}_\mathrm{L}\bm{x}_t}_2^2 + 
		\frac{1}{T_0}\sum_{t=P+1}^T\norm{\bm{\Delta}_\mathrm{S}\bm{x}_t}_2^2 + 
		\frac{2}{T_0}\sum_{t=P+1}^T \left\langle \bm{\Delta}_\mathrm{L}\bm{x}_t, \bm{\Delta}_\mathrm{S}\bm{x}_t\right\rangle,
	\end{align*}
	where the last term 
	\begin{equation*}
		\begin{split}
			\frac{2}{T_0}\sum_{t=P+1}^T \left\langle \bm{\Delta}_\mathrm{L}\bm{x}_t, \bm{\Delta}_\mathrm{S}\bm{x}_t\right\rangle 
			\geq 
			- \frac{2}{T_0}\sum_{t=P+1}^T \left| \left\langle \bm{\Delta}_\mathrm{L}\bm{x}_t, \bm{\Delta}_\mathrm{S}\bm{x}_t\right\rangle \right|
			\geq 
			- \frac{2}{T_0}\sum_{t=P+1}^T \norm{\bm{\Delta}_\mathrm{L}\bm{x}_t}_2 \norm{\bm{\Delta}_\mathrm{S}\bm{x}_t}_2 \\
			\geq
			- 2 \sqrt{\frac{1}{T_0}\sum_{t=P+1}^T\norm{\bm{\Delta}_\mathrm{L}\bm{x}_t}_2^2} \sqrt{\frac{1}{T_0}\sum_{t=P+1}^T\norm{\bm{\Delta}_\mathrm{S}\bm{x}_t}_2^2}.
		\end{split}
	\end{equation*}
	Define the empirical norm $\norm{\bm{\Delta}_\mathrm{L}}_T$ and $\norm{\bm{\Delta}_\mathrm{S}}_T$ with $\norm{\bm{\Delta}_\mathrm{L}}_T^2 = \sum_{t=P+1}^T\norm{\bm{\Delta}_\mathrm{L}\bm{x}_t}_2^2/T_0$ and $\norm{\bm{\Delta}_\mathrm{S}}_T^2 = \sum_{t=P+1}^T\norm{\bm{\Delta}_\mathrm{S}\bm{x}_t}_2^2/T_0$, respectively.
	Applying the restricted strong convexity results in Lemmas \ref{lemma:AR-RSC} and \ref{lemma:AR-RSC-sparse}, we have
	\begin{align}
		\frac{1}{T_0}\sum_{t=P+1}^T\norm{(\bm{\Delta}_\mathrm{L} + \bm{\Delta}_\mathrm{S})\bm{x}_t}_2^2
		\geq &
		\frac{1}{16}\kappa_{L, A}\norm{\bm{\Delta}_\mathrm{L}}_\Fr^2 + \frac{1}{16}\kappa_{L, A}\norm{\bm{\Delta}_\mathrm{S}}_\Fr^2
		-8 \kappa_{U, A}\norm{\bm{\Delta}_\mathrm{L}}_\Fr \norm{\bm{\Delta}_\mathrm{S}}_\Fr \nonumber \\
		\geq & \frac{1}{16}\kappa_{L, A}\norm{\bm{\Delta}_\mathrm{L}}_\Fr^2 + \frac{1}{16}\kappa_{L, A}\norm{\bm{\Delta}_\mathrm{S}}_\Fr^2
		- 16\kappa_{U, A} \frac{\alpha_\mathrm{L}}{\sqrt{PQ^2}} \norm{\bm{\Delta}_\mathrm{S}}_\Fr, 
		\label{eq:ar-sparse-lower}
	\end{align}
	with probability at least
	$1-2\exp\left\{-c_1d_{\mathrm{AR}}d_c\right\} -2\exp\left\{-c_2s\log(PQ^2)\right\}$, where $d_{\mathrm{AR}} = PR_yR_x + (R_y+R_x)\sum_{i=1}^{n}q_i$ and {$d_c = \log (nP^{1/2}R_yR_x^{1/2})$}.
	The last inequality uses the fact that $\norm{\bm{\Delta}_\mathrm{L}}_\Fr \leq \sqrt{PQ^2} \norm{\bm{\Delta}_\mathrm{L}}_\infty \leq \sqrt{PQ^2} \left(\lVert\widetilde{\cm{A}}_\mathrm{L}\rVert_\infty + \lVert{\cm{A}^\ast_\mathrm{L}\rVert}_\infty\right)$ and Assumption \ref{assump:identifiability}. 
	
	On the other hand, applying the deviation bounds in Lemmas \ref{lemma:AR-DB} and \ref{lemma:ar-deviation-sparse}, the right-hand side of \eqref{eq:ar-sparse-core} can be upper bounded by 
	\begin{align*}
		2C_1\kappa^2\kappa_{U,B}'\sqrt{d_{\mathrm{AR}}d_c/T_0} \norm{\bm{\Delta}_\mathrm{L}}_\Fr + 
		2C_2\kappa^2\kappa_{U,B}'
		\sqrt{s\log(PQ^2)/T_0} \norm{\bm{\Delta}_\mathrm{S}}_\Fr +
		\lambda \left( \lVert{\cm{A}^\ast_\mathrm{S}}\rVert_1 - \lVert{\widetilde{\cm{A}}_\mathrm{S}}\rVert_1\right).
	\end{align*}
	Note that the last term
	\begin{align*}
		\lambda \left( \lVert{\cm{A}^\ast_\mathrm{S}}\rVert_1 - \lVert{\widetilde{\cm{A}}_\mathrm{S}}\rVert_1\right)
		= &
		\lambda \left( \norm{(\bm{A}^\ast_{\mathrm{S}})_S}_1 - \norm{(\bm{A}^\ast_{\mathrm{S}})_S + (\bm{\Delta}_\mathrm{S})_S}_1 - \norm{(\bm{\Delta}_\mathrm{S})_{S^c}}_1 \right)\\
		\leq &
		\lambda \left(
		\norm{(\bm{\Delta}_\mathrm{S})_S}_1 - \norm{(\bm{\Delta}_\mathrm{S})_{S^c}}_1
		\right) \\
		\leq & 
		\lambda \sqrt{2s} \norm{\bm{\Delta}_\mathrm{S}}_\Fr,
	\end{align*}
	where $\bm{A}^\ast_{\mathrm{S}} = [\cm{A}^\ast_{\mathrm{S}}]_n$, $\norm{(\bm{\Delta}_\mathrm{S})_S}_1$ and $\norm{(\bm{\Delta}_\mathrm{S})_{S^c}}_1$ stand for the $l_1$-norm of $\bm{\Delta}_\mathrm{S}$ on and off the support $S$, respectively.
	
	Let $\lambda = 2C_2\kappa^2\kappa_{U,B}'
	\sqrt{s\log(PQ^2)/T_0} + 16\kappa_{U, A} \alpha_\mathrm{L} / \sqrt{PQ^2}$. Combining with \eqref{eq:ar-sparse-core} and \eqref{eq:ar-sparse-lower}, we have
	\begin{align}
		\frac{1}{16}\kappa_{L, A}\norm{\bm{\Delta}_\mathrm{L}}_\Fr^2 + \frac{1}{16}\kappa_{L, A}\norm{\bm{\Delta}_\mathrm{S}}_\Fr^2 
		& \leq 2C_1\kappa^2\kappa_{U,B}'\sqrt{d_{\mathrm{AR}}d_c/T_0}  \norm{\bm{\Delta}_\mathrm{L}}_\Fr + 2\lambda \sqrt{2s} \norm{\bm{\Delta}_\mathrm{S}}_\Fr  \nonumber \\
		\Longrightarrow \quad 
		\norm{\bm{\Delta}_\mathrm{L}}_\Fr^2 + \norm{\bm{\Delta}_\mathrm{S}}_\Fr^2 
		& \lesssim 
		\frac{\kappa^4 (\kappa_{U,B}')^2}{\kappa_{L, A}^2} \cdot \frac{d_cd_{\mathrm{AR}}}{T_0}+ 
		\frac{\kappa^4 (\kappa_{U,B}')^2}{\kappa_{L, A}^2}  \cdot \frac{s\log(PQ^2)}{T_0}+ 
		\frac{\kappa_{U, A}^2 }{\kappa_{L, A}^2} \cdot \frac{s \alpha_\mathrm{L}^2}{PQ^2}, \label{eq:ar-sprase-result}
	\end{align}
	which holds with probability at least $1-3\exp\left\{-c_3d_{\mathrm{AR}}d_c\right\} -3\exp\left\{-c_4s\log(PQ^2)\right\}$,
	when $T_0 \gtrsim \max \left(1,\kappa^4\kappa_{U,B}^2/\kappa_{L,A}^2\right) \left(d_{\mathrm{AR}} d_{\mathrm{c}} + s\log(PQ^2) \right)$, where $\kappa_{U, A}=C_{e}/\mu_{\mathrm{min}}(\cm{A}^\ast)$, $\kappa_{L, A}=c_{e}/$ $\mu_{\mathrm{max}}(\cm{A}^\ast)$, $\kappa_{U, B}=C_{e}/\mu_{\mathrm{min}}(\bm{B}^{\ast})$, $\kappa_{U, B}^\prime=C_{e}/{\mu^{1/2}_{\mathrm{min}}(\bm{B}^{\ast})}$.
\end{proof}

\section{Auxiliary Lemmas}  \label{appendix:sec-lemma}
This section provides the auxiliary lemmas that are used in establishing the theorems in Section \ref{appendix:sec-theory} and their proofs. We first introduce the concepts of Gaussian width and Talagrand's $\gamma_{\alpha}$ functional, which will be frequently utilized in the upcoming lemmas.
\begin{definition}[Gaussian width, Definition 7.5.1 of \cite{vershynin2019high}]
	The Gaussian width of a subset $T \subset \mathbb{R}^n$ is defined as 
	\[
	\omega(T) \defeq \mathbb{E} \sup_{\bm{x} \in T} \langle \bm{g}, \bm{x} \rangle 
	\quad \text{where} \quad
	\bm{g} \sim N(\bm{0}, \bm{I}_n).
	\]
\end{definition}

\begin{definition}[Talagrand's $\gamma_{\alpha}$ functional, Definition 8.5.1 of \cite{vershynin2019high}]
	Let $(T, d)$ be a metric space. A sequence of subsets $\left(T_k\right)_{k=0}^\infty$ of $T$ is called an admissible sequence if the cardinalities of $T_k$ satisfy $|T_0| = 1$ and $|T_k| \leq 2^{2^k}$ for all $k \geq 1$. The $\gamma_\alpha$ functional of $T$ is defined as 
	\[
	\gamma_\alpha(T,d) = \inf_{(T_k)}\sup_{t \in T} \sum_{k=0}^{\infty} 2^{k/\alpha}d(t, T_k)
	\]
	where the infimum is with respect to all admissible sequences.
\end{definition}

\begin{lemma}[\textbf{Covering number}] \label{lemma:covering}
	This lemma derives the covering number for a series of parameter space that are involved in the subsequent proof.
	\begin{itemize}
		\item[(a)] The cardinality of the $\epsilon$-covering net of $\mathcal{S}(\underline{d}, R)$, denoted as $\mathcal{N}(\mathcal{S}, \norm{\cdot}_\Fr, \epsilon)$, satisfies
		\begin{align*}
			\mathcal{N}(\mathcal{S}, \norm{\cdot}_\Fr, \epsilon) \leq 
			\left(1 + \frac{2n \sqrt{R}}{\epsilon}\right)^{R \sum_{i=1}^{n}d_i}.
		\end{align*}
		
		\item[(b)] Denote $\bm{\Xi}(\underline{q}, P, R_y, R_x) = \{\bm{\Delta} \in \mathbb{R}^{\prod_{i=1}^{n} q_i \times P  \prod_{i=1}^{n} q_i}: \bm{\Delta} = \bm{A}_1 - \bm{A}_2,    \mathcal{M}_n(\bm{A}_1),  \mathcal{M}_n(\bm{A}_2) \in \bm{\Theta}(\underline{q}, P, R_y, R_x)\}$ and $\bm{\Xi}'(\underline{q}, P, R_y, R_x) = \{\bm{\Delta} \in \bm{\Xi}(\underline{q}, P, R_y, R_x): \lVert \bm{\Delta} \rVert_\Fr \leq 1 \}$. The cardinality of the $\epsilon$-covering net of $\bm{\Xi}'(\underline{q}, P, R_y, R_x)$, denoted as $\mathcal{N}(\bm{\Xi}', \norm{\cdot}_\Fr, \epsilon)$, satisfies
		\begin{align*}
			\mathcal{N}(\bm{\Xi}', \norm{\cdot}_\Fr, \epsilon) \leq   
			\left(1 + \frac{16\sqrt{2}ng^\prime\sqrt{PR_yR_x}}{\epsilon} \right)^{4PR_yR_x + 2(R_y+R_x)\sum_{i=1}^{n}q_i},
		\end{align*} 
		where $g^\prime =  g\sqrt{\min(R_y, PR_x)}$.
		
		\item[(c)] Denote $\bm{\Xi}_\mathrm{LR}(\underline{q}, \underline{p}, R_y, R_x) = \{\bm{\Delta} \in \mathbb{R}^{\prod_{i=1}^{n} q_i \times \prod_{j=1}^{m} p_j}: \bm{\Delta} = \bm{A}_1 - \bm{A}_2,    \mathcal{M}_n(\bm{A}_1),  \mathcal{M}_n(\bm{A}_2) \in \bm{\Theta}_\mathrm{LR}(\underline{q}, \underline{p}, R_y, R_x)\}$ and $\bm{\Xi}'_{\mathrm{LR}}(\underline{q}, \underline{p}, R_y, R_x) = \{\bm{\Delta} \in \bm{\Xi}_\mathrm{LR}(\underline{q}, \underline{p}, R_y, R_x): \lVert \bm{\Delta} \rVert_\Fr \leq 1 \}$. The cardinality of the $\epsilon$-covering net of $\bm{\Xi}'_{\mathrm{LR}}(\underline{q}, \underline{p}, R_y, R_x)$, denoted as $\mathcal{N}(\bm{\Xi}'_\mathrm{LR}, \norm{\cdot}_\Fr, \epsilon)$, satisfies
		\begin{align*}
			\mathcal{N}(\bm{\Xi}'_\mathrm{LR}, \norm{\cdot}_\Fr, \epsilon) \leq   
			\left(1 + \frac{16\sqrt{2}(m \vee n )g^\prime\sqrt{R_yR_x}}{\epsilon} \right)^{4R_yR_x + 2R_y\sum_{i=1}^{n}q_i + 2R_x\sum_{j=1}^{m}p_j},
		\end{align*} 
		where $g^\prime =  g\sqrt{\min(R_y, R_x)}$ and $m \vee n$ stands for the maximum between $m$ and $n$.
		
		\item[(d)] Denote $ \bm{\Xi}_\mathrm{LR}^j(\underline{q}, \underline{p}, R_y, R_x) = \{\bm{\delta} \in \mathbb{R}^{1 \times \prod_{j=1}^{m}p_j}: \bm{\delta}= \bm{\Lambda}_{y,1}^j \bm{G}_1 \bm{\Lambda}_{x,1}^\top - \bm{\Lambda}_{y,2}^j \bm{G}_2 \bm{\Lambda}_{x,2}^\top,   \bm{\Lambda}_{y,k}^j \in \mathbb{R}^{1 \times R_y}, \bm{G}_k \in \mathbb{R}^{R_y \times R_x},  \bm{\Lambda}_{x,k} \in  \mathcal{S}(\underline{p}, R_x), \text{ for } k \in \{1,2\}, \norm{\bm{\delta}}_2 = \norm{\bm{\Delta}^j}_2 \}$,
		and further,
		${\bm{\Xi}_\mathrm{LR}^j}^\prime(\underline{q}, \underline{p}, R_y, R_x) = \{\bm{\delta} \in \bm{\Xi}_\mathrm{LR}^j(\underline{q}, \underline{p},R_y, R_x): \norm{\bm{\delta}}_2 \leq 1\}$. Then the cardinality of the $\epsilon$-covering net of ${\bm{\Xi}_\mathrm{LR}^j}^\prime(\underline{q},\underline{p}, R_y, R_x)$, denoted as $\mathcal{N}({\bm{\Xi}_\mathrm{LR}^j}^\prime, \norm{\cdot}_2, \epsilon)$, satisfies
		
		\[
		\mathcal{N}({\bm{\Xi}_\mathrm{LR}^j}^\prime, \norm{\cdot}_2, \epsilon) \leq 
		\left(1 + \frac{12\sqrt{2}mg^\prime\sqrt{R_yR_x}}{\epsilon} \right)^{4R_yR_x + 2R_y + 2R_x\sum_{j=1}^{m}p_j},
		\]
		where $g^\prime$ is defined same as in Lemma \ref{lemma:covering}(c).
		
		\item[(e)] Denote the parameter space of a sparse matrix as $\mathcal{H}_s(d_1, d_2) = \left\{\bm{A} \in \mathbb{R}^{d_1 \times d_2}: \norm{\bm{A}}_0 \leq s\right\}$ and $\mathcal{H}_s^\prime(d_1, d_2)  = \left\{\bm{A} \in \mathcal{H}_s(d_1, d_2) : \norm{\bm{A}}_\Fr \leq 1\right\}$. The cardinality of the $\epsilon$-covering net of $\mathcal{H}_s^\prime(d_1, d_2) $, denoted as $\mathcal{N}(\mathcal{H}_s^\prime, \norm{\cdot}_\Fr, \epsilon)$, satisfies 
		\[
		\mathcal{N}(\mathcal{H}_s^\prime, \norm{\cdot}_\Fr, \epsilon) \leq {d_1d_2 \choose s}\left(\frac{3}{\epsilon}\right)^s.
		\]
		Further, for any $\bm{\Delta}_a$, $\bm{\Delta}_b \in \mathcal{H}_s$, there exist $\bm{\Delta}_k$, $\bm{\Delta}_l \in \mathcal{H}_s$ such that $\bm{\Delta}_a - \bm{\Delta}_b = \bm{\Delta}_k + \bm{\Delta}_l$ and $\left\langle \bm{\Delta}_k, \bm{\Delta}_l \right\rangle = 0$. Consequently, $\norm{\bm{\Delta}_k}_\Fr + \norm{\bm{\Delta}_l}_\Fr \leq \sqrt{2} \norm{\bm{\Delta}_a - \bm{\Delta}_b}_\Fr$.
	\end{itemize}
\end{lemma}

\begin{proof}
	(a) 
	Denote the radius-$\nu$ ball by $\mathcal{B}_2^d(\nu) = \{\bm{b} \in \mathbb{R}^d: \lVert \bm{b} \rVert_2 \leq \nu\}$. Set the $\epsilon$-covering net to be $\bar{\mathcal{S}}(\underline{d}, R) = \{ \bar{\bm{\Lambda}} \in \mathbb{R}^{\prod_{i=1}^{n}d_i \times R}: \bar{\bm{\Lambda}}  = \bar{\bm{U}}_n \odot \cdots \odot \bar{\bm{U}}_1 , \bar{\bm{u}}^{(i)}_r \in \overline{ \mathcal{B}}_2^{d_i}(1), \text{ for } 1 \leq r \leq R, 1 \leq i \leq n \}$, where 
	$\bar{\bm{u}}^{(i)}_r$ is the $r$-th column of the matrix $\bar{\bm{U}}_i \in \mathbb{R}^{d_i \times R}$,
	$\overline{ \mathcal{B}}_2^{d_i}(1)$ is a $\epsilon / (n\sqrt{R})$-net for $\mathcal{B}_2^{d_i}(1)$ with 
	$|\bar{\mathcal{B}}_2^{d_i}| = (1 + 2n\sqrt{R} / \epsilon)^{d_i}$. 
	It is sufficient to show that for any $\bm{\Lambda} \in \mathcal{S}(\underline{d}, R)$, there exists a $\bar{\bm{\Lambda}} \in \bar{\mathcal{S}}(\underline{d}, R)$ such that $\lVert \bm{\Lambda}  - \bar{\bm{\Lambda}} \rVert_\Fr \leq \epsilon$. Let $\bm{\Lambda} (r)$ refer to the $r$-th column of $\bm{\Lambda}$. Since
	\begin{align*}
		& \lVert  \bm{\Lambda}(r)  - \bar{\bm{\Lambda}}(r) \rVert_2 \\
		= &  \lVert \text{vec}(\bm{u}^{(1)}_r  \circ \cdots \circ \bm{u}^{(n)}_r  -  \bar{\bm{u}}^{(1)}_r \circ \cdots \circ \bar{\bm{u}}^{(n)}_r) \rVert_2 \\
		\leq & \lVert (\bm{u}^{(1)}_r -  \bar{\bm{u}}^{(1)}_r)   \circ \bm{u}^{(2)}_r \circ  \cdots \circ \bm{u}^{(n)}_r \rVert_\Fr
		+ \cdots +
		\lVert  \bar{\bm{u}}^{(1)}_r   \circ  \cdots \circ \bar{\bm{u}}^{(n-1)}_r  \circ (\bm{u}^{(n)}_r -  \bar{\bm{u}}^{(n)}_r)\rVert_\Fr \\
		\leq & \sum_{i=1}^n \lVert \bm{u}^{(i)}_r -  \bar{\bm{u}}^{(i)}_r \rVert_2 
		\leq \epsilon / \sqrt{R},
	\end{align*} 
	we have 
	\begin{align*}
		\lVert \bm{\Lambda}  - \bar{\bm{\Lambda}} \rVert_\Fr
		= \sqrt{\sum_{r=1}^R \lVert  \bm{\Lambda}(r)  - \bar{\bm{\Lambda}}(r) \rVert_2^2  } 
		\leq \sqrt{R} \max_r \lVert  \bm{\Lambda}(r)  - \bar{\bm{\Lambda}}(r) \rVert_2
		\leq \epsilon, 
	\end{align*}
	and 
	$|\bar{\mathcal{S}}(\underline{d}, R)| = (1 + 2n \sqrt{R}/\epsilon)^{R\times \sum_{i=1}^{n}d_i}$. 
	
	(b) 
	Note that $\bm{\Delta}$ is of the form
	\begin{align*}
		\bm{\Delta} =  \bm{\Lambda}_{y,1} \bm{G}_1 (\bm{I}_P \otimes \bm{\Lambda}_{x,1}^\top) -  \bm{\Lambda}_{y,2} \bm{G}_2 (\bm{I}_P \otimes \bm{\Lambda}_{x,2}^\top) 
		=
		\underbrace{
			\begin{bmatrix}
				\bm{\Lambda}_{y,1} & \bm{\Lambda}_{y,2}
			\end{bmatrix}
		}_{\bm{\Lambda}_y}
		\underbrace{
			\begin{bmatrix}
				\bm{G}_1 & 0 \\
				0 & -\bm{G}_2
			\end{bmatrix}
		}_{\bm{G}}
		\underbrace{
			\begin{bmatrix}
				\bm{I}_P \otimes \bm{\Lambda}_{x,1}^\top \\ \bm{I}_P \otimes \bm{\Lambda}_{x,2}^\top
			\end{bmatrix}
		}_{\bm{\Lambda}_x^\top}.
	\end{align*}
	By Assumption \ref{assump:core}, we have $\norm{\bm{G}_1}_\Fr \leq \sqrt{\min(R_y, PR_x)} \norm{\bm{G}_1}_\op \leq g\sqrt{\min(R_y, PR_x)}$ and same for $\bm{G}_2$. Define $g^\prime \defeq Pg\sqrt{\min(R_y, PR_x)}$.
	Then the covering net can be constructed to be
	\begin{equation}
		\label{eq:covering-net}
		\begin{split}
			\bar{\bm{\Xi}}(\underline{q}, P, R_y, R_x) = \{\bar{\bm{\Lambda}}_y \bar{\bm{G}}   
			\underbrace{
				\begin{bmatrix}
					\bm{I}_P \otimes \bar{\bm{\Lambda}}_{x,1}^\top \\ \bm{I}_P \otimes \bar{\bm{\Lambda}}_{x,2}^\top
				\end{bmatrix}
			}_{\bar{\bm{\Lambda}}_x^\top},
			\bar{\bm{\Lambda}}_y \in \bar{\mathcal{S}}(\underline{q}, 2R_y), 
			\vectorize(\bar{\bm{G}}) \in\bar{ \mathcal{B}}_2^{4PR_yR_x}(\sqrt{2}g^\prime), \\
			\bar{\bm{\Lambda}}_{x,1}, \bar{\bm{\Lambda}}_{x,2} \in \bar{\mathcal{S}}(\underline{q}, R_x)
			\}
		\end{split}
	\end{equation}
	where $\bar{\mathcal{S}}(\underline{q}, 2R_y)$ is an $\epsilon/(8g^\prime\sqrt{PR_x})$-covering net for ${\mathcal{S}}(\underline{q}, 2R_y)$ with $|\bar{\mathcal{S}}(\underline{q}, 2R_y)| \leq (1 +16\sqrt{2}n g^\prime$ $\sqrt{PR_yR_x} / \epsilon)^{2R_y\sum_{i=1}^{n}q_i}$ by Lemma \ref{lemma:covering}(a), $\bar{\mathcal{S}}(\underline{q}, R_x)$ is an $\epsilon/(8g^\prime\sqrt{PR_y})$-covering net for ${\mathcal{S}}(\underline{q}, R_x)$ with $|\bar{\mathcal{S}}(\underline{q}, R_x)| \leq (1 +16n g^\prime$$\sqrt{PR_yR_x} / \epsilon)^{R_x\sum_{i=1}^{n}q_i}$ by Lemma \ref{lemma:covering}(a).
	Moreover, 
	$\bar{ \mathcal{B}}_2^{4PR_yR_x}(\sqrt{2}g^\prime)$ is an $\epsilon / (8\sqrt{PR_yR_x})$-covering net for ${ \mathcal{B}}_2^{4PR_yR_x}(\sqrt{2}g^\prime)$ with $|\bar{ \mathcal{B}}_2^{4PR_yR_x}(\sqrt{2}g^\prime)| \leq (1 +16\sqrt{2}g^\prime\sqrt{PR_yR_x} / \epsilon)^{4PR_yR_x}$.
	
	Then we have for any $\bm{\Delta} \in  \bm{\Xi}^\prime(\underline{q}, P, R_y, R_x)$, there exists $\bar{\bm{\Delta}} \in \bar{\bm{\Xi}}^\prime(\underline{q}, P, R_y, R_x)$ such that 
	\begin{align*}
		\norm{\bm{\Delta} - \bar{\bm{\Delta}}}_\Fr 
		&  = \norm{(\bm{\Lambda}_y - \bar{\bm{\Lambda}}_y) \bm{G} \bm{\Lambda}_x^\top
			+ \bar{\bm{\Lambda}}_y(\bm{G} - \bar{\bm{G}}) \bm{\Lambda}_x^\top
			+ \bar{\bm{\Lambda}}_y \bar{\bm{G}} (\bm{\Lambda}_x^\top - \bar{\bm{\Lambda}}_x^\top) 
		}_\Fr \\
		& \leq \norm{\bm{\Lambda}_y - \bar{\bm{\Lambda}}_y}_\Fr \norm{\bm{G}}_\Fr \norm{\bm{\Lambda}_x}_\Fr + 
		\norm{\bar{\bm{\Lambda}}_y}_\Fr \norm{\bm{G} - \bar{\bm{G}}}_\Fr \norm{\bm{\Lambda}_x}_\Fr +
		\norm{\bar{\bm{\Lambda}}_y}_\Fr \norm{\bar{\bm{G}}}_\Fr \norm{\bm{\Lambda}_x- \bar{\bm{\Lambda}}_x}_\Fr  \\
		& \leq 2g^\prime\sqrt{PR_x} \norm{\bm{\Lambda}_y - \bar{\bm{\Lambda}}_y}_\Fr 
		+ 2\sqrt{PR_yR_x}\norm{\bm{G} - \bar{\bm{G}}}_\Fr + 2g^\prime \sqrt{R_y}  \norm{\bm{\Lambda}_x - \bar{\bm{\Lambda}}_x}_\Fr \\
		& \leq 2g^\prime\sqrt{PR_x} \norm{\bm{\Lambda}_y - \bar{\bm{\Lambda}}_y}_\Fr 
		+ 2\sqrt{PR_yR_x}\norm{\bm{G} - \bar{\bm{G}}}_\Fr + 
		2g^\prime \sqrt{PR_y} \norm{\bm{\Lambda}_{x,1} -\bar{ \bm{\Lambda}}_{x,1}}_\Fr \\
		& \hspace{10mm}+ 2g^\prime \sqrt{PR_y} \norm{\bm{\Lambda}_{x,2} -\bar{ \bm{\Lambda}}_{x,2}}_\Fr \\
		& \leq \epsilon
		,
	\end{align*}
	where the second inequality uses the facts that 
	\[\norm{\bm{\Lambda}_x}_\Fr = \sqrt{ \norm{\bm{I}_P \otimes \bm{\Lambda}_{x,1}}_\Fr^2 + \norm{\bm{I}_P \otimes \bm{\Lambda}_{x,2}}_\Fr^2}  =
	\sqrt{\norm{\bm{I}_P}_\Fr^2 \norm{\bm{\Lambda}_{x,1}}_\Fr^2 + \norm{\bm{I}_P}_\Fr^2 \norm{\bm{\Lambda}_{x,2}}_\Fr^2} = \sqrt{2PR_x},
	\]
	and $\norm{\bm{\Lambda}_y }_\Fr = \sqrt{2R_y}$ since all of its columns are unit-norm as well as $\norm{\bm{G}}_\Fr \leq \sqrt{2}g^\prime$. The third inequality can be implied from the following: 
	\begin{align*}
		\norm{\bm{\Lambda}_x - \bar{\bm{\Lambda}}_x}_\Fr  &= \sqrt{\norm{\bm{I}_P \otimes (\bm{\Lambda}_{x,1} -\bar{ \bm{\Lambda}}_{x,1})}_\Fr^2 +  \norm{\bm{I}_P \otimes (\bm{\Lambda}_{x,2} -\bar{ \bm{\Lambda}}_{x,2})}_\Fr^2} \\
		& \leq \norm{\bm{I}_P \otimes (\bm{\Lambda}_{x,1} -\bar{ \bm{\Lambda}}_{x,1})}_\Fr +  \norm{\bm{I}_P \otimes (\bm{\Lambda}_{x,2} -\bar{ \bm{\Lambda}}_{x,2})}_\Fr  \\
		& \leq \sqrt{P} ( \norm{\bm{\Lambda}_{x,1} -\bar{ \bm{\Lambda}}_{x,1}}_\Fr  + \norm{\bm{\Lambda}_{x,2} -\bar{ \bm{\Lambda}}_{x,2}}_\Fr).
	\end{align*}
	Hence
	\[
	\begin{split}
		\lvert \bar{\bm{\Xi}}(\underline{q}, P, R_y, R_x) \rvert & \leq |\bar{\mathcal{S}}(\underline{q}, 2R_y)|
		\times |\bar{ \mathcal{B}}_2^{4PR_yR_x}(\sqrt{2}g^\prime)| \times |\bar{\mathcal{S}}(\underline{q}, R_x)|^2
		\\ & \leq \left(1 + \frac{16\sqrt{2}ng^\prime\sqrt{PR_yR_x}}{\epsilon} \right)^{4PR_yR_x + 2(R_y+R_x)\sum_{i=1}^{n}q_i}.
	\end{split}
	\]
	
	(c) The proof of this lemma can be easily obtained by modifying that of Lemma \ref{lemma:covering}(b). Specifically, let $P=1$ and since $\bm{\Lambda}_{x,1}, \bm{\Lambda}_{x,2} \in \mathcal{S}(\underline{p}, R_x)$, we therefore alternatively consider $\bar{\bm{\Lambda}}_{x,1}, \bar{\bm{\Lambda}}_{x,2} \in \bar{\mathcal{S}}(\underline{p}, R_x)$ in \eqref{eq:covering-net}, where $\bar{\mathcal{S}}(\underline{p}, R_x)$ is an $\epsilon/(8g^\prime\sqrt{R_y})$-covering net for ${\mathcal{S}}(\underline{p}, R_x)$ with $|\bar{\mathcal{S}}(\underline{p}, R_x)| \leq (1 +16mg^\prime$$\sqrt{R_yR_x} / \epsilon)^{R_x\sum_{j=1}^{m}p_j}$. The rest of the proof is the same as in Lemma \ref{lemma:covering}(b).
	
	(d) The proof is similar to that of Lemma \ref{lemma:covering}(b) and (c). Note that 
	\[
	\bm{\delta} = 
	\underbrace{
		\begin{bmatrix*}
			\bm{\Lambda}_{y,1}^j  & \bm{\Lambda}_{y,2}^j  
	\end{bmatrix*}}_{\bm{\Lambda}_y^j}
	\underbrace{
		\begin{bmatrix*}
			\bm{G}_1 & 0 \\ 0 & -\bm{G}_2
	\end{bmatrix*}}_{\bm{G}}
	\underbrace{
		\begin{bmatrix}
			\bm{\Lambda}_{x,1}^\top \\ \bm{\Lambda}_{x,2}^\top
		\end{bmatrix}
	}_{\bm{\Lambda}_x^\top}.
	\]
	We can construct the covering net to be
	\begin{equation*}
		\begin{split}
			\bar{\bm{\Xi}}_\mathrm{LR}^j(\underline{q}, \underline{p}, R_y, R_x) = \{\bar{\bm{\Lambda}}_y^j \bar{\bm{G}} {\bar{\bm{\Lambda}}_x^\top},
			\bar{\bm{\Lambda}}_y^j \in \bar{\mathcal{B}}_2^{2R_y}(\sqrt{2R_y}), 
			\vectorize(\bar{\bm{G}}) \in\bar{ \mathcal{B}}_2^{4R_yR_x}(\sqrt{2}g^\prime), 
			\bar{\bm{\Lambda}}_{x} \in \bar{\mathcal{S}}(\underline{p}, 2R_x)
			\},
		\end{split}
	\end{equation*}
	where $\bar{\mathcal{B}}_2^{2R_y}(\sqrt{2R_y})$ is an $\epsilon/(6g^\prime\sqrt{R_x})$-covering net for ${\mathcal{B}}_2^{2R_y}(\sqrt{2R_y})$ with $|\bar{\mathcal{B}}_2^{2R_y}(\sqrt{2R_y})| \leq (1 +12\sqrt{2} g^\prime$ $\sqrt{R_yR_x} / \epsilon)^{2R_y}$, $\bar{\mathcal{S}}(\underline{p}, 2R_x)$ is an $\epsilon/(6g^\prime\sqrt{R_y})$-covering net for ${\mathcal{S}}(\underline{p}, 2R_x)$ with $|\bar{\mathcal{S}}(\underline{p}, 2R_x)| \leq (1 +12\sqrt{2}m g^\prime$$\sqrt{R_yR_x} / \epsilon)^{2R_x\sum_{j=1}^{m}p_j}$ by Lemma \ref{lemma:covering}(a),
	$\bar{ \mathcal{B}}_2^{4R_yR_x}(\sqrt{2}g^\prime)$ is an $\epsilon / (6\sqrt{R_yR_x})$-covering net for ${ \mathcal{B}}_2^{4R_yR_x}(\sqrt{2}g^\prime)$ with $|\bar{ \mathcal{B}}_2^{4R_yR_x}(\sqrt{2}g^\prime)| \leq (1 + 12\sqrt{2}g^\prime\sqrt{R_yR_x} / \epsilon)^{4R_yR_x}$.
	
	It is sufficient to show that for any $\bm{\delta} \in {\bm{\Xi}_\mathrm{LR}^j}^\prime$, there exists a $\bar{\bm{\delta}} \in \bar{\bm{\Xi}}_\mathrm{LR}^j$ such that $\norm{\bm{\delta}  - \bar{\bm{\delta}}}_2 \leq \epsilon$. 
	We verify that
	\begin{align*}
		\norm{\bm{\delta}  - \bar{\bm{\delta}}}_2
		& = 
		\norm{ \bm{\Lambda}_{y}^j  \bm{G} \bm{\Lambda}_x^\top -  \bar{\bm{\Lambda}}_{y}^j  \bar{\bm{G}} \bar{\bm{\Lambda}}_x^\top}_2\\ 
		& \leq 
		\norm{ \bm{\Lambda}_{y}^j -  \bar{\bm{\Lambda}}_{y}^j}_2 
		\norm{\bm{G}}_\Fr \norm{ \bm{\Lambda}_x}_\Fr + 
		\norm{\bar{\bm{\Lambda}}_{y}^j}_2 \norm{\bm{G} - \bar{\bm{G}}}_\Fr \norm{ \bm{\Lambda}_x}_\Fr  +
		\norm{\bar{\bm{\Lambda}}_{y}^j}_2 \norm{\bar{\bm{G}}}_\Fr \norm{ \bm{\Lambda}_x - \bar{\bm{\Lambda}}_x}_\Fr \\
		& \leq  
		2g^\prime\sqrt{R_x} 	\norm{ \bm{\Lambda}_{y}^j -  \bar{\bm{\Lambda}}_{y}^j}_2 +
		2\sqrt{R_yR_x}\norm{\bm{G} - \bar{\bm{G}}}_\Fr + 2g^\prime \sqrt{R_y}  \norm{\bm{\Lambda}_x - \bar{\bm{\Lambda}}_x}_\Fr \\
		& \leq \epsilon,
	\end{align*}
	where the second inequality uses the fact that $\norm{\bm{\Lambda}_{y}^j }_2 \leq \norm{\bm{\Lambda}_{y}}_\Fr \leq \sqrt{2R_y}$. 
	Hence
	\[
	\begin{split}
		\lvert \bar{\bm{\Xi}}_\mathrm{LR}^j(\underline{q}, \underline{p}, R_y, R_x) \rvert & \leq |\bar{\mathcal{B}}_2^{2R_y}(\sqrt{2R_y})|
		\times |\bar{ \mathcal{B}}_2^{4R_yR_x}(\sqrt{2}g^\prime)| \times |\bar{\mathcal{S}}(\underline{p}, 2R_x)|
		\\ & \leq \left(1 + \frac{12\sqrt{2}mg^\prime\sqrt{R_yR_x}}{\epsilon} \right)^{4R_yR_x + 2R_y + 2R_x\sum_{j=1}^{m}p_j}.
	\end{split}
	\]

	(e) The covering number follows trivially from (4.10) of \cite{vershynin2019high} by considering all the possible positions of the $s$ nonzero entries. Besides, for any $\bm{\Delta}_a$, $\bm{\Delta}_b \in \mathcal{H}_s(d_1, d_2) $, we have $\norm{\bm{\Delta}_a - \bm{\Delta}_b}_0 \leq \norm{\bm{\Delta}_a}_0 + \norm{\bm{\Delta}_b}_0 \leq 2s$. Let $\bm{\Delta}_k \in \mathbb{R}^{d_1 \times d_2}$ contain $\lfloor{\norm{\bm{\Delta}_a - \bm{\Delta}_b}_0 / 2}\rfloor$ - nonzero entries with $\bm{\Delta}_k(i,j) = (\bm{\Delta}_a - \bm{\Delta}_b)(i,j)$ for any $\bm{\Delta}_k(i,j) \neq 0$, and $\bm{\Delta}_l = \bm{\Delta}_a - \bm{\Delta}_b - \bm{\Delta}_k$. Then one can easily verify that $\left\langle \bm{\Delta}_k, \bm{\Delta}_l \right\rangle = 0$.
	Consequently, $\norm{\bm{\Delta}_k}_\Fr + \norm{\bm{\Delta}_l}_\Fr \leq 
	\sqrt{2} \sqrt{\norm{\bm{\Delta}_k}_\Fr^2 + \norm{\bm{\Delta}_l}_\Fr ^2} =
	\sqrt{2} \norm{\bm{\Delta}_a - \bm{\Delta}_b}_\Fr$, where the first step results from the Cauchy-Schwarz inequality.
\end{proof}

\begin{lemma}[\textbf{Regression RSC}]\label{lemma:RSCregression}
	Suppose that Assumptions \ref{assump:reg-input}, \ref{assump:reg-error}, \ref{assump:core} hold and $T \gtrsim  \sigma^4 (d_\mathrm{LR}^\prime d_c^\prime + \sum_{i=1}^{n}\log q_i)$, we have	
	\[
	\frac{1}{4}\sqrt{c_x}\norm{\bm{\Delta}}_{\mathrm{F}}\leq\norm{\bm{\Delta}}_T\leq2\sqrt{C_x}\norm{\bm{\Delta}}_{\mathrm{F}},
	\]
	for all $\bm{\Delta}\in \bm{\Xi}_\mathrm{LR}(\underline{q}, \underline{p}, R_y, R_x) $ with probability at least 
	$1-\exp\left\{-c_1d_\mathrm{LR}^\prime d_c^\prime  -c_2 \sum_{i=1}^n \log q_i\right\}$,
	where $d_\mathrm{LR}^\prime = R_yR_x + R_y + R_x\sum_{j=1}^{m}p_j$ and $d_c^\prime = \log \left({m(R_y \wedge R_x)R_y^{1/2}R_x^{1/2}}\right)$, and $c$'s are some positive constants.
\end{lemma}
\begin{proof}
	Note that 
	\[\norm{\bm{\Delta}}_T^2=
	\frac{1}{T}\sum_{t=1}^T\norm{\bm{\Delta}\bm{x}_{t}}_2^2
	= \sum_{j=1}^Q \Big( \frac{1}{T} \sum_{t=1}^T\norm{\bm{\Delta}^j\bm{x}_{t}}_2^2 \Big)
	\defeq \sum_{j=1}^Q \norm{\bm{\Delta}^j}_T^2
	,\]
	where $\bm{\Delta}^j \in \mathbb{R}^{1 \times \prod_{j=1}^{m}p_j}$ denotes the $j$-th row of $\bm{\Delta}$.
	By Lemma \ref{lemma:reg-rsc-general} and Assumption \ref{assump:reg-input}, we have when $T \gtrsim \sigma^4 \omega^2({\bm{\Xi}_\mathrm{LR}^j}^\prime)$,
	\begin{align*}
		\frac{1}{16}c_x\norm{\bm{\Delta}^j}_2^2
		\leq  \inf_{\bm{\Delta}^j\in\bm{\Xi}^j_\mathrm{LR}} \norm{\bm{\Delta}^j}_T^2
		\leq \sup_{\bm{\Delta}^j\in\bm{\Xi}^j_\mathrm{LR}} \norm{\bm{\Delta}^j}_T^2
		\leq 
		4 C_x \norm{\bm{\Delta}^j}_2^2
	\end{align*} 
	holds with probability at least $1 - \exp\left\{ -\eta_1 T / \sigma^4\right\}$, where $\eta_1$ are some positive constants,
	$\bm{\Xi}^j_\mathrm{LR} \defeq \bm{\Xi}_\mathrm{LR}^j(\underline{q}, \underline{p}, R_y, R_x) = \{\bm{\delta} \in \mathbb{R}^{1 \times \prod_{j=1}^{m}p_j}: \bm{\delta}= \bm{\Lambda}_{y,1}^j \bm{G}_1 \bm{\Lambda}_{x,1}^\top - \bm{\Lambda}_{y,2}^j \bm{G}_2 \bm{\Lambda}_{x,2}^\top,   \bm{\Lambda}_{y,k}^j \in \mathbb{R}^{1 \times R_y}, \bm{G}_k \in \mathbb{R}^{R_y \times R_x}, $ $\bm{\Lambda}_{x,k} \in  \mathcal{S}(\underline{p}, R_x), \text{ for } k \in \{1,2\}, \norm{\bm{\delta}}_2 = \norm{\bm{\Delta}^j}_2 \}$,
	and lastly,
	${\bm{\Xi}^j_\mathrm{LR}}^\prime \defeq {\bm{\Xi}_\mathrm{LR}^j}^\prime(\underline{q}, \underline{p}, R_y, R_x) = \{\bm{\delta} \in \bm{\Xi}_\mathrm{LR}^j(\underline{q}, \underline{p}, $ $R_y, R_x): \norm{\bm{\delta}}_2 \leq 1\}$.
	
	Then
	\begin{align}\label{eq:reg-rsc-prob}
		& \mathbb{P}  \left\{ 
		\frac{1}{16}c_x \norm{\bm{\Delta}}_\Fr^2 
		\leq 
		\inf_{\bm{\Delta}\in\bm{\Xi}_\mathrm{LR}} \norm{\bm{\Delta}}_T^2 
		\leq 
		\sup_{\bm{\Delta}\in\bm{\Xi}_\mathrm{LR}} \norm{\bm{\Delta}}_T^2
		\leq 
		4 C_x\norm{\bm{\Delta}}_\Fr^2
		\right\} \nonumber \\ 
		\geq & 	\mathbb{P} \left\{ 
		\frac{1}{16}c_x \cdot \sum_{j=1}^Q \norm{\bm{\Delta}^j}_2^2
		\leq
		\sum_{j=1}^Q  \inf_{\bm{\Delta}^j\in\bm{\Xi}_\mathrm{LR}^j} \norm{\bm{\Delta}^j}_T^2
		\leq 
		\sum_{j=1}^Q  \sup_{\bm{\Delta}^j\in\bm{\Xi}_\mathrm{LR}^j} \norm{\bm{\Delta}^j}_T^2
		\leq 
		4C_x \cdot  \sum_{j=1}^Q \norm{\bm{\Delta}^j}_2^2
		\right\} \nonumber\\
		\geq & \mathbb{P}  \left\{ 
		\bigcap_{j=1}^Q
		\left\{ 
		\frac{1}{16}c_x \norm{\bm{\Delta}^j}_2^2 \leq  \inf_{\bm{\Delta}^j\in\bm{\Xi}_\mathrm{LR}^j} \norm{\bm{\Delta}^j}_T^2
		\leq \sup_{\bm{\Delta}^j\in\bm{\Xi}_\mathrm{LR}^j} \norm{\bm{\Delta}^j}_T^2
		\leq 4C_x \norm{\bm{\Delta}^j}_2^2
		\right\}
		\right\} \nonumber\\
		\geq & 1 - \sum_{j=1}^{Q} \mathbb{P}    \left\{
		\inf_{\bm{\Delta}^j\in\bm{\Xi}_\mathrm{LR}^j} \norm{\bm{\Delta}^j}_T^2 \leq 
		\frac{1}{16}c_x \norm{\bm{\Delta}^j}_2^2,
		\sup_{\bm{\Delta}^j\in\bm{\Xi}_\mathrm{LR}^j} \norm{\bm{\Delta}^j}_T^2
		\geq 4C_x \norm{\bm{\Delta}^j}_2^2
		\right\}
		\nonumber\\
		\geq & 1 - \exp\left\{-\eta_1T / \sigma^4 + \log Q\right\},
	\end{align}
	when $T \gtrsim \sigma^4 \omega^2({\bm{\Xi}^j_\mathrm{LR}}^\prime)$. Note that the sample size requirement here uses the fact that $\omega({\bm{\Xi}_\mathrm{LR}^1}^\prime) = \dots = \omega({\bm{\Xi}_\mathrm{LR}^Q}^\prime)$.
	
	By the Dudley's inequality (see e.g., Theorem 8.1.10 of \cite{vershynin2019high}) and Lemma \ref{lemma:covering}(d), we have
	\begin{align*}
		\omega({\bm{\Xi}^j_\mathrm{LR}}^\prime) 
		& \leq C_1\int_{0}^{2} \sqrt{\log \mathcal{N} \left({\bm{\Xi}^j_\mathrm{LR}}^\prime, \norm{\cdot}_2, \epsilon \right)} \mathrm{d}\epsilon \\
		& \leq C_1 \int_{0}^{2}\left[\left(4R_yR_x + 2R_y + 2R_x\sum_{j=1}^{m}p_j\right)
		\log \left(C_2 \frac{m(R_y \wedge R_x)\sqrt{R_yR_x}}{\epsilon}\right)
		\right]^{1/2} \mathrm{d} \epsilon \\
		& \leq C_3 \left(R_yR_x + R_y + R_x\sum_{j=1}^{m}p_j\right)^{1/2} 
		\left[\log \left({m(R_y \wedge R_x)\sqrt{R_yR_x}}\right)
		\right]^{1/2} .
	\end{align*}
	Denote $d_\mathrm{LR}^\prime = R_yR_x + R_y + R_x\sum_{j=1}^{m}p_j$ and $d_c^\prime = \log \left({m(R_y \wedge R_x)R_y^{1/2}R_x^{1/2}}\right)$.
	Then, when $T \gtrsim \sigma^4 (d_c^\prime d_\mathrm{LR}^\prime + \log Q)$, the following holds:
	\begin{align*}
		\frac{1}{16}c_x \norm{\bm{\Delta}}_\Fr^2 
		\leq 
		\inf_{\bm{\Delta}\in\bm{\Xi}_\mathrm{LR}} \norm{\bm{\Delta}}_T^2 
		\leq 
		\sup_{\bm{\Delta}\in\bm{\Xi}_\mathrm{LR}} \norm{\bm{\Delta}}_T^2
		\leq 
		4 C_x\norm{\bm{\Delta}}_\Fr^2,
	\end{align*}
	with probability at least $1 - \exp\left\{-c_1d_c^\prime d_\mathrm{LR}^\prime - c_2\log Q\right\}$.
\end{proof}

\begin{lemma}[\textbf{Regression deviation bound}]\label{lemma:reg-deviation}
	Suppose that Assumptions \ref{assump:reg-input}, \ref{assump:reg-error}, \ref{assump:core} hold and $T \gtrsim d_cd_{\mathrm{LR}} $, we have
	\[
	\sup_{\bm{\Delta}\in\bm{\Xi}_\mathrm{LR}^\prime}\frac{1}{T}\sum_{t=1}^T\langle\bm{e}_t,\bm{\Delta}\bm{x}_{t}\rangle \leq C\kappa \sigma \sqrt{d_cd_{\mathrm{LR}} /T}	
	\]
	with probability at least $1 - \exp\left\{-cd_cd_{\mathrm{LR}} \right\}$,
	where $d_\mathrm{LR} = R_yR_x + R_y\sum_{i=1}^{n}q_i + R_x\sum_{j=1}^{m}p_j$, $d_c = \log ((m \vee n )(R_y \wedge R_x)R_y^{1/2}R_x^{1/2})$, and $C, c$ are some positive constants.
\end{lemma}
\begin{proof}
	For a fixed matrix $\bm{M} \in \mathbb{R}^{Q \times \prod_{j=1}^{m}p_j}$, one can check that, conditional on $\bm{x}_{t}$, $\langle\bm{e}_t, \bm{M}\bm{x}_t  \rangle$ is a sub-Gaussian random variable with parameter $\kappa^2\norm{\bm{M}\bm{x}_t}_2^2 $.
	Then 
	\begin{align}
		\mathbb{E} \left[ \exp \left(\lambda \sum_{t=1}^{T}\langle \bm{e}_t,\bm{M}\bm{x}_{t}\rangle \right) \right] 
		=& \prod_{t=1}^{T} \mathbb{E} \left[ \exp \left(\lambda \langle \bm{e}_t,\bm{M}\bm{x}_{t}\rangle \right) \right]  
		\label{eq:reg-dev-iid}\\
		=& \prod_{t=1}^{T}
		\mathbb{E}_x \left\{	\mathbb{E}_e \left[ 
		\exp \left(\lambda \langle \bm{e}_t,\bm{M}\bm{x}_{t}\rangle \right) | \bm{x}_t
		\right]\right\}  \nonumber \\
		\leq & 
		\prod_{t=1}^{T}
		\mathbb{E} \left[
		\exp \left( C_1 \lambda^2 \kappa^2 \norm{\bm{M}\bm{x}_t}_2^2 \right)
		\right]  \nonumber \\
		\leq & 
		\exp \left( C_2 \lambda^2 T \kappa^2 \sigma^2 \norm{\bm{M}}_\Fr^2 \right) \label{eq:reg-dev-bound}
	\end{align}
	holds for any $0 < \lambda \leq c / \left(\kappa \sigma \norm{\bm{M}}_\op\right)$,
	where the equality in \eqref{eq:reg-dev-iid} is due to the independence in Assumptions \ref{assump:reg-input} and \ref{assump:reg-error}, and the inequality in \eqref{eq:reg-dev-bound} is a direct consequence of Lemma \ref{lemma:exercise}. 
	For any $\bm{\Delta}_k, \bm{\Delta}_l \in \bm{\Xi}_\mathrm{LR}^\prime$, we thus have \begin{align*}
		& \mathbb{P} \left\{
		\frac{1}{T}\sum_{t=1}^T\langle\bm{e}_t,(\bm{\Delta}_k - \bm{\Delta}_l) \bm{x}_{t}\rangle \geq u 
		\right\} \\ \leq & 
		\frac{\mathbb{E} \left[ \exp \left(\lambda \sum_{t=1}^{T}\langle \bm{e}_t,(\bm{\Delta}_k - \bm{\Delta}_l)\bm{x}_{t} / T\rangle \right) \right] }{\exp(\lambda u)} \\
		\leq & \exp \left(C_2 \lambda^2 \kappa^2 \sigma^2 \norm{\bm{\Delta}_k - \bm{\Delta}_l}_\Fr^2 / T - \lambda u\right) \\
		\leq & \exp \left\{
		- \min \left(
		\frac{u^2}{C_3^2\kappa^2 \sigma^2\norm{\bm{\Delta}_k - \bm{\Delta}_l}_\Fr^2 / T },
		\frac{u}{C_3 \kappa \sigma \norm{\bm{\Delta}_k - \bm{\Delta}_l}_\Fr / T} 
		\right)
		\right\},
	\end{align*}
	where the second inequality uses \eqref{eq:reg-dev-bound} by setting $\bm{M} = (\bm{\Delta}_k - \bm{\Delta}_l) / T$.
	The other direction of the tail bound can be easily verified with the same approach.
	
	By Lemma \ref{lemma:chaining}, we have 
	\begin{equation}
		\begin{split}
			\mathbb{P}\left\{
			\sup_{\bm{\Delta}\in\bm{\Xi}_\mathrm{LR}^\prime}\frac{1}{T}\sum_{t=1}^T\langle\bm{e}_t,\bm{\Delta}\bm{x}_{t}\rangle \geq 
			C_4 \left(
			\gamma_2(\bm{\Xi}_\mathrm{LR}^\prime, d_2) + \gamma_1(\bm{\Xi}_\mathrm{LR}^\prime, d_1) + u \cdot \mathrm{Diam}_2(\bm{\Xi}_\mathrm{LR}^\prime) 
			+ u^2 \cdot \mathrm{Diam}_1(\bm{\Xi}_\mathrm{LR}^\prime)
			\right) 
			\right\} \\
			\leq \exp(-u^2),
		\end{split}
		\label{eq:reg-dev-prob}
	\end{equation}
	where $d_1(\bm{\Delta}_k, \bm{\Delta}_l) = C_3 \kappa \sigma \norm{\bm{\Delta}_k - \bm{\Delta}_l}_\Fr / T$, and $d_2(\bm{\Delta}_k, \bm{\Delta}_l)=C_3 \kappa \sigma \norm{\bm{\Delta}_k - \bm{\Delta}_l}_\Fr / \sqrt{T}$.
	Since by Theorem 8.6.1 of \cite{vershynin2019high} and Eq.(46) and Lemma 2.7 of \cite{melnyk2016estimating}, 
	\begin{align*}
		\gamma_2(\bm{\Xi}_\mathrm{LR}^\prime, d_2) \leq \left(C_3  \kappa \sigma / \sqrt{T}\right) \gamma_2(\bm{\Xi}_\mathrm{LR}^\prime, 	\norm{\cdot}_\Fr) \leq \left(C_5  \kappa \sigma / \sqrt{T}\right) \omega(\bm{\Xi}_\mathrm{LR}^\prime), \\
		\gamma_1(\bm{\Xi}_\mathrm{LR}^\prime, d_1) \leq \left(C_3  \kappa \sigma / T\right) \gamma_1(\bm{\Xi}_\mathrm{LR}^\prime, \norm{\cdot}_\Fr) \leq \left(C_6  \kappa \sigma / T \right) \omega^2(\bm{\Xi}_\mathrm{LR}^\prime), \\
		\mathrm{Diam}_2(\bm{\Xi}_\mathrm{LR}^\prime) =  2 C_3  \kappa \sigma / \sqrt{T}, \quad 
		\mathrm{Diam}_1(\bm{\Xi}_\mathrm{LR}^\prime) = 2C_3  \kappa \sigma / T.
	\end{align*}
	Combined with \eqref{eq:reg-dev-prob}, we have when $T \gtrsim \left( \omega(\bm{\Xi}_\mathrm{LR}^\prime) + u\right)^2 $,
	\begin{align}
		\mathbb{P}\left\{
		\sup_{\bm{\Delta}\in\bm{\Xi}_\mathrm{LR}^\prime}\frac{1}{T}\sum_{t=1}^T\langle\bm{e}_t,\bm{\Delta}\bm{x}_{t}\rangle \geq
		C_7\kappa \sigma \frac{\omega(\bm{\Xi}_\mathrm{LR}^\prime) + u}{\sqrt{T}}
		\right\} \leq \exp(-u^2).
		\label{eq:reg-dev-prob1}
	\end{align}
	
	To bound $\omega(\bm{\Xi}_\mathrm{LR}^\prime)$, we use
	the Dudley's inequality (see e.g., Theorem 8.1.10 of \cite{vershynin2019high}) and Lemma \ref{lemma:covering}(c), which leads to
	\begin{align*}
		& \omega(\bm{\Xi}_\mathrm{LR}^\prime) \\
		\leq &  C_8\int_{0}^{2} \sqrt{\log \	\mathcal{N}(\bm{\Xi}_\mathrm{LR}', \norm{\cdot}_\Fr, \epsilon)}
		\mathrm{d}\epsilon \\
		\leq & C_8 \int_{0}^{2}\left\{\left( 4R_yR_x + 2R_y\sum_{i=1}^{n}q_i + 2R_x\sum_{j=1}^{m}p_j \right)
		\log \left(C_9 \frac{(m \vee n )(R_y \wedge R_x)\sqrt{R_yR_x}}{\epsilon} \right)
		\right\}^{1/2} \mathrm{d} \epsilon \\
		\leq &  C_{10} \left(R_yR_x + R_y\sum_{i=1}^{n}q_i + R_x\sum_{j=1}^{m}p_j\right)^{1/2} 
		\left[ \log \left((m \vee n )(R_y \wedge R_x)\sqrt{R_yR_x} \right) \right]^{1/2}.
	\end{align*}
	Taking $u = \omega(\bm{\Xi}_\mathrm{LR}^\prime)$, $d_\mathrm{LR} = R_yR_x + R_y\sum_{i=1}^{n}q_i + R_x\sum_{j=1}^{m}p_j$ and $d_c = \log ((m \vee n )(R_y \wedge R_x)R_y^{1/2}R_x^{1/2})$, we can infer from \eqref{eq:reg-dev-prob1} that
	\begin{align*}
		\mathbb{P}\left\{
		\sup_{\bm{\Delta}\in\bm{\Xi}_\mathrm{LR}^\prime}\frac{1}{T}\sum_{t=1}^T\langle\bm{e}_t,\bm{\Delta}\bm{x}_{t}\rangle \leq
		C\kappa \sigma \sqrt{d_cd_{\mathrm{LR}} /T}	\right\} \geq 1 - \exp(-cd_cd_{\mathrm{LR}}).
	\end{align*}
\end{proof}

\begin{lemma}[\textbf{Autoregression RSC}]\label{lemma:AR-RSC} 
	Suppose that Assumptions \ref{assump:core}, \ref{assump:stationarity}, \ref{assump:errorauto} hold, and $T-P\gtrsim \kappa^4\kappa_{U,B}^2d_{\mathrm{AR}}d_c /\kappa_{L,A}^2$, we have	
	\[
	\frac{1}{4}\sqrt{\kappa_{L, A}}
	\|\bm{\Delta}\|_{\mathrm{F}}\leq\|\bm{\Delta}\|_T\leq2\sqrt{\kappa_{U, A}}\|\bm{\Delta}\|_{\mathrm{F}},
	\]
	for all {$\bm{\Delta}\in\bm{\Xi}$} with probability at least $1-2\exp\left\{-c d_{\mathrm{AR}}d_c\right\}$, where $d_{\mathrm{AR}} = PR_yR_x + (R_y+R_x)\sum_{i=1}^{n}q_i$, {$d_c = \log (nP^{1/2}R_yR_x^{1/2})$}, $\kappa_{U, A}=C_{e}/\mu_{\mathrm{min}}(\cm{A}^{\ast})$, $\kappa_{L, A}=c_{e}/\mu_{\mathrm{max}}(\cm{A}^\ast)$, $\kappa_{U, B}=C_{e}/\mu_{\mathrm{min}}(\bm{B}^{\ast})$ and $c$ is a positive constant.
\end{lemma}

\begin{proof}
	Recall that $\norm{\bm{\Delta}}_{T}^2 := \sum_{t=P+1}^{T}\norm{\bm{\Delta}\bm{x}_{t}}_2^2/{T_0}$ is the empirical norm.
	It is sufficient to show that $\sup_{\bm{\Delta}\in \bm{\Xi}'}\norm{\bm{\Delta}}_{T} \leq2\sqrt{\kappa_{U, A}}$ and $\inf_{\bm{\Delta}\in \bm{\Xi}'}\norm{\bm{\Delta}}_{T} \geq\sqrt{\kappa_{L, A}} / 4$ hold with high probability. 
	
	Recall the VAR($P$) representation of the autoregressive model as $\bm{y}_t = \bm{A} \bm{x}_t + \bm{e}_t$, where we denote $\bm{y}_t = \vectorize(\cm{Y}_t)$, $\bm{A} = [\cm{A}]_n$, and $\bm{x}_t = \vectorize(\cm{X}_t)$ for convenience. Note that $\bm{A} = (\bm{A}_1, \dots, \bm{A}_P)$ where $\bm{A}_k = [\cm{A}_k]_n \in \mathbb{R}^{Q\times Q}$ and $Q=\prod_{i=1}^n q_i$. The VAR($P$) representation can be further written into an VAR(1) form, $\bm{x}_t = \bm{B} \bm{x}_{t-1}+\bm{\zeta}_t$, and hence the VMA($\infty$) representation of $\bm{x}_t = \sum_{j=0}^{\infty}\bm{B}^{j}\bm{\zeta}_{t-j}$ or $\bm{z}=\bm{P}\bm{e}$, where
	\[
	\bm{B}=
	\begin{bmatrix}
		\bm{A}_1 & \bm{A}_2 & \cdots &  \bm{A}_{P-1} &\bm{A}_P \\
		\bm{I}_Q &\bm{O} & \cdots & \bm{O} & \bm{O}\\
		\bm{O} & \bm{I}_Q & \cdots & \bm{O} & \bm{O} \\
		\vdots & \vdots & \ddots & \vdots & \vdots \\
		\bm{O} & \bm{O} & \cdots & \bm{I}_Q & \bm{O} 
	\end{bmatrix},
	\quad
	\bm{P}=
	\begin{bmatrix}
		\bm{I}_{PQ} & \bm{B} & \bm{B}^2 & \cdots & \bm{B}^{{T_0}-1} & \cdots\\
		\bm{O} & \bm{I}_{PQ} & \bm{B} & \cdots & \bm{B}^{{T_0}-2} &\cdots  \\
		\bm{O} & \bm{O} & \bm{I}_{PQ} & \cdots & \bm{B}^{{T_0}-3} & \cdots\\
		\vdots & \vdots & \vdots & \ddots & \vdots &\cdots\\
		\bm{O} & \bm{O} & \bm{O} & \cdots & \bm{I}_{PQ} & \cdots
	\end{bmatrix},
	\]
	$\bm{\zeta}_t = (\bm{e}_{t-1}^\top, \bm{0}, \dots, \bm{0})^\top \in \mathbb{R}^{PQ}$, $\bm{e} = (\bm{\zeta}_{{T}}^\top, \bm{\zeta}_{{T}-1}^\top, \dots, \bm{\zeta}_{P+1}^\top)^\top \in \mathbb{R}^{PQ{T_0}}$ and $\bm{z} = (\bm{x}_{{T}}^\top, \bm{x}_{{T}-1}^\top, \dots, $ $\bm{x}_{P+1}^\top)^\top \in \mathbb{R}^{PQ{T_0}}$. Moreover, by Assumption \ref{assump:errorauto}, $\bm{e}=\bm{D}\bar{\bm{\xi}}$, where $\bm{D}=\bm{I}_{P{T_0}}\otimes\bm{\Sigma}_e^{1/2}$, $\bar{\bm{\xi}}_t = (\bm{\xi}_t^\top, \bm{0}, \dots, \bm{0})^\top \in \mathbb{R}^{PQ}$, $\bar{\bm{\xi}} = (\bar{\bm{\xi}}_{{T}-1}^\top,  \bar{\bm{\xi}}_{{T}-2}^\top, \dots, \bar{\bm{\xi}}_{P}^\top)^\top \in \mathbb{R}^{PQ{T_0}}$.
	
	To use the conclusion of Lemma \ref{lemma:chaining}, we need to verify that the mixed tail condition is satisfied with $\bm{W}(\bm{\Delta}) = \norm{\bm{\Delta}}_T^2 -\mathbb{E}\norm{\bm{\Delta}}_T^2$. For any $\bm{\Delta}_i, \bm{\Delta}_j\in \bm{\Xi}'$, we have 
	\[\left|\left(\norm{\bm{\Delta}_i}_T^2 -\mathbb{E}\norm{\bm{\Delta}_i}_T^2\right) - \left(\norm{\bm{\Delta}_j}_T^2 -\mathbb{E}\norm{\bm{\Delta}_j}_T^2\right)\right| = \left|\left(\norm{\bm{\Delta}_i}_T^2 - \norm{\bm{\Delta}_j}_T^2\right) -\mathbb{E}\left(\norm{\bm{\Delta}_i}_T^2 - \norm{\bm{\Delta}_j}_T^2\right)\right|.\]
	We can rewrite $\norm{\bm{\Delta}_i}_T^2 - \norm{\bm{\Delta}_j}_T^2$ as 
	\begin{equation}\label{lemma4:def}
		\begin{split}
			\norm{\bm{\Delta}_i}_T^2 - \norm{\bm{\Delta}_j}_T^2
			&=\frac{1}{T_0}\sum_{t=P+1}^T\norm{\bm{\Delta}_i\bm{x}_{t}}_2^2 - \frac{1}{T_0}\sum_{t=P+1}^T\norm{\bm{\Delta}_j\bm{x}_{t}}_2^2 =\frac{1}{T_0}\bm{z}^\top(\bm{I}_{T_0}\otimes\bm{M})\bm{z} \\
			&=\frac{1}{T_0}\bm{e}^\top\bm{P}^\top(\bm{I}_{T_0}\otimes\bm{M})\bm{P}\bm{e}
			=\frac{1}{T_0}\bar{\bm{\xi}}^\top\bm{D}^\top\bm{P}^\top(\bm{I}_{T_0}\otimes\bm{M})\bm{P}\bm{D}\bar{\bm{\xi}}\\
			&=\bar{\bm{\xi}}^\top(\frac{1}{T_0}\bm{D}^\top\bm{P}^\top(\bm{I}_{T_0}\otimes\bm{M})\bm{P}\bm{D})\bar{\bm{\xi}}=\bar{\bm{\xi}}^\top\bm{Q}\bar{\bm{\xi}},
		\end{split}
	\end{equation}
	where $\bm{M} = \bm{\Delta}_i^\top \bm{\Delta}_i - \bm{\Delta}_j^\top \bm{\Delta}_j$ and $\bm{Q} = \bm{D}^\top\bm{P}^\top(\bm{I}_{T_0} \otimes \bm{M}) \bm{P}\bm{D} / T_0$.
	Note that
	\begin{equation*}
		\begin{split}
			\norm{\bm{M}}_{\mathrm{F}}
			&=\norm{\bm{\Delta}_i^\top \bm{\Delta}_i - \bm{\Delta}_j^\top \bm{\Delta}_j}_{\mathrm{F}} =\norm{\bm{\Delta}_i^\top \bm{\Delta}_i - \bm{\Delta}_j^\top \bm{\Delta}_i + \bm{\Delta}_j^\top \bm{\Delta}_i - \bm{\Delta}_j^\top \bm{\Delta}_j}_{\mathrm{F}} \\
			&=\norm{(\bm{\Delta}_i  - \bm{\Delta}_j)^\top \bm{\Delta}_i + \bm{\Delta}_j^\top (\bm{\Delta}_i - \bm{\Delta}_j)}_{\mathrm{F}}
			\leq 2\norm{\bm{\Delta}_i  - \bm{\Delta}_j}_{\mathrm{F}}.
		\end{split}
	\end{equation*}
	We would provide upper bound for $\norm{\bm{Q}}_{\op}$ and $\norm{\bm{Q}}_{\mathrm{F}}^2$.
	\begin{equation}\label{eq:ar-rsc-op}
		\begin{split}
			\norm{\bm{Q}}_{\op}
			&=\norm{\frac{1}{{T_0}}\bm{D}^\top\bm{P}^\top(\bm{I}_{T_0} \otimes \bm{M}) \bm{P}\bm{D}}_{\op}
			\leq\frac{1}{{T_0}}\norm{\bm{D}}_{\op}^2\norm{\bm{P}}_{\op}^2\norm{\bm{I}_{T_0}\otimes\bm{M}}_{\op}\\ 
			&\leq\frac{1}{{T_0}}C_{e}\lambda_{\mathrm{max}}(\bm{P}\bm{P}^\top)\norm{\bm{I}_{T_0}}_{\op}\norm{\bm{M}}_{\mathrm{F}}\leq\frac{2}{{T_0}}C_{e}\lambda_{\mathrm{max}}(\bm{P}\bm{P}^\top)\norm{\bm{\Delta}_i  - \bm{\Delta}_j}_{\mathrm{F}}.\\
			\norm{\bm{Q}}_{\mathrm{F}}^2& = \norm{\frac{1}{{T_0}}\bm{D}^\top\bm{P}^\top(\bm{I}_{T_0} \otimes \bm{M}) \bm{P}\bm{D}}_{\mathrm{F}}^2 \leq \frac{1}{{T_0}^2} \norm{\bm{D}}_{\op}^4\norm{\bm{P}}_{\op}^4 \norm{\bm{I}_{T_0}\otimes\bm{M}}_{\mathrm{F}}^2\\
			&= \frac{1}{{T_0}^2} C_{e}^2\lambda_{\mathrm{max}}(\bm{P}\bm{P}^\top)^2
			\norm{\bm{I}_{T_0}}_{\mathrm{F}}^2  \norm{\bm{M}}_{\mathrm{F}}^2 \leq \frac{4}{{T_0}}C_{e}^2\lambda_{\mathrm{max}}(\bm{P}\bm{P}^\top)^2 \norm{\bm{\Delta}_i  - \bm{\Delta}_j}_{\mathrm{F}}^2.
		\end{split}
	\end{equation}
	
	Next we derive a concentration inequality for $\left|\left(\norm{\bm{\Delta}_i}_T^2 - \norm{\bm{\Delta}_j}_T^2\right) -\mathbb{E}\left(\norm{\bm{\Delta}_i}_T^2 - \norm{\bm{\Delta}_j}_T^2\right)\right|$.	By the Hanson-Wright inequality in Theorem 6.2.1 of \cite{vershynin2019high}, together with \eqref{lemma4:def} and \eqref{eq:ar-rsc-op}, we have for any $u>0$,
	\begin{align}
		&\hspace{15pt}\mathbb{P}\left\{\left|\left(\norm{\bm{\Delta}_i}_T^2 - \norm{\bm{\Delta}_j}_T^2\right) -\mathbb{E}\left(\norm{\bm{\Delta}_i}_T^2 - \norm{\bm{\Delta}_j}_T^2\right)\right|\geq u\right\} \nonumber\\
		&= \mathbb{P}\left\{\left|\bar{\bm{\xi}}^\top\bm{Q}\bar{\bm{\xi}} -\mathbb{E}\left(\bar{\bm{\xi}}^\top\bm{Q}\bar{\bm{\xi}}\right)\right|\geq u\right\} \nonumber\\
		&= \mathbb{P}\left\{\left|\widetilde{\bm{\xi}}^\top\widetilde{\bm{Q}}\widetilde{\bm{\xi}} -\mathbb{E}\left(\widetilde{\bm{\xi}}^\top\widetilde{\bm{Q}}\widetilde{\bm{\xi}}\right)\right|\geq u\right\} \nonumber\\
		&\leq 2\exp\left\{-c\min\left(\frac{u}{\kappa^2\norm{\widetilde{\bm{Q}}}}_{\mathrm{op}},\frac{u^2}{\kappa^4\norm{\widetilde{\bm{Q}}}_{\mathrm{F}}^2}\right)\right\}  \nonumber\\
		&\leq 2\exp\left\{-c\min\left(\frac{u}{2\kappa^2\kappa_{U, B}\norm{\bm{\Delta}_i-\bm{\Delta}_j}_{\mathrm{F}}/{T_0}},\frac{u^2}{\left(2\kappa^2\kappa_{U, B}\norm{\bm{\Delta}_i-\bm{\Delta}_j}_{\mathrm{F}}/\sqrt{{T_0}}\right)^2}\right)\right\}, \label{eq:ar-rsc-hansonwright}
	\end{align}
	where $\widetilde{\bm{\xi}}$ is the non-zero part of $\bar{\bm{\xi}}$ with independent sub-gaussian entries, $\widetilde{\bm{Q}}$ is the corresponding submatrix of $\bm{Q}$, and $\kappa_{U,B}=C_{e}\lambda_{\mathrm{max}}(\bm{P}\bm{P}^\top)$. Note that the last inequality uses the fact that $\norm{\widetilde{\bm{Q}}}_{\mathrm{op}} \lesssim \norm{\bm{Q}}_\op$ and $\norm{\widetilde{\bm{Q}}}_{\Fr} \leq \norm{\bm{Q}}_\Fr$.
	
	This concentration inequality satisfies the condition of Lemma \ref{lemma:chaining} with $d_1(\bm{\Delta}_i,\bm{\Delta}_j) = 2c^\prime\kappa^2\kappa_{U, B}$ $\norm{\bm{\Delta}_i-\bm{\Delta}_j}_{\mathrm{F}}/{T_0}$ and $d_2(\bm{\Delta}_i, \bm{\Delta}_j) = 2c^{\prime\prime}\kappa^2\kappa_{U, B}\norm{\bm{\Delta}_i-\bm{\Delta}_j}_{\mathrm{F}}/\sqrt{{T_0}}$. We have the following concentration bound for the empirical norm of $\bm{\Delta}$ by directly applying Lemma \ref{lemma:chaining},
	\begin{equation}\label{AR-concentration}
		\begin{split}
			\mathbb{P}\left\{\sup_{\bm{\Delta}\in\bm{\Xi}'}\left|\norm{\bm{\Delta}}_T^2 -\mathbb{E}\norm{\bm{\Delta}}_T^2\right| > C\left(\gamma_2(\bm{\Xi}',d_2) + \gamma_1(\bm{\Xi}',d_1) + u\mathrm{Diam}_2({\bm{\Xi}'}) + u^2\mathrm{Diam}_1(\bm{\Xi}')\right)\right\} \\ \leq 2\exp(-u^2).
		\end{split}
	\end{equation}
	Note that $\omega(\bm{\Xi}')$ is the gaussian-width of $\bm{\Xi}'$. By Theorem 8.6.1 of \cite{vershynin2019high} and Eq.(46) and Lemma 2.7 of \cite{melnyk2016estimating}, we have
	\begin{equation*}\label{RSC-gamma-fcn}
		\begin{split}
			\gamma_2(\bm{\Xi}',d_2)
			&=\left(2c^{\prime\prime}\kappa^2\kappa_{U,B}/\sqrt{{T_0}}\right)\cdot\gamma_2(\bm{\Xi}',\norm{\cdot}_{\mathrm{F}})
			\leq\left(c_2\kappa^2\kappa_{U,B}/\sqrt{{T_0}}\right)\cdot\omega(\bm{\Xi}'),\\
			\gamma_1(\bm{\Xi}',d_1)
			&=\left(2c^\prime\kappa^2\kappa_{U,B}/{T_0}\right)\cdot\gamma_1(\bm{\Xi}',\norm{\cdot}_{\mathrm{F}}) \leq\left(c_1\kappa^2\kappa_{U,B}/{T_0}\right)\cdot\omega^2(\bm{\Xi}'),\\
			\mathrm{Diam}_2({\bm{\Xi}'})
			&= 4c^{\prime\prime}\kappa^2\kappa_{U, B}/\sqrt{{T_0}},\quad
			\mathrm{Diam}_1(\bm{\Xi}') = 4c^\prime\kappa^2\kappa_{U, B}/{T_0}.
		\end{split}
	\end{equation*}
	Thus we can rewrite \eqref{AR-concentration} as
	\begin{equation*}
		\mathbb{P}\left\{\sup_{\bm{\Delta}\in\bm{\Xi}'}\left|\norm{\bm{\Delta}}_T^2 -\mathbb{E}\norm{\bm{\Delta}}_T^2\right| > c_2^\prime\kappa^2\kappa_{U,B}\frac{u+\omega(\bm{\Xi}')}{\sqrt{{T_0}}} + c_1^\prime\kappa^2\kappa_{U,B}\frac{u^2+\omega^2(\bm{\Xi}')}{{T_0}} \right\}\notag \leq 2\exp(-u^2).
	\end{equation*}
	When ${T_0}\gtrsim \left(u+\omega(\bm{\Xi}')\right)^2$, the inequality can be reduced to
	\begin{equation}\label{AR-RSC-concentration}
		\mathbb{P}\left\{\sup_{\bm{\Delta}\in\bm{\Xi}'}\left|\norm{\bm{\Delta}}_T^2 -\mathbb{E}\norm{\bm{\Delta}}_T^2\right| > c_3\kappa^2\kappa_{U,B}\frac{u+\omega(\bm{\Xi}')}{\sqrt{{T_0}}}\right\} \leq 2\exp(-u^2).
	\end{equation}
	
	We then provide an upper and lower bound for $\mathbb{E}\norm{\bm{\Delta}}_T^2$. By the spectral measure of ARMA process in \cite{basu2015regularized}, we have $\lambda_{\mathrm{min}}\{\mathbb{E}(\bm{x}_t \bm{x}_t^\top)\} \geq \lambda_{\mathrm{min}}(\bm{\Sigma}_{e}) / \mu_{\mathrm{max}}(\cm{A}^\ast)=\kappa_{L, A}$ and $\lambda_{\mathrm{max}}\{\mathbb{E}(\bm{x}_t \bm{x}_t^\top)\} \leq \lambda_{\mathrm{max}}(\bm{\Sigma}_{e}) / \mu_{\mathrm{min}}(\cm{A}^\ast)=\kappa_{U, A}$. Then
	\begin{equation}\label{AR-RSC-expectation}
		\kappa_{L, A} \leq \mathbb{E}\norm{\bm{\Delta}}_T^2 = \mathbb{E}\left(\frac{1}{T_0} \sum_{t=P+1}^{T} \mathrm{vec}(\bm{\Delta}^\top)^\top (\bm{I}_Q \otimes \bm{x}_t \bm{x}_t^\top)\mathrm{vec}(\bm{\Delta}^\top)\right) \leq \kappa_{U, A},
	\end{equation}
	since $\|\bm{\Delta}\|_{\mathrm{F}}=1$.
	
	Combining \eqref{AR-RSC-concentration} and \eqref{AR-RSC-expectation}, we get the concentration inequality for $\norm{\bm{\Delta}}_T^2$ with
	\begin{equation}\label{AR-RSC-conclusion}
		\begin{split}
			&\mathbb{P}\left\{\kappa_{L,A} - c_3\kappa^2\kappa_{U,B}\frac{u+\omega(\bm{\Xi}')}{\sqrt{T_0}} \leq\inf_{\bm{\Delta}\in\bm{\Xi}'}\norm{\bm{\Delta}}_T^2\leq\sup_{\bm{\Delta}\in\bm{\Xi}'}\norm{\bm{\Delta}}_T^2\leq \kappa_{U, A} + c_3\kappa^2\kappa_{U,B}\frac{u+\omega(\bm{\Xi}')}{\sqrt{T_0}}\right\}\\
			&\hspace{350pt}\geq 1-2\exp(-u^2).
		\end{split}
	\end{equation}
	
	By the Dudley's inequality (see e.g., Theorem 8.1.10 of \cite{vershynin2019high}) and Lemma \ref{lemma:covering}(b), we have
	\begin{align}\label{gaussian-width}
		\omega(\bm{\Xi}')
		& \leq C_1\int_{0}^{2}\sqrt{\log\mathcal{N}\left(\bm{\Xi}',\norm{\cdot}_{\mathrm{F}},\epsilon\right)}\mathrm{d}\epsilon\notag\\
		& \leq C_1\int_{0}^{2}\left(4PR_yR_x + 2(R_y+R_x)\sum_{i=1}^{n}q_i\right)^{1/2} 
		\left[\log\left(\frac{C_2n(R_y \wedge PR_x) \sqrt{PR_yR_x}}{\epsilon} \right)\right]^{1/2}\mathrm{d}\epsilon\notag\\
		& \leq C_3\left(4PR_yR_x + 2(R_y+R_x)\sum_{i=1}^{n}q_i\right)^{1/2} 
		\left[\log (nR_y\sqrt{PR_x}) \right]^{1/2} \nonumber\\ 
		& =C_4 \sqrt{d_{\mathrm{AR}}{d_c}},
	\end{align}
	where $d_{\mathrm{AR}} = PR_yR_x + (R_y+R_x)\sum_{i=1}^{n}q_i$ and {$d_c = \log (nP^{1/2}R_yR_x^{1/2})$}.
	By letting $u = \omega(\bm{\Xi}')$ and $T_0\gtrsim \kappa^4\kappa_{U,B}^2d_{\mathrm{AR}}d_c/\kappa_{L,A}^2$, \eqref{AR-RSC-conclusion} can be reduced to
	\begin{equation*}
		\mathbb{P}\left\{\frac{\kappa_{L,A}}{16}\leq\inf_{\bm{\Delta}\in\bm{\Xi}'}\norm{\bm{\Delta}}_T^2\leq\sup_{\bm{\Delta}\in\bm{\Xi}'}\norm{\bm{\Delta}}_T^2\leq 4\kappa_{U, A}\right\}\geq 1-2\exp\left\{-c d_{\mathrm{AR}}d_c\right\}.
	\end{equation*}
	By spectral measure of ARMA process in \cite{basu2015regularized}, we could replace $\lambda_{\mathrm{max}}(\bm{P}\bm{P}^\top)$ by $1/\mu_{\mathrm{min}}(\bm{B}^\ast)$ where $\bm{B}^\ast(z) \vcentcolon= \bm{I}_{PQ} - \bm{B}^\ast z$ and $\mu_{\mathrm{\min}}(\bm{B}^\ast) = \min_{|z|=1}\lambda_{\mathrm{min}}(\bar{\bm{B}^\ast}(z)$ $\bm{B}^\ast(z))$. Thus, the proof is completed.
\end{proof}

\begin{lemma}[\textbf{Autoregression deviation bound}]\label{lemma:AR-DB}
	Suppose that Assumptions \ref{assump:core}, \ref{assump:stationarity}, \ref{assump:errorauto} hold and $T-P\gtrsim d_{\mathrm{AR}}d_{\mathrm{c}}$, we have	
	\[
	\sup_{\bm{\Delta}\in \bm{\Xi}'} \frac{1}{T-P}\sum_{t=P+1}^T\langle \bm{e}_{t},\bm{\Delta}\bm{x}_t\rangle \leq C_1\kappa^2\kappa_{U,B}'\sqrt{d_{\mathrm{AR}}d_c/(T-P)} 
	\]
	with probability at least $1-\exp\left\{-C_2 d_{\mathrm{AR}}d_c\right\}$, where $\kappa_{U, B}^\prime=C_{e}/{\mu^{1/2}_{\mathrm{min}}(\bm{B}^{\ast})}$, $d_{\mathrm{AR}} = PR_yR_x + (R_y+R_x)\sum_{i=1}^{n}q_i$, {$d_c = \log (nP^{1/2}R_yR_x^{1/2})$}, and $C_1, C_2$ are some positive constants.
\end{lemma}

\begin{proof}
	To use the conclusion of Lemma \ref{lemma:chaining}, we first define $\bm{W}(\bm{\Delta}) = \sum_{t=P+1}^T\langle \bm{e}_{t},\bm{\Delta}\bm{x}_t\rangle/T_0$. For any $\bm{\Delta}_i, \bm{\Delta}_j\in \bm{\Xi}'$, we have
	\begin{equation*}
		\begin{split}
			&\left|\frac{1}{T_0}\sum_{t=P+1}^T\langle \bm{e}_{t},\bm{\Delta}_i\bm{x}_t\rangle - \frac{1}{T_0}\sum_{t=P+1}^T\langle \bm{e}_{t},\bm{\Delta}_j\bm{x}_t\rangle \right|
			=\left|\sum_{t=P+1}^T\langle \bm{e}_{t},\bm{M}\bm{x}_t\rangle \right|.
		\end{split}
	\end{equation*}
	where $\bm{M}=(\bm{\Delta}_i-\bm{\Delta}_j)/T_0$ with $\norm{\bm{M}}_{\mathrm{F}} = \norm{(\bm{\Delta}_i-\bm{\Delta}_j)}_{\mathrm{F}}/T_0$. 
	
	We then provide an upper bound for $\mathbb{E}\left[\exp\left(\lambda\sum_{t=P+1}^T\langle \bm{e}_{t},\bm{M}\bm{x}_t\rangle\right)\right]$, for any positive $\lambda$.
	By Assumption \ref{assump:errorauto} and Lemma \ref{lemma:martingale}, we have
	\begin{align}\label{martingale-inequality}
		\mathbb{E}\left[\exp\left(\lambda\sum_{t=P+1}^T\langle \bm{e}_{t},\bm{M}\bm{x}_t\rangle\right)\right]
		&=\mathbb{E}\left[\exp\left(\lambda\sum_{t=P+1}^T\langle \bm{\xi}_t, \bm{\Sigma}_e^{1/2}\bm{M}\bm{x}_t\rangle\right)\right]\notag\\
		&\leq\mathbb{E}\left[\exp\left(C^\prime\lambda^2\kappa^2\sum_{t=P+1}^{T}\|\bm{\Sigma}_e^{1/2}\bm{M}\bm{x}_t\|_2^{2}\right)\right]\notag\\
		&\leq\mathbb{E}\left[\exp\left(C^\prime\lambda^2\kappa^2C_{e}\sum_{t=P+1}^{T}\|\bm{M}\bm{x}_t\|_2^{2}\right)\right],
	\end{align}
	where every entry of $\bm{\xi}_t$ is mean-zero and i.i.d. $\kappa^2$-sub-Gaussian.
	By a similar argument in Lemma \ref{lemma:AR-RSC}, we can rewrite 
	\begin{equation}\label{from-x-to-xi}
		\sum_{t=P+1}^{T}\|\bm{M}\bm{x}_t\|_2^{2} =\bar{\bm{\xi}}^\top(\bm{D}^\top\bm{P}^\top(\bm{I}_{T_0}\otimes\bm{M}^\top\bm{M})\bm{P}\bm{D})\bar{\bm{\xi}} =\bar{\bm{\xi}}^\top\bm{Q}_1^\top\bm{Q}_1\bar{\bm{\xi}} = \norm{\bm{Q}_1\bar{\bm{\xi}}}_2^2,
	\end{equation}
	where $\bm{Q}_1=(\bm{I}_{T_0}\otimes\bm{M})\bm{P}\bm{D}$. The upper bounds for $\norm{\bm{Q}_1}_{\op}$ and $\norm{\bm{Q}_1}_{\mathrm{F}}^2$ are
	\begin{equation}\label{lemma4:op}
		\begin{split}
			\norm{\bm{Q}_1}_{\op}
			&=\norm{(\bm{I}_{T_0}\otimes\bm{M})\bm{P}\bm{D}}_{\op}
			\leq\norm{\bm{D}}_{\op}\norm{\bm{P}}_{\op}\norm{\bm{I}_{T_0}\otimes\bm{M}}_{\op}\\ 
			&\leq C_{e}^{1/2}\lambda_{\mathrm{max}}^{1/2}(\bm{P}\bm{P}^\top)\norm{\bm{I}_{T_0}}_{\op}\norm{\bm{M}}_{\mathrm{F}}\leq\frac{1}{{T_0}}C_{e}^{1/2}\lambda_{\mathrm{max}}^{1/2}(\bm{P}\bm{P}^\top)\norm{\bm{\Delta}_i  - \bm{\Delta}_j}_{\mathrm{F}}.\\
			\norm{\bm{Q}_1}_{\mathrm{F}}^2& = \norm{(\bm{I}_{T_0}\otimes\bm{M})\bm{P}\bm{D}}_{\mathrm{F}}^2 \leq \norm{\bm{D}}_{\op}^2\norm{\bm{P}}_{\op}^2 \norm{\bm{I}_{T_0}\otimes\bm{M}}_{\mathrm{F}}^2\\
			&= C_{e}\lambda_{\mathrm{max}}(\bm{P}\bm{P}^\top)
			\norm{\bm{I}_{T_0}}_{\mathrm{F}}^2  \norm{\bm{M}}_{\mathrm{F}}^2 \leq \frac{1}{{T_0}}C_{e}\lambda_{\mathrm{max}}(\bm{P}\bm{P}^\top) \norm{\bm{\Delta}_i  - \bm{\Delta}_j}_{\mathrm{F}}^2.
		\end{split}
	\end{equation}
	Let $\widetilde{\bm{\xi}}$ be the non-zero part of $\bar{\bm{\xi}}$ with independent sub-gaussian entries, $\widetilde{\bm{Q}}_1$ be the corresponding submatrix of $\bm{Q}_1$, we have $\norm{\bm{Q}_1\bar{\bm{\xi}}}_2^2 = \norm{\widetilde{\bm{Q}}_1\widetilde{\bm{\xi}}}_2^2$. Combining this equation with \eqref{martingale-inequality} and \eqref{from-x-to-xi}, by Lemma \ref{lemma:exercise}, we have
	\begin{equation}\label{subG-to-subE}
		\begin{split}
			\mathbb{E}\left[\exp\left(\lambda\sum_{t=P+1}^T\langle \bm{e}_{t},\bm{M}\bm{x}_t\rangle\right)\right]
			&\leq\mathbb{E}\left[\exp\left(C^\prime\lambda^2\kappa^2C_{e}\norm{\widetilde{\bm{Q}}_1\widetilde{\bm{\xi}}}_2^2\right)\right]\\
			&\leq \exp\left(C\lambda^2\kappa^4C_{e}\norm{\widetilde{\bm{Q}}_1}_{\mathrm{F}}^2\right)\\
			&\leq\exp\left(C\lambda^2\kappa^4C_{e}\norm{\bm{Q}_1}_{\mathrm{F}}^2\right),
		\end{split}
	\end{equation}
	provided that $0 < \lambda \leq c_1 / (\kappa^2C_{e}^{1/2}\norm{\widetilde{\bm{Q}}_1}_{\op})$. Since $\widetilde{\bm{Q}}_1$ is a submatrix of $\bm{Q}_1$, we have $ \norm{\widetilde{\bm{Q}}_1}_{\op} \lesssim \norm{\bm{Q}_1}_{\op}$. Thus, \eqref{subG-to-subE} still holds for all $0 < \lambda \leq c / (\kappa^2C_{e}^{1/2}\norm{\bm{Q}_1}_{\op})$.
	
	Next, we verify the concentration inequality required by Lemma \ref{lemma:chaining}. By the standard Chernoff's argument and \eqref{subG-to-subE}, we have
	\begin{align*}
		&\mathbb{P}\left\{\left|\sum_{t=P+1}^T\langle \bm{e}_{t},\bm{M}\bm{x}_t\rangle \right|\geq u\right\} \\
		\leq&\inf\limits_{\lambda>0}\frac{\mathbb{E}\left[\exp\left(\lambda\sum_{t=P+1}^T\langle \bm{e}_{t},\bm{M}\bm{x}_t\rangle\right)\right]}{\exp(\lambda u)}\\
		\leq&\inf\limits_{0 < \lambda \leq c / (\kappa^2C_{e}^{1/2}\norm{\bm{Q}_1}_{\op})}\exp\left(C\lambda^2\kappa^4C_{e}\norm{\bm{Q}_1}_{\mathrm{F}}^2-\lambda u\right)\\
		\leq&\exp\left\{-c'\min\left(\frac{u}{\kappa^2C_{e}^{1/2}\norm{\bm{Q}_1}_{\op}}, \frac{u^2}{\kappa^4C_{e}\norm{\bm{Q}_1}_{\mathrm{F}}^2}\right)\right\}\\
		\leq&\exp\left\{-\min\left(\frac{u}{c^{\prime\prime}\kappa^2\kappa_{U,B}'\norm{\bm{\Delta}_i  - \bm{\Delta}_j}_{\mathrm{F}}/{T_0}}, \frac{u^2}{\left(c^{\prime\prime}\kappa^2\kappa_{U,B}' \norm{\bm{\Delta}_i  - \bm{\Delta}_j}_{\mathrm{F}}/\sqrt{{T_0}}\right)^2}\right)\right\}.
	\end{align*}
	where $\kappa_{U,B}'=C_{e}\lambda_{\mathrm{max}}^{1/2}(\bm{P}\bm{P}^\top)$. Let $d_1'(\bm{\Delta}_i, \bm{\Delta}_j) = c^{\prime\prime}\kappa^2\kappa_{U,B}'\norm{\bm{\Delta}_i  - \bm{\Delta}_j}_{\mathrm{F}}/T_0$ and $d_2'(\bm{\Delta}_i, \bm{\Delta}_j) = c^{\prime\prime}\kappa^2\kappa_{U,B}' $ $\norm{\bm{\Delta}_i  - \bm{\Delta}_j}_{\mathrm{F}}/\sqrt{T_0}$. We have the following concentration bound by Lemma \ref{lemma:chaining},
	\begin{equation}\label{DB-AR-concentration}
		\begin{split}
			\mathbb{P}\left\{\sup_{\bm{\Delta}_i\in\bm{\Xi}'}\frac{1}{T_0}\sum_{t=P+1}^T\langle \bm{e}_{t},\bm{\Delta}_i\bm{x}_t\rangle > C\left(\gamma_2(\bm{\Xi}',d_2') + \gamma_1(\bm{\Xi}',d_1') + u\mathrm{Diam}'_2({\bm{\Xi}'}) + u^2\mathrm{Diam}'_1(\bm{\Xi}')\right)\right\} \\ \leq \exp(-u^2),
		\end{split}
	\end{equation}
	where
	\begin{equation}\label{DB-gamma-fcn}
		\begin{split}
			\gamma_2(\bm{\Xi}',d_2')
			&=\left(c^{\prime\prime}\kappa^2\kappa_{U,B}'/\sqrt{{T_0}}\right)\cdot\gamma_2(\bm{\Xi}',\norm{\cdot}_{\mathrm{F}})
			\leq\left(c_2\kappa^2\kappa_{U,B}'/\sqrt{{T_0}}\right)\cdot\omega(\bm{\Xi}'),\\
			\gamma_1(\bm{\Xi}',d_1')
			&=\left(c^{\prime\prime}\kappa^2\kappa_{U,B}'/{T_0}\right)\cdot\gamma_1(\bm{\Xi}',\norm{\cdot}_{\mathrm{F}}) \leq\left(c_1\kappa^2\kappa_{U,B}'/{T_0}\right)\cdot\omega^2(\bm{\Xi}'),\\
			\mathrm{Diam}'_2({\bm{\Xi}'})
			&= 2c^{\prime\prime}\kappa^2\kappa_{U,B}'/\sqrt{{T_0}},\quad
			\mathrm{Diam}_1'(\bm{\Xi}') = 2c^{\prime\prime}\kappa^2\kappa_{U,B}'/{T_0}.
		\end{split}
	\end{equation}
	Substituting \eqref{DB-gamma-fcn} into \eqref{DB-AR-concentration}, we have
	\begin{align}\label{AR-DB-concentration}
		&\mathbb{P}\left\{\sup_{\bm{\Delta}\in\bm{\Xi}'}\frac{1}{{T_0}}\sum_{t=P+1}^T\langle \bm{e}_{t},\bm{\Delta}\bm{x}_t\rangle > c_2^\prime\kappa^2\kappa_{U,B}'\frac{u+\omega(\bm{\Xi}')}{\sqrt{{T_0}}} + c_1^\prime\kappa^2\kappa_{U,B}'\frac{u^2+\omega^2(\bm{\Xi}')}{{T_0}}\right\}\leq \exp(-u^2).
	\end{align}
	
	When $T_0\gtrsim \left(u+\omega(\bm{\Xi}')\right)^2$ with $u = \omega(\bm{\Xi}')$, by \eqref{gaussian-width}, the upper bound \eqref{AR-DB-concentration} can be reduced to
	\[
	\mathbb{P}\left\{\sup_{\bm{\Delta}\in\bm{\Xi}'}\frac{1}{T_0}\sum_{t=P+1}^T\langle \bm{e}_{t},\bm{\Delta}\bm{x}_t\rangle \leq C_1\kappa^2\kappa_{U,B}'\sqrt{d_{\mathrm{AR}}d_c/T_0} \right\}\notag \geq 1-\exp(-C_2 d_{\mathrm{AR}}d_c).
	\]
	Again by spectral measure of ARMA process in \cite{basu2015regularized}, we could replace $\lambda_{\mathrm{max}}(\bm{P}\bm{P}^\top)$ by $1/\mu_{\mathrm{min}}(\bm{B}^\ast)$ where $\bm{B}^\ast(z) \vcentcolon= \bm{I}_{pQ} - \bm{B}^\ast z$ and $\mu_{\mathrm{\min}}(\bm{B}^\ast) = \min_{|z|=1}\lambda_{\mathrm{min}}(\bar{\bm{B}^\ast}(z)$ $\bm{B}^\ast(z))$.
\end{proof}

\begin{lemma}[\textbf{Autoregression RSC - Sparsity}]\label{lemma:AR-RSC-sparse} 
	Suppose that Assumptions \ref{assump:stationarity}, \ref{assump:errorauto} and \ref{assump:sparse} hold, and $T-P \gtrsim \kappa^4 \kappa_{U, B}^2 s\log(PQ^2) / \kappa_{L, A}^2$, we have	
	\[
	\frac{1}{4}\sqrt{\kappa_{L, A}}
	\left\|\bm{\Delta}_\mathrm{S}\right\|_{\mathrm{F}}
	\leq
	\left\|\bm{\Delta}_\mathrm{S} \right\|_T
	\leq2\sqrt{\kappa_{U, A}}\left\|\bm{\Delta}_\mathrm{S} \right\|_{\mathrm{F}},
	\]
	for all $\bm{\Delta}_\mathrm{S} \in \mathcal{H}_{2s}(Q, PQ)$ with probability at least $1-2\exp\left\{-cs\log(PQ^2)\right\}$, where $\kappa_{U, A}=C_{e}/\mu_{\mathrm{min}}(\cm{A}^{\ast})$, $\kappa_{L, A}=c_{e}/\mu_{\mathrm{max}}(\cm{A}^\ast)$, $\kappa_{U, B}=C_{e}/\mu_{\mathrm{min}}(\bm{B}^{\ast})$, $Q=\prod_{i=1}^{n}q_i$, and $c$ is a positive constant.
\end{lemma}

\begin{proof}
	The proof of this lemma closely follows Lemma \ref{lemma:AR-RSC}.
	Using the similar arguments as the derivation of \eqref{AR-RSC-concentration} in Lemma \ref{lemma:AR-RSC}, we have
	\begin{equation}\label{eq:ar-rsc-sparse-sup}
		\mathbb{P}\left\{\sup_{\bm{\Delta}_\mathrm{S}\in\mathcal{H}_{2s}^\prime}\left|\norm{\bm{\Delta}_\mathrm{S}}_T^2 -\mathbb{E}\norm{\bm{\Delta}_\mathrm{S}}_T^2\right| > c_1\kappa^2\kappa_{U,B}\frac{u+\omega(\mathcal{H}_{2s}^\prime)}{\sqrt{T_0}}\right\} \leq 2\exp(-u^2).
	\end{equation}
	when $T_0\gtrsim \left(u+\omega(\mathcal{H}_{2s}^\prime)\right)^2$. Since the lower and upper bounds for $\mathbb{E}\norm{\bm{\Delta}}_T^2$ in \eqref{AR-RSC-expectation} of Lemma \ref{lemma:AR-RSC} hold for any $\bm{\Delta}$ with $\norm{\bm{\Delta}}_\Fr = 1$, then 
	\begin{equation}\label{eq:ar-rsc-sparse-expectation}
		\kappa_{L, A} \leq \mathbb{E}\norm{\bm{\Delta}_\mathrm{S}}_T^2 = \mathbb{E}\left(\frac{1}{T_0} \sum_{t=P+1}^{T} \mathrm{vec}(\bm{\Delta}_\mathrm{S}^\top)^\top (\bm{I}_Q \otimes \bm{x}_t \bm{x}_t^\top)\mathrm{vec}(\bm{\Delta}_\mathrm{S}^\top)\right) \leq \kappa_{U, A}.
	\end{equation}
	Combining \eqref{eq:ar-rsc-sparse-sup} and \eqref{eq:ar-rsc-sparse-expectation}, we get the concentration inequality for $\norm{\bm{\Delta}_\mathrm{S}}_T^2$ with
	\begin{equation} \label{eq:ar-rsc-sparse-conclusion}
		\begin{split}
			&\mathbb{P}\left\{\kappa_{L,A} - c_1\kappa^2\kappa_{U,B}\frac{u+\omega(\mathcal{H}_{2s}^\prime)}{\sqrt{T_0}} \leq\inf_{\bm{\Delta}_\mathrm{S}\in\mathcal{H}_{2s}^\prime}\norm{\bm{\Delta}_\mathrm{S}}_T^2\leq\sup_{\bm{\Delta}_\mathrm{S}\in\mathcal{H}_{2s}^\prime}\norm{\bm{\Delta}_\mathrm{S}}_T^2\leq \kappa_{U, A} + c_1\kappa^2\kappa_{U,B}\frac{u+\omega(\mathcal{H}_{2s}^\prime)}{\sqrt{T_0}}\right\}\\
			&\hspace{350pt}\geq 1-2\exp(-u^2).
		\end{split}
	\end{equation}
	By the Dudley's inequality and Lemma \ref{lemma:covering}(e), we have
	\begin{align} \label{eq:ar-rsc-sparse-gs}
		\omega({\mathcal{H}^\prime_{2s}}) \leq
		C_1 \int_{0}^{2} \sqrt{\log \mathcal{N}(\mathcal{H}^\prime_{2s}, \norm{\cdot}_\Fr, \epsilon)} \mathrm{d}\epsilon 
		\leq 
		C_1 \int_{0}^{2} \left\{2s\log\left(\frac{3PQ^2}{\epsilon}\right)\right\}^{1/2} \mathrm{d}\epsilon 
		\leq
		C_2 s^{1/2} (\log (PQ^2))^{1/2}.
	\end{align}
	By letting $u = \omega({\mathcal{H}^\prime_{2s}})$ and $T_0 \gtrsim \kappa^4 \kappa_{U, B}^2 s\log(PQ^2) / \kappa_{L, A}^2$,
	\eqref{eq:ar-rsc-sparse-conclusion} can be reduced to 
	\[
	\mathbb{P}\left\{\frac{\kappa_{L,A}}{16}\leq\inf_{\bm{\Delta}_\mathrm{S}\in\mathcal{H}_{2s}^\prime}\norm{\bm{\Delta}_\mathrm{S}}_T^2\leq\sup_{\bm{\Delta}_\mathrm{S}\in\mathcal{H}_{2s}^\prime}\norm{\bm{\Delta}_\mathrm{S}}_T^2 \leq
	4\kappa_{U, A}\right\}\geq 1-2\exp\left\{-cs\log(PQ^2)\right\}.
	\]
\end{proof}

\begin{lemma}[\textbf{Autoregression deviation bound - Sparsity}]\label{lemma:ar-deviation-sparse}
	Suppose that Assumptions \ref{assump:stationarity}, \ref{assump:errorauto} and \ref{assump:sparse} hold and $T- P\gtrsim s \log(PQ^2)$, we have
	\[
	\sup_{\bm{\Delta}_\mathrm{S}\in\mathcal{H}_{2s}^\prime}\frac{1}{T_0}\sum_{t=P+1}^T\langle \bm{e}_{t},\bm{\Delta}_\mathrm{S}\bm{x}_t\rangle \leq C\kappa^2\kappa_{U,B}'
	\sqrt{s\log(PQ^2)/(T-P)}
	\]
	with probability at least $ 1-\exp\left\{-c s\log(PQ^2)\right\}$,
	where $\kappa_{U, B}^\prime=C_{e}/{\mu^{1/2}_{\mathrm{min}}(\bm{B}^{\ast})}$, and $C,c$'s are some positive constants.
\end{lemma}

\begin{proof}
	The proof of this lemma closely follows Lemma \ref{lemma:AR-DB}. Using the similar argument as the derivation of \eqref{DB-AR-concentration} in Lemma \ref{lemma:AR-DB}, we have 
	\begin{equation} \label{eq:ar-dev-sparse-sup}
		\begin{split}
			\mathbb{P}\left\{\sup_{\bm{\Delta}_\mathrm{S}\in\mathcal{H}_{2s}^\prime}\frac{1}{T_0}\sum_{t=P+1}^T\langle \bm{e}_{t},\bm{\Delta}_\mathrm{S}\bm{x}_t\rangle 
			> C_1 \left(\gamma_2(\mathcal{H}_{2s}^\prime,d_2') + 
			\gamma_1(\mathcal{H}_{2s}^\prime,d_1') + 
			u\mathrm{Diam}'_2({\mathcal{H}_{2s}^\prime}) + 
			u^2\mathrm{Diam}'_1(\mathcal{H}_{2s}^\prime)\right)\right\} \\ 
			\leq \exp(-u^2),
		\end{split}
	\end{equation}
	where $d_1'(\bm{\Delta}_{\mathrm{S},i}, \bm{\Delta}_{\mathrm{S},j}) = c_1\kappa^2\kappa_{U,B}'\norm{\bm{\Delta}_{\mathrm{S},i}  - \bm{\Delta}_{\mathrm{S},j}}_{\mathrm{F}}/T_0$ and $d_2'(\bm{\Delta}_{\mathrm{S},i}, \bm{\Delta}_{\mathrm{S},j}) = c_1\kappa^2\kappa_{U,B}' \norm{\bm{\Delta}_{\mathrm{S},i}  - \bm{\Delta}_{\mathrm{S},j}}_{\mathrm{F}}$ $/\sqrt{T_0}$.
	Since
	\begin{equation} \label{eq:ar-dev-sparse-gamma}
		\begin{split}
			\gamma_2(\mathcal{H}_{2s}^\prime,d_2')
			&=\left(c_1\kappa^2\kappa_{U,B}'/\sqrt{T_0}\right)\cdot\gamma_2(\mathcal{H}_{2s}^\prime,\norm{\cdot}_{\mathrm{F}})
			\leq\left(c_2\kappa^2\kappa_{U,B}'/\sqrt{T_0}\right)\cdot\omega(\mathcal{H}_{2s}^\prime),\\
			\gamma_1(\mathcal{H}_{2s}^\prime,d_1')
			&=\left(c_1\kappa^2\kappa_{U,B}'/T_0\right)\cdot\gamma_1(\mathcal{H}_{2s}^\prime,\norm{\cdot}_{\mathrm{F}}) \leq\left(c_3\kappa^2\kappa_{U,B}'/T_0\right)\cdot\omega^2(\mathcal{H}_{2s}^\prime),\\
			\mathrm{Diam}'_2({\mathcal{H}_{2s}^\prime})
			&= 2c_1\kappa^2\kappa_{U,B}'/\sqrt{T_0},\quad
			\mathrm{Diam}_1'(\mathcal{H}_{2s}^\prime) = 2c_1\kappa^2\kappa_{U,B}'/T_0.
		\end{split}
	\end{equation}
	Substituting \eqref{eq:ar-dev-sparse-gamma} into \eqref{eq:ar-dev-sparse-sup}, we have
	\begin{align}\label{eq:ar-dev-sparse-bound}
		&\mathbb{P}\left\{\sup_{\bm{\Delta}_\mathrm{S}\in\mathcal{H}_{2s}^\prime}\frac{1}{T_0}\sum_{t=P+1}^T\langle \bm{e}_{t},\bm{\Delta}_\mathrm{S}\bm{x}_t\rangle 
		> c_4\kappa^2\kappa_{U,B}'\frac{u+\omega(\mathcal{H}_{2s}^\prime)}{\sqrt{T_0}} + c_5\kappa^2\kappa_{U,B}'\frac{u^2+\omega^2(\mathcal{H}_{2s}^\prime)}{T_0}\right\}\leq \exp(-u^2).
	\end{align}
	
	When $T_0 \gtrsim \left(u+\omega(\mathcal{H}_{2s}^\prime)\right)^2$ with $u = \omega(\mathcal{H}_{2s}^\prime)$, by \eqref{eq:ar-rsc-sparse-gs}, the upper bound \eqref{eq:ar-dev-sparse-bound} can be reduced to
	\[
	\mathbb{P}\left\{\sup_{\bm{\Delta}_\mathrm{S}\in\mathcal{H}_{2s}^\prime}\frac{1}{T_0}\sum_{t=P+1}^T\langle \bm{e}_{t},\bm{\Delta}\bm{x}_t\rangle \leq C\kappa^2\kappa_{U,B}'
	\sqrt{s\log(PQ^2)/T_0} \right\}\notag \geq 1-\exp\left\{-c s\log(PQ^2)\right\}.
	\]
\end{proof}

\begin{lemma} \label{lemma:reg-rsc-general}
	Let the predictors $\{\bm{x}_t\}$ follow Assumption \ref{assump:reg-input} for any $1 \leq t \leq T$. Then for any $\bm{u} \in \bm{\Phi} \subseteq S^{\prod_{j=1}^{m}p_j-1}$ with $\bm{\Phi}$ being an arbitrary parameter space in the unit sphere, we have when $T \gtrsim \sigma^4 \omega^2(\bm{\Phi})$,
	\[
	\frac{1}{16} c_x \leq \inf_{\bm{u} \in \bm{\Phi}} \frac{1}{T} \sum_{t=1}^T \norm{\bm{u}^\top \bm{x}_t}_2^2
	\leq \sup_{\bm{u} \in \bm{\Phi}} \frac{1}{T} \sum_{t=1}^T \norm{\bm{u}^\top \bm{x}_t}_2^2
	\leq 4C_x,
	\]
	with probability at least $1 - \exp\left\{-\eta T / \sigma^4\right\}$, where $\eta$ is some positive constant.
\end{lemma}
\begin{proof}
	This lemma is a modification of Theorem 6 in \cite{banerjee2015estimation}. For simplicity of notations, let $\bm{x}_0$ be i.i.d. as $\bm{x}_t$ for any $1 \leq t \leq T$, and thus $\bm{x}_0$ also satisfies Assumption \ref{assump:reg-input}. To apply Theorem 10 of \cite{banerjee2015estimation}, we consider the following class of functions:
	\begin{equation*}
		F = \left\{
		f_{\bm{u}}, \bm{u} \in \bm{\Phi}: 
		f_{\bm{u}}(\cdot) = \frac{1}{\sqrt{\bm{u}^\top \bm{\Sigma}_x \bm{u}}} 
		\left\langle \cdot, \bm{u} \right\rangle
		\right\}.
	\end{equation*}
	Since for any $f_{\bm{u}} \in F$,
	\[
	\vertiii{f_{\bm{u}}}_{L_2}^2 = \frac{1}{\bm{u}^\top \bm{\Sigma}_x \bm{u}} \mathbb{E}[\bm{u}^\top\bm{x}_0\bm{x}_0^\top\bm{u}] = 1,
	\]
	then $F$ is a subset of the unit sphere, i.e., $F \subseteq S_{L_2}$.
	
	Next, we provide the upper bound for $\sup_{f_{\bm{u}} \in F} \vertiii{f_{\bm{u}}}_{\psi_2} = \sup_{\bm{u} \in \bm{\Phi}} \vertiii{\left\langle \bm{x}_0, \bm{u} \right\rangle / \sqrt{\bm{u}^\top \bm{\Sigma}_x \bm{u}}}_{\psi_2} $. Let $\widetilde{\bm{x}}_0 = \bm{\Sigma}_x^{-1/2}\bm{x_0}$, then $\widetilde{\bm{x}}_0$ is isotropic with $\vertiii{\widetilde{\bm{x}}_0}_{\psi_2} \leq \sigma$. Hence
	\[
	\left\langle \bm{x}_0, \bm{u} \right\rangle = 
	\left\langle \widetilde{\bm{x}}_0 , \bm{\Sigma}_x^{1/2}\bm{u} \right\rangle = \sqrt{\bm{u}^\top \bm{\Sigma}_x \bm{u}} 
	\left\langle \widetilde{\bm{x}}_0 , \frac{\bm{\Sigma}_x^{1/2}\bm{u}}{\norm{\bm{\Sigma}_x^{1/2} \bm{u}}}_2  \right\rangle
	\]
	and
	\begin{equation*}
		\sup_{f_{\bm{u}} \in F} \vertiii{ f_{\bm{u}}}_{\psi_2} = \sup_{\bm{u} \in \bm{\Phi}}
		\vertiii{	\frac{1}{\sqrt{\bm{u}^\top \bm{\Sigma}_x \bm{u}}}
			\left\langle \bm{x}_0, \bm{u} \right\rangle }_{\psi_2}  = 
		\sup_{\bm{u} \in \bm{\Phi}}\vertiii{\left\langle \widetilde{\bm{x}}_0 , \frac{\bm{\Sigma}_x^{1/2}\bm{u}}{\norm{\bm{\Sigma}_x^{1/2} \bm{u}}}_2  \right\rangle}_{\psi_2}
		\leq 
		\vertiii{\widetilde{\bm{x}}_0}_{\psi_2} \leq \sigma.
	\end{equation*}
	This leads to the following upper bounds:
	\[
	\gamma_2(F \cap S_{L_2}, \vertiii{\cdot}_{\psi_2}) \leq 
	\sigma 	\gamma_2(F \cap S_{L_2}, \vertiii{\cdot}_{L_2}) \leq
	c_1\sigma \omega(\bm{\Phi}).
	\]
	In the context of Theorem 10 from \cite{banerjee2015estimation}, let $\sqrt{T} \geq 2c_1c_2\sigma^2 \omega(\bm{\Phi})$ and $\theta = 1/2$, then the condition 
	\[
	\theta \sqrt{T} \geq c_1c_2\sigma^2 \omega(\bm{\Phi}) \geq c_2\sigma \gamma_2(F \cap S_{L_2}, \vertiii{\cdot}_{\psi_2})
	\] 
	is satisfied. Therefore, by Theorem 10 from \cite{banerjee2015estimation}, we have
	\[
	\sup_{\bm{u} \in \bm{\Phi}}
	\left| 
	\frac{1}{T} \frac{1}{\bm{u}^\top \bm{\Sigma}_x \bm{u}}
	\sum_{t=1}^{T} \left\langle \bm{x}_t, \bm{u} \right\rangle^2 - 1 
	\right| \leq \frac{1}{2}
	\]
	with probability at least $1 - \exp\left\{-\eta T / \sigma^4\right\}$, where $\eta$ is some positive constant.
	As a result, we have
	\begin{equation}
		\sup_{\bm{u} \in \bm{\Phi}}
		\left(
		\frac{1}{T} \frac{1}{\bm{u}^\top \bm{\Sigma}_x \bm{u}}
		\sum_{t=1}^{T} \left\langle \bm{x}_t, \bm{u} \right\rangle^2 - 1 
		\right) \leq \frac{1}{2}
		\quad \text{and} \quad
		\sup_{\bm{u} \in \bm{\Phi}}
		\left( 1 - 
		\frac{1}{T} \frac{1}{\bm{u}^\top \bm{\Sigma}_x \bm{u}}
		\sum_{t=1}^{T} \left\langle \bm{x}_t, \bm{u} \right\rangle^2 
		\right) \leq \frac{1}{2}. \label{eq:rsc-general-bounds}
	\end{equation}
	
	By Assumption 1, we have $c_x \leq \bm{u}^\top \bm{\Sigma}_x \bm{u} \leq C_x$ for any $\bm{u} \in S^{\prod_{j=1}^{m}p_j - 1}$. Then from \eqref{eq:rsc-general-bounds}, we can get
	\[
	\frac{1}{C_x} \sup_{\bm{u} \in \bm{\Phi}} \frac{1}{T} \sum_{t=1}^T \norm{\bm{u}^\top \bm{x}_t}_2^2
	\leq  
	\sup_{\bm{u} \in \bm{\Phi}}
	\frac{1}{T} \frac{1}{\bm{u}^\top \bm{\Sigma}_x \bm{u}}
	\sum_{t=1}^{T} \left\langle \bm{x}_t, \bm{u} \right\rangle^2
	\leq \frac{3}{2} < 4
	\]
	and 
	\[
	\frac{1}{c_x} \inf_{\bm{u} \in \bm{\Phi}} \frac{1}{T} \sum_{t=1}^T \norm{\bm{u}^\top \bm{x}_t}_2^2
	\geq
	\inf_{\bm{u} \in \bm{\Phi}}
	\frac{1}{T} \frac{1}{\bm{u}^\top \bm{\Sigma}_x \bm{u}}
	\sum_{t=1}^{T} \left\langle \bm{x}_t, \bm{u} \right\rangle^2
	\geq \frac{1}{2} > \frac{1}{16}.
	\]
	Combining the two, we have when $T \gtrsim \sigma^4 \omega^2(\bm{\Phi})$,
	\[
	\frac{1}{16} c_x \leq \inf_{\bm{u} \in \bm{\Phi}} \frac{1}{T} \sum_{t=1}^T \norm{\bm{u}^\top \bm{x}_t}_2^2
	\leq \sup_{\bm{u} \in \bm{\Phi}} \frac{1}{T} \sum_{t=1}^T \norm{\bm{u}^\top \bm{x}_t}_2^2
	\leq 4C_x,
	\]
	with probability at least $1 - \exp\left\{-\eta T / \sigma^4\right\}$.
\end{proof}

\begin{lemma}[\textbf{Chaining}]\label{lemma:chaining} 
	Suppose the random process $\bm{W}(\bm{\Delta})_{\bm{\Delta}\in\bm{\Xi}'}$ has a mixed tail that for any $\bm{\Delta}_i, \bm{\Delta}_j \in \bm{\Xi}'$ and all $u > 0$,
	\[
	\mathbb{P}\left\{\left|\bm{W}(\bm{\Delta}_i) - \bm{W}(\bm{\Delta}_j)\right|\geq u\right\}\notag \leq 2\exp\left\{-\min\left(\frac{u^2}{d_2(\bm{\Delta}_i, \bm{\Delta}_j )^2}, \frac{u}{d_1(\bm{\Delta}_i, \bm{\Delta}_j )}\right)\right\},
	\]
	then we have
	\[
	\mathbb{P}\left\{\sup_{\bm{\Delta}\in\bm{\Xi}'}\left|\bm{W}(\bm{\Delta})\right| > C\left(\gamma_2(\bm{\Xi}',d_2) + \gamma_1(\bm{\Xi}',d_1) + u\mathrm{Diam}_2({\bm{\Xi}'}) + u^2\mathrm{Diam}_1(\bm{\Xi}')\right) \right\}\notag \leq 2\exp(-u^2),
	\]
	where $\mathrm{Diam}_2({\bm{\Xi}'})$ and $\mathrm{Diam}_1(\bm{\Xi}')$ are the diameters of $\bm{\Xi}'$ with respect to the semi-metric $d_2$ and $d_1$, respectively.
\end{lemma}

\begin{proof}
	This is a direct conclusion from Theorem 3.5 of \cite{dirksen2015tail}, and further use the fact that when $\bm{\Delta} = 0$, $\bm{W}(\bm{\Delta})=\bm{0}$ because $\bm{0}$ is an element in the parameter space. 
\end{proof}

\begin{lemma}[\textbf{Martingale inequality}]\label{lemma:martingale} 
	Let $\{\mathcal{F}_t, t\in\mathbb{Z}\}$ be a filtration. Suppose that $\{\bm{w}_t\}$ and $\{\bm{e}_t\}$ are processes taking values in $\mathbb{R}^n$, and for each integer $t$, $\bm{w}_t$ is $\mathcal{F}_{t-1}$-measurable, $\bm{e}_t$ is $\mathcal{F}_{t}$-measurable, and every entry of $\bm{e}_t\mid\mathcal{F}_{t-1}$ is mean-zero and i.i.d. $\kappa^2$-sub-Gaussian distributed. Then for any $\lambda>0$, we have 
	\begin{equation*}
		\mathbb{E}\left[\exp\left(\lambda\sum_{t=1}^{T}\langle\bm{w}_t,\bm{e}_t \rangle\right)\right] \leq \mathbb{E}\left[\exp\left(C\lambda^2\kappa^2\sum_{t=1}^{T}\|\bm{w}_{t}\|_2^{2}\right)\right],
	\end{equation*}
	where $C$ is some positive constants.
\end{lemma}

\begin{proof}
	By the tower rule, we have
	\begin{align*}
		\mathbb{E}\left[\exp\left(\lambda\sum_{t=1}^{T}\langle\bm{w}_t,\bm{e}_t \rangle\right)\right] 
		&= \mathbb{E}\left[\mathbb{E}\left[ \exp\left(\lambda\sum_{t=1}^{T} \langle\bm{w}_t, \bm{e}_t\rangle \right)\mid\mathcal{F}_{T-1}\right]\right]\\
		&= \mathbb{E}\left[\exp\left(\lambda\sum_{t=1}^{T-1}\langle\bm{w}_t, \bm{e}_t \rangle\right) \mathbb{E}\left[\exp\left(\lambda \langle\bm{w}_{T}, \bm{e}_T \rangle\right) \mid\mathcal{F}_{T-1}\right]\right]\\
		&\leq \mathbb{E}\left[\exp\left(c^\prime\lambda^2\kappa^2\|\bm{w}_{T}\|_2^{2}\right) \mathbb{E}\left[\exp\left(\lambda\sum_{t=1}^{T-1}\langle\bm{w}_t, \bm{e}_t \rangle\right)\right] \right],
	\end{align*}
	where the above inequality follows from the the fact that $\langle \bm{w}_{T},\bm{e}_T\rangle \mid \mathcal{F}_{T-1}$ is mean-zero and $\kappa^2\|\bm{w}_{T}\|^2_2$-sub-Gaussian, since the entries in $\bm{e}_T\mid \mathcal{F}_{t-1}$ is mean-zero and i.i.d. $\kappa^2$-sub-Gaussian. Thus we have $\mathbb{E}\left[ \exp\left(\lambda\langle\bm{w}_{T}, \bm{e}_T \rangle \right) \mid\mathcal{F}_{T-1}\right] \leq \exp\left(c^\prime\lambda^2\kappa^2\|\bm{w}_{T}\|_2^{2}\right)$.
	Then induction on $t$ gives the conclusion.
\end{proof}

\begin{lemma}\label{lemma:exercise} 
	Suppose that $\bm{\xi}$ is a mean zero, sub-gaussian random vector in $\mathbb{R}^n$
	with $\norm{\bm{\xi}}_{\psi_2}\leq \kappa$. Let $\bm{M}$ be an $m\times n$ matrix. We have
	\begin{equation*}
		\mathbb{E}\left[\exp\left(\lambda^2\norm{\bm{M}\bm{\xi}}_2^2\right)\right] \leq \exp\left(C\kappa^2\lambda^2\norm{\bm{M}}_{\mathrm{F}}^2\right)\hspace{10pt}\text{provided } |\lambda|\leq\frac{c}{\kappa\norm{\bm{M}}_{\op}},
	\end{equation*}
	where $C$ and $c$ are some positive constants.
\end{lemma}

\begin{proof}This proof is similar to the proof of Hanson-Wright inequality in Section 6.2 of \cite{vershynin2019high}. 
	Let $\bm{g} \in \mathbb{R}^m$ be a standard Gaussian random vector, i.e. $\bm{g} \sim N(\bm{0}, \bm{I}_m)$,  and independent with $\bm{\xi}$. Note that
	\begin{align*}
		\mathbb{E}\left[\exp\left(\lambda\bm{g}^\top\bm{M}\bm{\xi}\right)\right] &=\mathbb{E}_{\bm{\xi}}\left[\mathbb{E}_{\bm{g}}\left[\exp\left(\lambda\bm{g}^\top\bm{M}\bm{\xi}\right)\mid\bm{\xi}\right]\right] =\mathbb{E}\left[\exp\left(\lambda^2\norm{\bm{M}\bm{\xi}}_2^2/2\right)\right],\\
		\mathbb{E}\left[\exp\left(\lambda\bm{g}^\top\bm{M}\bm{\xi}\right)\right] &=\mathbb{E}_{\bm{g}}\left[\mathbb{E}_{\bm{\xi}}\left[\exp\left(\lambda\bm{g}^\top\bm{M}\bm{\xi}\right)\mid\bm{g}\right]\right] \leq\mathbb{E}\left[\exp\left(C\kappa^2\lambda^2\norm{\bm{M}^\top\bm{g}}_2^2\right)\right],
	\end{align*}
	where $C$ is a positive constant. Comparing these two lines, we get
	\begin{equation*}
		\mathbb{E}\left[\exp\left(\lambda^2\norm{\bm{M}\bm{\xi}}_2^2/2\right)\right] \leq\mathbb{E}\left[\exp\left(C\kappa^2\lambda^2\norm{\bm{M}^\top\bm{g}}_2^2\right)\right].
	\end{equation*}
	Expressing $\bm{M}\bm{M}^\top \in \mathbb{R}^{m \times m}$ through its singular value decomposition $\bm{M}\bm{M}^\top=\sum_{i=1}^{m}s_i^2\bm{u}_i\bm{u}_i^\top$, where $s_i$ are the singular values of $\bm{M}$. We can write
	\[\norm{\bm{M}^\top\bm{g}}_2^2 = \bm{g}^\top\bm{M}\bm{M}^\top\bm{g} = \sum_{i=1}^{m}s_i^2\langle\bm{u}_i,\bm{g}\rangle\langle\bm{u}_i,\bm{g}\rangle.\]
	By rotation invariance of the normal distribution, $\left(\langle\bm{u}_i,\bm{g}\rangle\right)_{i=1}^m$ is still a standard normal random vector, which is denoted by $\bm{g}^\prime$. By independence we have
	\begin{equation*}
		\begin{split}
			\mathbb{E}\left[\exp\left(C\kappa^2\lambda^2\norm{\bm{M}^\top\bm{g}}_2^2\right)\right] &=\prod_{i=1}^{m}\mathbb{E}\left[\exp\left(C\kappa^2\lambda^2s_i^2g_i^{\prime2}\right)\right]\\ &\leq\prod_{i=1}^{m}\exp\left(C'\kappa^2\lambda^2s_i^2\right)\\ &=\exp\left(C'\kappa^2\lambda^2\sum_{i=1}^{n}s_i^2\right)\\ &=\exp\left(C'\kappa^2\lambda^2\norm{\bm{M}}_{\mathrm{F}}^2\right)
		\end{split}
	\end{equation*}
	provided that $\lambda^2\kappa^2 s_i^2  \leq c'$ for all $i=1,2,\ldots,m$, where the inequality follows from Proposition 2.7.1 in \cite{vershynin2019high} for the sub-exponential random variable ${g_i^\prime}^2$. Note that $\max_is_i=\norm{\bm{M}}_{\op}$, so the lemma is proved.
\end{proof}

\stopcontents

\bibliography{CP}

\begin{thebibliography}{}

\bibitem[Agarwal et~al., 2012]{agarwal2012noisy}
Agarwal, A., Negahban, S., and Wainwright, M.~J. (2012).
\newblock Noisy matrix decomposition via convex relaxation: Optimal rates in
  high dimensions.
\newblock {\em The Annals of statistics}, 40(2):1171--1197.

\bibitem[Bahadori et~al., 2014]{bahadori2014fast}
Bahadori, M.~T., Yu, Q.~R., and Liu, Y. (2014).
\newblock Fast multivariate spatio-temporal analysis via low rank tensor
  learning.
\newblock In Ghahramani, Z., Welling, M., Cortes, C., Lawrence, N., and
  Weinberger, K., editors, {\em Advances in Neural Information Processing
  Systems}, volume~27. Curran Associates, Inc.

\bibitem[Bai and Wang, 2016]{bai2016econometric}
Bai, J. and Wang, P. (2016).
\newblock {Econometric Analysis of Large Factor Models}.
\newblock {\em Annual Review of Economics}, 8(1):53--80.

\bibitem[Bai et~al., 2023]{bai2023multiple}
Bai, P., Safikhani, A., and Michailidis, G. (2023).
\newblock Multiple change point detection in reduced rank high dimensional
  vector autoregressive models.
\newblock {\em Journal of the American Statistical Association},
  118(544):2776--2792.

\bibitem[Banerjee et~al., 2015]{banerjee2015estimation}
Banerjee, A., Chen, S., Fazayeli, F., and Sivakumar, V. (2015).
\newblock Estimation with norm regularization.

\bibitem[Basu et~al., 2019]{basu2019low}
Basu, S., Li, X., and Michailidis, G. (2019).
\newblock Low rank and structured modeling of high-dimensional vector
  autoregressions.
\newblock {\em IEEE Transactions on Signal Processing}, 67(5):1207--1222.

\bibitem[Basu and Michailidis, 2015]{basu2015regularized}
Basu, S. and Michailidis, G. (2015).
\newblock Regularized estimation in sparse high-dimensional time series models.
\newblock {\em The Annals of Statistics}, 43:1535--1567.

\bibitem[Bertsimas and Parys, 2020]{bertsimas2020sparse}
Bertsimas, D. and Parys, B.~V. (2020).
\newblock Sparse high-dimensional regression: Exact scalable algorithms and
  phase transitions.
\newblock {\em The Annals of Statistics}, 48(1):pp. 300--323.

\bibitem[Bi et~al., 2018]{bi2018multilayer}
Bi, X., Qu, A., and Shen, X. (2018).
\newblock Multilayer tensor factorization with applications to recommender
  systems.
\newblock {\em The Annals of Statistics}, 46(6B):3308--3333.

\bibitem[Billio et~al., 2024]{billio2024bayesian}
Billio, M., Casarin, R., and Iacopini, M. (2024).
\newblock Bayesian markov-switching tensor regression for time-varying
  networks.
\newblock {\em Journal of the American Statistical Association},
  119(545):109--121.

\bibitem[Cai et~al., 2020]{cai2020uncertainty}
Cai, C., Poor, H.~V., and Chen, Y. (2020).
\newblock Uncertainty quantification for nonconvex tensor completion:
  Confidence intervals, heteroscedasticity and optimality.
\newblock In III, H.~D. and Singh, A., editors, {\em Proceedings of the 37th
  International Conference on Machine Learning}, volume 119 of {\em Proceedings
  of Machine Learning Research}, pages 1271--1282. PMLR.

\bibitem[Cai et~al., 2022]{cai2022provable}
Cai, J.-F., Li, J., and Xia, D. (2022).
\newblock Provable tensor-train format tensor completion by riemannian
  optimization.
\newblock {\em Journal of Machine Learning Research}, 23(123):1--77.

\bibitem[Cai et~al., 2023]{cai2023generalized}
Cai, J.-F., Li, J., and Xia, D. (2023).
\newblock Generalized low-rank plus sparse tensor estimation by fast riemannian
  optimization.
\newblock {\em Journal of the American Statistical Association},
  118(544):2588--2604.

\bibitem[Chen et~al., 2022]{chen2022factor}
Chen, R., Yang, D., and Zhang, C.-H. (2022).
\newblock Factor models for high-dimensional tensor time series.
\newblock {\em Journal of the American Statistical Association},
  117(537):94--116.

\bibitem[Clarke and Van~Gorder, 2003]{clarke2003improving}
Clarke, A.~J. and Van~Gorder, S. (2003).
\newblock Improving el niño prediction using a space-time integration of
  indo-pacific winds and equatorial pacific upper ocean heat content.
\newblock {\em Geophysical Research Letters}, 30(7).

\bibitem[Davis et~al., 2016]{davis2016sparse}
Davis, R.~A., Zang, P., and Zheng, T. (2016).
\newblock Sparse vector autoregressive modeling.
\newblock {\em Journal of Computational and Graphical Statistics},
  25(4):1077--1096.

\bibitem[de~Leeuw et~al., 1976]{de1976additive}
de~Leeuw, J., Young, F.~W., and Takane, Y. (1976).
\newblock Additive structure in qualitative data: An alternating least squares
  method with optimal scaling features.
\newblock {\em Psychometrika}, 41(4):471--503.

\bibitem[Dirksen, 2015]{dirksen2015tail}
Dirksen, S. (2015).
\newblock {Tail bounds via generic chaining}.
\newblock {\em Electronic Journal of Probability}, 20:1 -- 29.

\bibitem[Fan and Li, 2001]{fan2001scad}
Fan, J. and Li, R. (2001).
\newblock Variable selection via nonconcave penalized likelihood and its oracle
  properties.
\newblock {\em Journal of the American Statistical Association},
  96(456):1348--1360.

\bibitem[Feng et~al., 2021]{feng2021brain}
Feng, L., Bi, X., and Zhang, H. (2021).
\newblock Brain regions identified as being associated with verbal reasoning
  through the use of imaging regression via internal variation.
\newblock {\em Journal of the American Statistical Association},
  116(533):144--158.
\newblock PMID: 34955572.

\bibitem[Gahrooei et~al., 2021]{gahrooei2021multiple}
Gahrooei, M.~R., Yan, H., Paynabar, K., and Shi, J. (2021).
\newblock Multiple tensor-on-tensor regression: An approach for modeling
  processes with heterogeneous sources of data.
\newblock {\em Technometrics}, 63(2):147--159.

\bibitem[Gozolchiani et~al., 2011]{gozo2011emergence}
Gozolchiani, A., Havlin, S., and Yamasaki, K. (2011).
\newblock Emergence of el ni\~no as an autonomous component in the climate
  network.
\newblock {\em Phys. Rev. Lett.}, 107:148501.

\bibitem[Guhaniyogi et~al., 2017]{guhaniyogi2017bayesian}
Guhaniyogi, R., Qamar, S., and Dunson, D.~B. (2017).
\newblock Bayesian tensor regression.
\newblock {\em Journal of Machine Learning Research}, 18(79):1--31.

\bibitem[Guo et~al., 2012]{guo2012tensor}
Guo, W., Kotsia, I., and Patras, I. (2012).
\newblock Tensor learning for regression.
\newblock {\em IEEE Transactions on Image Processing}, 21(2):816--827.

\bibitem[Han et~al., 2022]{han2022an}
Han, R., Willett, R., and Zhang, A.~R. (2022).
\newblock {An optimal statistical and computational framework for generalized
  tensor estimation}.
\newblock {\em The Annals of Statistics}, 50(1):1 -- 29.

\bibitem[Han et~al., 2024]{han2024cp}
Han, Y., Yang, D., Zhang, C.-H., and Chen, R. (2024).
\newblock Cp factor model for dynamic tensors.

\bibitem[Hao et~al., 2021]{hao2021sparse}
Hao, B., Wang, B., Wang, P., Zhang, J., Yang, J., and Sun, W.~W. (2021).
\newblock Sparse tensor additive regression.
\newblock {\em Journal of Machine Learning Research}, 22(64):1--43.

\bibitem[Harshman, 1970]{harshman1970foundation}
Harshman, R.~A. (1970).
\newblock {F}oundations of the {P}{A}{R}{A}{F}{A}{C} procedure: {M}odels and
  conditions for an "explanatory" multi-modal factor analysis.
\newblock {\em UCLA Working Papers in Phonetics}, 16:1--84.

\bibitem[Hastie et~al., 2015]{hastie2015statistical}
Hastie, T., Tibshirani, R., and Wainwright, M. (2015).
\newblock {\em Statistical Learning with Sparsity: The Lasso and
  Generalizations}.
\newblock Chapman \& Hall/CRC.

\bibitem[Huang et~al., 2025]{huang2023supervised}
Huang, F., Lu, K., Zheng, Y., and Li, G. (2025).
\newblock Supervised factor modeling for high-dimensional linear time series.

\bibitem[Kolda and Bader, 2009]{kolda2009tensor}
Kolda, T.~G. and Bader, B.~W. (2009).
\newblock Tensor decompositions and applications.
\newblock {\em SIAM Review}, 51(3):455--500.

\bibitem[Li and Zhang, 2021]{li2021tensor}
Li, C. and Zhang, H. (2021).
\newblock {Tensor quantile regression with application to association between
  neuroimages and human intelligence}.
\newblock {\em The Annals of Applied Statistics}, 15(3):1455 -- 1477.

\bibitem[Li and Zhang, 2017]{li2017parsimonious}
Li, L. and Zhang, X. (2017).
\newblock {Parsimonious Tensor Response Regression}.
\newblock {\em Journal of the American Statistical Association},
  112(519):1131--1146.

\bibitem[Li et~al., 2018]{li2018tucker}
Li, X., Xu, D., Zhou, H., and Li, L. (2018).
\newblock Tucker tensor regression and neuroimaging analysis.
\newblock {\em Statistics in biosciences}, 10(3):520--545.

\bibitem[Li and Xiao, 2021]{li2021multi}
Li, Z. and Xiao, H. (2021).
\newblock Multi-linear tensor autoregressive models.

\bibitem[Llosa-Vite and Maitra, 2023]{llosa2023reduced}
Llosa-Vite, C. and Maitra, R. (2023).
\newblock Reduced-rank tensor-on-tensor regression and tensor-variate analysis
  of variance.
\newblock {\em IEEE Transactions on Pattern Analysis and Machine Intelligence},
  45(2):2282--2296.

\bibitem[Lock, 2018]{lock2018tensor}
Lock, E.~F. (2018).
\newblock Tensor-on-tensor regression.
\newblock {\em Journal of Computational and Graphical Statistics},
  27(3):638--647.
\newblock PMID: 30337798.

\bibitem[Lozano et~al., 2009]{lozano2009grouped}
Lozano, A.~C., Abe, N., Liu, Y., and Rosset, S. (2009).
\newblock Grouped graphical granger modeling for gene expression regulatory
  networks discovery.
\newblock {\em Bioinformatics}, 25(12):i110--i118.

\bibitem[Ludescher et~al., 2013]{josef2013improved}
Ludescher, J., Gozolchiani, A., Bogachev, M.~I., Bunde, A., Havlin, S., and
  Schellnhuber, H.~J. (2013).
\newblock Improved el niño forecasting by cooperativity detection.
\newblock {\em Proceedings of the National Academy of Sciences},
  110(29):11742--11745.

\bibitem[McPhaden, 2003]{mc2003tropical}
McPhaden, M.~J. (2003).
\newblock Tropical pacific ocean heat content variations and enso persistence
  barriers.
\newblock {\em Geophysical Research Letters}, 30(9).

\bibitem[Melnyk and Banerjee, 2016]{melnyk2016estimating}
Melnyk, I. and Banerjee, A. (2016).
\newblock Estimating structured vector autoregressive models.
\newblock In {\em Proceedings of the 33rd International Conference on
  International Conference on Machine Learning - Volume 48}, ICML'16, page
  830–839. JMLR.org.

\bibitem[Oseledets, 2011]{oseledets2011tensor}
Oseledets, I.~V. (2011).
\newblock Tensor-train decomposition.
\newblock {\em SIAM Journal on Scientific Computing}, 33(5):2295--2317.

\bibitem[Raskutti et~al., 2019]{raskutti2019convex}
Raskutti, G., Yuan, M., and Chen, H. (2019).
\newblock Convex regularization for high-dimensional multiresponse tensor
  regression.
\newblock {\em The Annals of Statistics}, 47(3):pp. 1554--1584.

\bibitem[Sarachik and Cane, 2010]{sarachik2010the}
Sarachik, E.~S. and Cane, M.~A. (2010).
\newblock {\em The El Niño-Southern Oscillation Phenomenon}.
\newblock Cambridge University Press.

\bibitem[Si et~al., 2024]{si2024efficient}
Si, Y., Zhang, Y., Cai, Y., Liu, C., and Li, G. (2024).
\newblock An efficient tensor regression for high-dimensional data.

\bibitem[Stock and Watson, 2011]{stock2011dynamic}
Stock, J. and Watson, M. (2011).
\newblock {\em Dynamic Factor Models}.
\newblock Oxford University Press, Oxford.

\bibitem[Tibshirani, 1996]{tibshirani1996lasso}
Tibshirani, R. (1996).
\newblock Regression shrinkage and selection via the lasso.
\newblock {\em Journal of the Royal Statistical Society. Series B
  (Methodological)}, 58(1):267--288.

\bibitem[Tucker, 1966]{tucker1966some}
Tucker, L. (1966).
\newblock Some mathematical notes on three-mode factor analysis.
\newblock {\em Psychometrika}, 31(3):279--311.

\bibitem[Vershynin, 2018]{vershynin2019high}
Vershynin, R. (2018).
\newblock {\em High-Dimensional Probability: An Introduction with Applications
  in Data Science}.
\newblock Cambridge Series in Statistical and Probabilistic Mathematics.
  Cambridge University Press.

\bibitem[Wainwright, 2019]{wainwright2019high}
Wainwright, M.~J. (2019).
\newblock {\em High-Dimensional Statistics: A Non-Asymptotic Viewpoint}.
\newblock Cambridge Series in Statistical and Probabilistic Mathematics.
  Cambridge University Press.

\bibitem[Wang et~al., 2024]{wang2024high}
Wang, D., Zheng, Y., and Li, G. (2024).
\newblock High-dimensional low-rank tensor autoregressive time series modeling.
\newblock {\em Journal of Econometrics}, 238(1):105544.

\bibitem[Wang et~al., 2022]{wang2022high}
Wang, D., Zheng, Y., Lian, H., and Li, G. (2022).
\newblock High-dimensional vector autoregressive time series modeling via
  tensor decomposition.
\newblock {\em Journal of the American Statistical Association},
  117(539):1338--1356.

\bibitem[Wilms et~al., 2023]{Wilms2023}
Wilms, I., Basu, S., Bien, J., and Matteson, D.~S. (2023).
\newblock Sparse identification and estimation of large-scale vector
  autoregressive moving averages.
\newblock {\em Journal of the American Statistical Association}, 118:571--582.

\bibitem[Xia et~al., 2022]{xia2022inference}
Xia, D., Zhang, A.~R., and Zhou, Y. (2022).
\newblock {Inference for low-rank tensors—no need to debias}.
\newblock {\em The Annals of Statistics}, 50(2):1220 -- 1245.

\bibitem[Xu et~al., 2018]{xu2018muscat}
Xu, J., Liu, X., Wilson, T., Tan, P.-N., Hatami, P., and Luo, L. (2018).
\newblock Muscat: Multi-scale spatio-temporal learning with application to
  climate modeling.
\newblock In {\em IJCAI}, pages 2912--2918.

\bibitem[Yu and Liu, 2016]{yu2016learn}
Yu, R. and Liu, Y. (2016).
\newblock Learning from multiway data: Simple and efficient tensor regression.
\newblock In Balcan, M.~F. and Weinberger, K.~Q., editors, {\em Proceedings of
  The 33rd International Conference on Machine Learning}, volume~48 of {\em
  Proceedings of Machine Learning Research}, pages 373--381, New York, New
  York, USA. PMLR.

\bibitem[Zhang et~al., 2019]{zhang2019tensor}
Zhang, X., Li, L., Zhou, H., Zhou, Y., and Shen, D. (2019).
\newblock Tensor generalized estimating equations for longitudinal imaging
  analysis.
\newblock {\em Statistica Sinica}, 29(4):1977--2005.

\bibitem[Zhao et~al., 2016]{zhao2016tensorring}
Zhao, Q., Zhou, G., Xie, S., Zhang, L., and Cichocki, A. (2016).
\newblock Tensor ring decomposition.
\newblock {\em CoRR}, abs/1606.05535.

\bibitem[Zheng, 2024]{Zheng2024}
Zheng, Y. (2024).
\newblock An interpretable and efficient infinite-order vector autoregressive
  model for high-dimensional time series.
\newblock {\em Journal of the American Statistical Association}, page To
  appear.

\bibitem[Zheng and Cheng, 2020]{zheng2020finite}
Zheng, Y. and Cheng, G. (2020).
\newblock {Finite-time analysis of vector autoregressive models under linear
  restrictions}.
\newblock {\em Biometrika}, 108(2):469--489.

\bibitem[Zhou et~al., 2013]{zhou2013tensor}
Zhou, H., Li, L., and Zhu, H. (2013).
\newblock Tensor regression with applications in neuroimaging data analysis.
\newblock {\em Journal of the American Statistical Association},
  108(502):540--552.

\bibitem[Zhou and Zhang, 2023]{zhou2023a}
Zhou, L. and Zhang, R.-H. (2023).
\newblock A self-attention–based neural network for three-dimensional
  multivariate modeling and its skillful enso predictions.
\newblock {\em Science Advances}, 9(10):eadf2827.

\end{thebibliography}

\end{document}